\documentclass[aps,prx,twocolumn,eqsecnum,superscriptaddress,10pt]{revtex4-2}

\usepackage{algorithm,algorithmic,amsmath}     

\usepackage{amssymb,graphicx,tikz,xcolor,colortbl,booktabs,multirow,siunitx,adjustbox}

\newtheorem{theorem}{Theorem}

\usepackage[justification=justified,singlelinecheck=false]{subcaption}
\usepackage[utf8]{inputenc}
\usepackage[T1]{fontenc}
\usepackage{booktabs,makecell,graphicx}    
\usepackage[abs]{overpic}
\usepackage[T1]{fontenc}
\usepackage[utf8]{inputenc}

\usepackage{graphicx,dcolumn,bm,physics,hyperref}
\usepackage[mathlines]{lineno}

\usepackage{ragged2e} 
\usepackage{caption}
\DeclareCaptionJustification{justified}{\justifying} 
\def\bea{\begin{eqnarray}}
\def\eea{\end{eqnarray}}

\usepackage[dvipsnames]{xcolor}

\begin{document}

\title{Landscape-Similarity-Guided Optimization in Divide-and-Conquer QAOA}

\author{Sokea Sang}
\affiliation{Department of AI Convergence, Pukyong National University, Busan 48513, Republic of Korea}
\author{Leanghok Hour}
\affiliation{AI Research Institute, Pukyong National University, Busan 48513, Republic of Korea}
\author{Sanghyeon Lee}
\affiliation{Department of AI Convergence, Pukyong National University, Busan 48513, Republic of Korea}
\author{Aniket Patra}
\affiliation{Department of Physics, Pukyong National University, Busan 48513, Republic of Korea}
\affiliation{Department of Physics, Indian Institute of Technology Kharagpur, Kharagpur 721302, West Bengal, India}
\author{Hee Chul Park}
\affiliation{Department of Physics, Pukyong National University, Busan 48513, Republic of Korea}
\author{Moon Jip Park}
\affiliation{Department of Physics, Hanyang University, Seoul 04763, Republic of Korea}
%
\author{Youngsun Han}\email{youngsun@pknu.ac.kr}
\affiliation{Department of AI Convergence, Pukyong National University, Busan 48513, Republic of Korea}

\date{\today}

\begin{abstract}

Divide-and-conquer strategies mitigate hardware constraints for the Quantum Approximate Optimization Algorithm (QAOA) on Noisy Intermediate-Scale Quantum (NISQ) devices by partitioning large interaction graphs into smaller, hardware-compatible sub-problems. However, this approach introduces a severe classical training bottleneck: a decomposition across $m$ boundary nodes generates $2^m$ distinct sub-problems that typically require independent optimization. In this work, we quantify correlations between reduced QAOA energy landscapes using a landscape-overlap measure $q$ and Pearson correlation $r$, and investigate their usefulness for parameter transfer. For $p=1$ on the studied long-range percolation graphs with $L=50$--$400$, the displayed $q(s)$ curves cross near $s\approx0.6$. This observation on finite systems does not establish a thermodynamic phase transition. 
Exploiting these correlations, we introduce Doubly Optimized
QAOA (DO-QAOA), which fully optimizes one representative
reduced instance and transfers its parameters to the remaining
frozen assignments, with target-specific refinement when
required.
We evaluated DO-QAOA on more than 6,000 circuits under the realistic noise model, DO-QAOA uses approximately $285$--$386\times$ fewer quantum shots and has approximately $6.5$--$23.9\times$ shorter runtime than FrozenQubits, with equal or lower dataset-averaged Approximation Ratio Gap (ARG). 
\end{abstract}

\maketitle

\section{Introduction}\label{sec:introduction}

\begin{table*}[ht]
\centering
\caption{Comparison of DO-QAOA with state-of-the-art variational quantum algorithms and QAOA-based methods. Our method (DO-QAOA) retains the circuit reduction benefits of the Divide-and-Conquer (D\&C) strategy. It fully trains one representative,
$K_{\mathrm{rep}}=1$, and transfers its parameters to the
remaining frozen assignments. 
Target-specific refinement and
final branch evaluation incur additional costs that are not
included in this full-training component. Grouping instances
into multiple landscape classes $K$ is a possible extension
discussed in Section~\ref{sec:discussion}, rather than an
implemented step of the present method.}

\label{tab:overview_sota}
\resizebox{\textwidth}{!}{%
\begin{tabular}{lcccccc}
\toprule
\toprule
Method & Strategy & \makecell{Training\\Complexity} & \makecell{Circuit\\Depth} & \makecell{Parameter\\Space} & \makecell{Noise\\Resilience} \\
\midrule
Baseline QAOA & Global Optimization & $\mathcal{O}(N_{\text{shots}} \cdot N_{\text{iter}})$ & High & $2p$ & Low \\
HEA \cite{hea_qaoa} (Nature’17) & Hardware Efficient & $\mathcal{O}(N_{\text{shots}} \cdot N_{\text{iter}})$ & Medium & $3np$ & Medium \\
Red-QAOA \cite{red_qaoa} (ASPLOS’24) & Graph Reduction & $\mathcal{O}(N_{\text{shots}} \cdot N_{\text{iter}})$ & Reduced & $2p$ & Medium \\
Frozen Qubits \cite{FrozenQubits_Ayanzadeh2023} (ASPLOS’23) & D \& C & $\mathcal{O}(2^m \cdot N_{\text{shots}} \cdot N_{\text{iter}})$ & Low & $2p$ & High \\
\midrule
DO-QAOA (Ours) & \makecell{D \& C + \\ Transfer Learning} & $\mathcal{O}(K_{\mathrm{rep}} \cdot N_{\text{shots}} \cdot N_{\text{iter}})$ & Low & $2p$ & High \\
\bottomrule
\bottomrule
\end{tabular}%
}
\end{table*}

The Quantum Approximate Optimization Algorithm (QAOA) is a leading candidate for demonstrating near-term quantum advantage~\cite{farhi2014quantumapproximateoptimizationalgorithm,qaoa_fermionic_view_Wang2018, biology_qaoa_advantage}. However, on Noisy Intermediate-Scale Quantum (NISQ) hardware~\cite{NISQquantumcomputingin,nisq_hardware_contrain,not_all_swap_gate,sabre}, the deep circuits required for standard QAOA lead to severe error accumulation, drastically degrading performance~\cite{FrozenQubits_Ayanzadeh2023, dc-qaoa}. To mitigate these hardware constraints, divide-and-conquer strategies have emerged as a powerful paradigm, enabling the decomposition of large, highly connected graph problems into smaller, hardware-compatible sub-instances~\cite{FrozenQubits_Ayanzadeh2023, dc-qaoa, frozenLGP}. A prominent example is the FrozenQubits technique~\cite{FrozenQubits_Ayanzadeh2023}, which partitions the interaction graph by fixing a small subset of high-degree nodes to classical states. Such structural decomposition approaches significantly reduce the required circuit depth and quantum gate overhead, allowing complex optimization tasks to be mapped onto contemporary noisy processors.

Despite its circuit-level benefits, the standard divide-and-conquer approach suffers from a fatal bottleneck. A partition with $m$ boundary cuts typically generates $2^{m}$ distinct reduced sub-problems, each with varying induced linear biases. Existing methodologies treat each instance as an entirely independent optimization task. This exponential proliferation of training sessions quickly becomes intractable even for modest $m$, negating the throughput benefits of graph partitioning and shifting the barrier from quantum fidelity to classical resource exhaustion.

Overcoming this exponential overhead requires a fundamental understanding of how local interventions such as variable freezing reshape the global variational landscape of QAOA. While prior works have extensively characterized the geometric structure of QAOA energy landscapes on parent graphs~\cite{qaoa_bounds_lotshaw2021empirical, transferability_optimal_qaoa_parameters_random_graph, ParameterTransfer_Shaydulin2023, fixed_angle_PhysRevA}, these analyses are restricted to single-instance formulations and do not address the structural consequences of graph decimation.

In particular, the behavior, correlations, and topological organization of the $2^m$ sub-problem landscapes induced by freezing $m$ qubits have, to our knowledge, not been systematically investigated. This represents a critical gap, as divide-and-conquer strategies inherently operate in this exponentially expanded instance space, where each sub-problem is typically assumed to possess an independent variational landscape.

This regime is precisely where standard divide-and-conquer approaches encounter a fundamental bottleneck. Under the assumption of landscape independence, generating $2^m$ sub-problems necessitates $2^m$ separate optimization procedures, resulting in an exponential measurement overhead. 
We investigate correlations among the decimated sub-problem landscapes and assess whether parameters trained on one representative can be reused across the remaining assignments.

Building on this insight, we leverage the emergent landscape structure to bypass redundant optimization. Rather than treating each sub-problem independently, we solve a single representative instance and transfer the optimal parameters to the remaining sub-problems, with refinement when required, thereby reducing repeated optimization and yielding substantial measured savings in runtime and quantum shots while maintaining a competitive Approximation Ratio Gap (ARG).

Universality in complex systems captures the idea that macroscopic structure is governed by a small set of coarse variables and is largely insensitive to microscopic details~\cite{spin_glasses_order_parameters, optimization_by_simulated_annealing, basic_notions_condensed_matter}. Recently, this viewpoint has expanded from conventional critical phenomena to the geometry of energy landscapes in optimization and machine learning, sharpening interest in how local operations or constraints can reshape global landscape structures~\cite{identifying_attacking_saddle_point, barren_plateaus_in_quantum_neural_network, theory_overparametrization_quantum_neural_network}.
Energy landscapes generated by shallow variational quantum circuits likewise often exhibit pronounced basin structures~\cite{qaoa_bounds_lotshaw2021empirical, transferability_optimal_qaoa_parameters_random_graph, ParameterTransfer_Shaydulin2023, fixed_angle_PhysRevA} and characteristic curvature patterns, motivating a natural question for divide-and-conquer QAOA: Do the $2^m$ sub-problems generated by graph partitioning share enough topological similarity to allow for universal parameter transfer?

This question becomes crucial for the divide-and-conquer QAOA on NISQ devices.
Specifically, if local constraints such as freezing qubits in divide-and-conquer strategies do not fundamentally distort the optimization landscape, one can bypass the intractable overhead of independent retraining by transferring optimal parameters across sub-problems.
\begin{figure}[t]
  \includegraphics[width=0.46\textwidth]{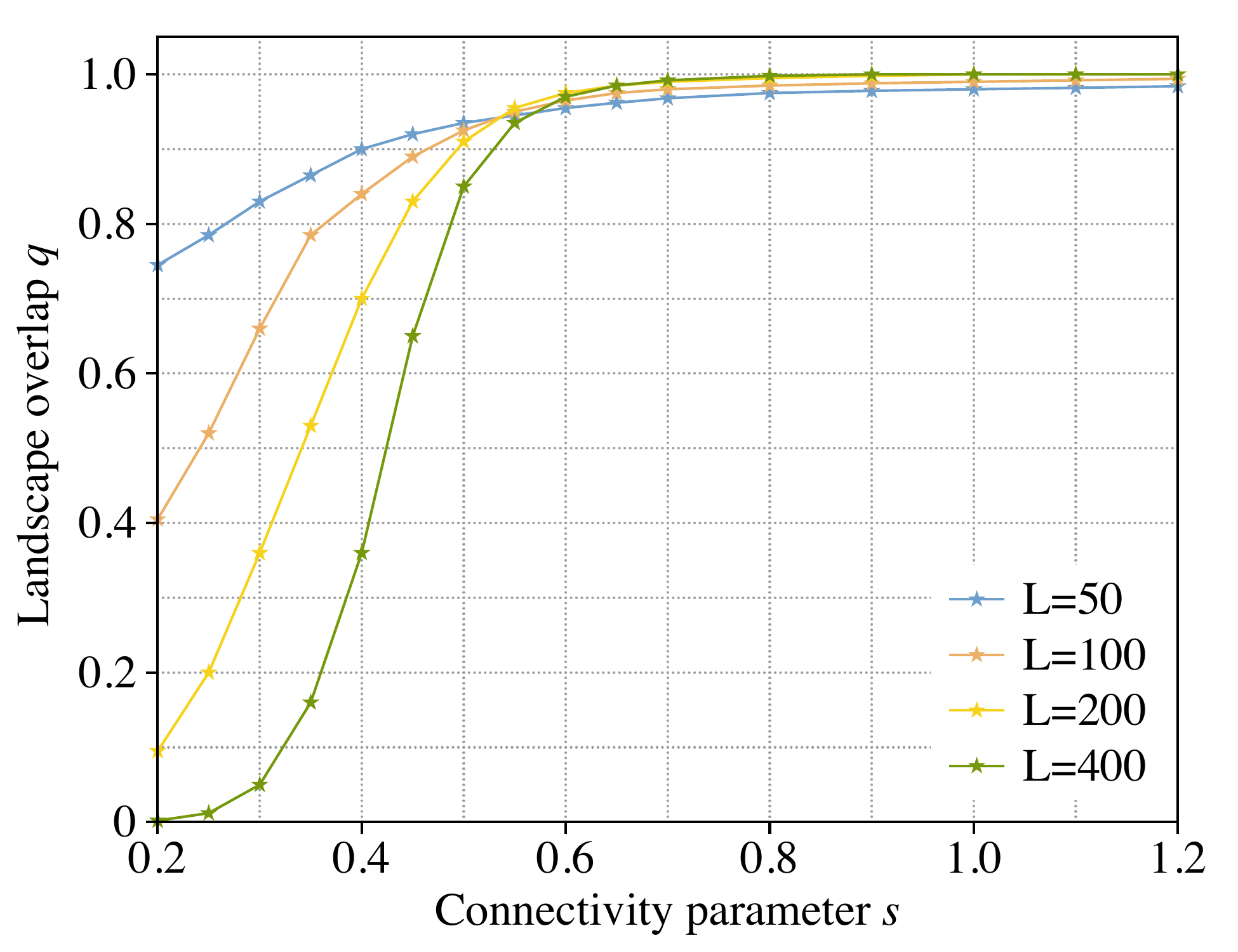}
  \caption{Landscape-overlap measure $q$ versus decay exponent $s$ for $p=1$ and $L=50,100,200,400$. The displayed curves cross in a region near $s\approx0.6$ and show different size dependence at smaller and larger $s$. This is an observation on finite systems; the present data do not determine its nature in the thermodynamic limit.}
  \label{fig:percolation_phase_boundary_1}
\end{figure}
Here, we investigate landscape similarity in divide-and-conquer QAOA, motivated by universality in complex disordered systems~\cite{ universality_in_spin_glass_model_CARMONA2006215, basic_notions_condensed_matter}. 
Freezing a small set of core nodes often perturbs the instance microscopically, yet leaves the macroscopic geometry of the variational energy landscape, such as the hierarchy of basins and barriers, largely intact,
and the resulting landscapes can remain correlated across frozen assignments.
We quantify this similarity in terms of correlations in the energy landscape over parameter space and show that it can remain high across many frozen assignments.

Motivated by the observed landscape correlations, we introduce
Doubly Optimized QAOA (DO-QAOA), a divide-and-conquer strategy
that reduces repeated full optimization across frozen
assignments. The implemented procedure fully optimizes one
representative sub-problem and transfers its parameters to
the remaining instances. Thus, $K_{\mathrm{rep}}=1$ denotes
an algorithmic choice, rather than an inferred number of
landscape classes.
However, freezing qubits transforms their interactions into induced local magnetic fields (linear biases) on the active subgraph . Since these biases vary across frozen configurations, they act as perturbations that can distort the energy landscape. To ensure robustness, we introduce a physics-informed Bias-Aware Transfer Rule that quantifies this distortion and selectively triggers lightweight warm-start fine-tuning only when the landscape shifts significantly.

The QAOA circuit light cone provides a qualitative motivation for studying the dependence of landscape similarity on connectivity~\cite{fixed_param_qaoa_brandao2018, bounds_generation_correllations_lieb_robinson, farhi2020quantum}. Freezing a small core perturbs the instance locally, but the induced information can only spread over a finite range set by the QAOA depth. The relevant competition is between the circuit light-cone~\cite{finite_group_velocity_quantum_spin_system, qaoa_fermionic_view_Wang2018}, which bounds the effective correlation length of a shallow ansatz, and the graph’s connectivity-controlled propagation, which determines whether perturbations percolate through the interaction network~\cite{bounds_generation_correllations_lieb_robinson, collective_dynamic_small_world}. 

In Fig.~\ref{fig:percolation_phase_boundary_1}, we examine the landscape-overlap measure $q$ as the decay exponent $s$ varies. For $p=1$ and $L=50,100,200,400$, the displayed curves cross in a region near $s\approx0.6$. At smaller $s$, $q$ generally decreases with $L$; at larger $s$, it increases toward unity. These observations describe the finite systems studied, rather than establishing a critical connectivity threshold or self-averaging in the thermodynamic limit.
The contour examples in Fig.~\ref{fig:pcontour_results}
show minima in similar regions of parameter space
for the plotted instances.

Many practically relevant graph classes, including power-law~\cite{agler2016microbial,clauset2016colorado,power_law_classification,house2015testing,measurement_analysis_online_social_network, RevModPhys.87.925} and real-world networks, contain only a small number of high-degree nodes~\cite{agler2016microbial, meghanathan2017complex, GAMERMANN2019122204, power_law_classification, house2015testing, measurement_analysis_online_social_network, RevModPhys.87.925}. Consequently, freezing a small number of qubits suffices to capture most circuit-depth reduction benefits, while freezing additional nodes yields diminishing returns (we discuss in more detail in Section~\ref{sec:discussion} and Appendix~\ref{sec:appendix_limits}).

We evaluated DO-QAOA on more than 6{,}000 circuits spanning synthetic power-law and regular graphs~\cite{barabasi1999emergence, regular_sk_graphs, harrigan2021, solve_spin_glass} as well as real-world networks, including AIDS, Linux, and IMDb datasets~\cite{aids_riesen2008iam, linux_graph, imdb}. 
At $m=3$, the dataset-averaged ARG values match or improve upon those of FrozenQubits in Table~\ref{tab:additional_comparison}. For the same $m=3$ comparisons at $p=1$, on benchmark graphs with $N\leq20$ under the \texttt{FakeBrisbane} noise model, the reported total shot counts decrease from $65.54\times10^6$ to $0.17$--$0.23\times10^6$, corresponding to approximately $285$--$386\times$ fewer shots. The reported runtimes are approximately $6.5$--$23.9\times$ shorter. These are finite-benchmark resource ratios, not asymptotic speedups with system size.
%
We examine landscape correlations of $2^m$ decimated sub-problems and evaluate a single-representative parameter-transfer strategy across the frozen assignments, without inferring the number of landscape classes.
Finally, we introduce an end-to-end pipeline that integrates graph partitioning with representative parameter transfer, compared to standard divide-and-conquer methods as summarized in Table~\ref{tab:overview_sota}.

\begin{figure*}[t]
    \centering
    
    \begin{subfigure}[b]{0.3\textwidth}
        \centering
        \begin{overpic}[width=\textwidth]{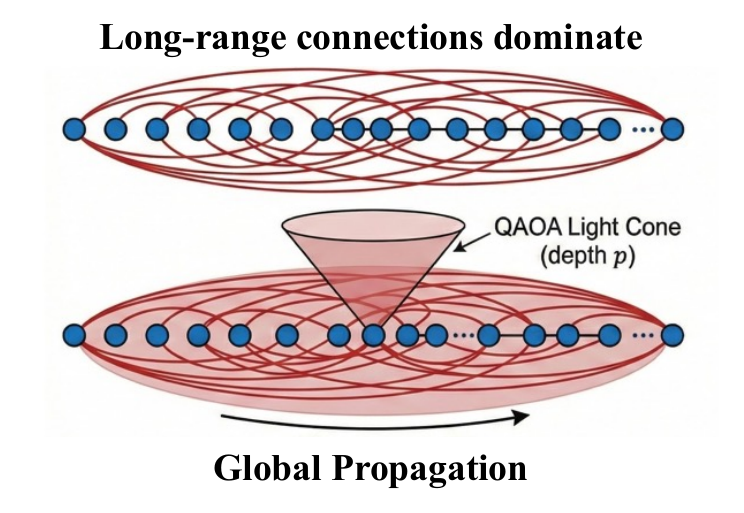}
            \put(5, 130){\textbf{(a)}}
        \end{overpic}
        \vspace{-0.8em} 
        \label{fig:low_overlab_a}
    \end{subfigure}
    \hspace{0.3em} 
    \begin{subfigure}[b]{0.3\textwidth}
        \centering
        \begin{overpic}[width=\textwidth]{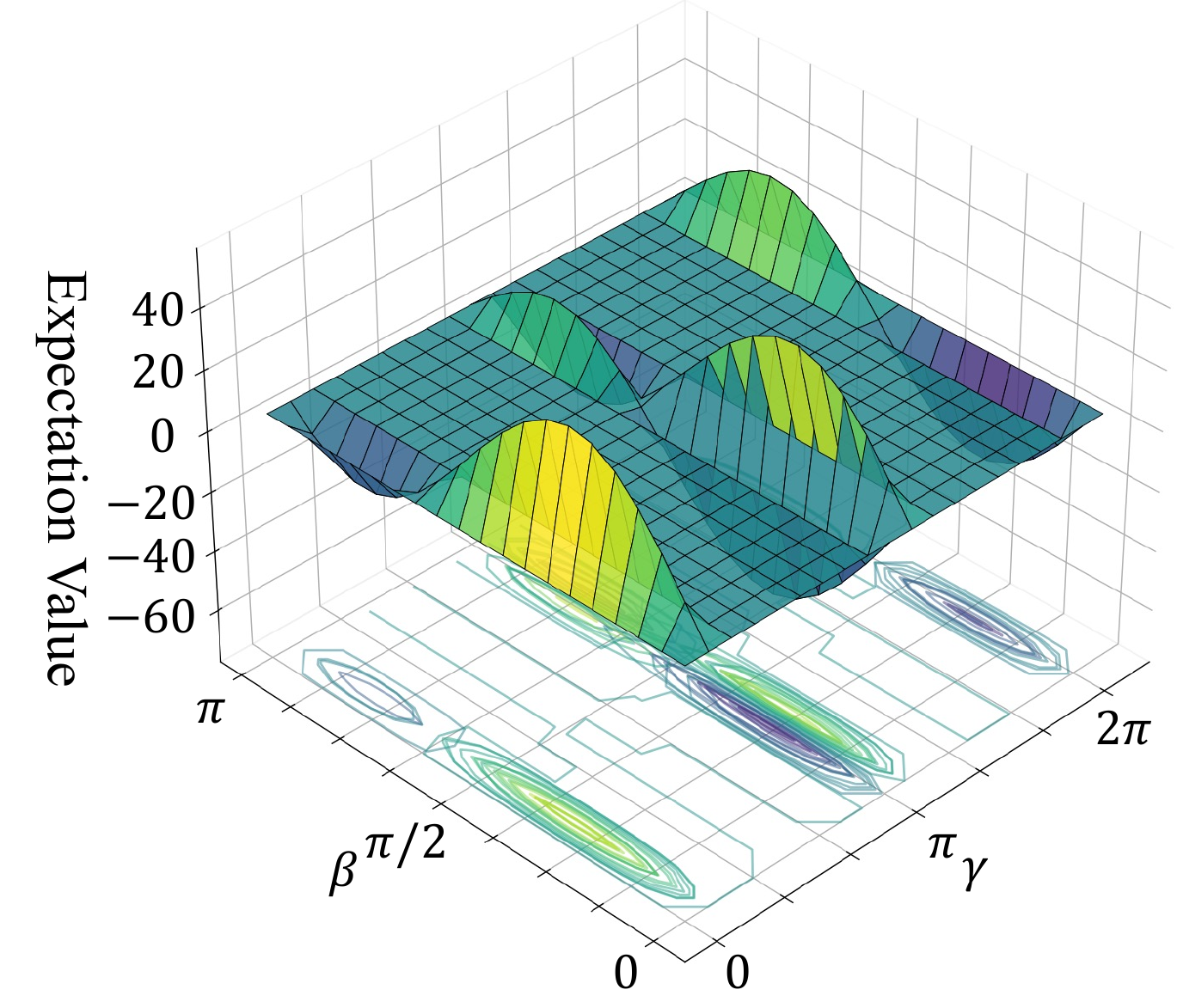}
            \put(5, 130){\textbf{(b)}}
        \end{overpic}
        \vspace{-0.8em}
        \label{fig:low_overlab_b}
    \end{subfigure}
    \hspace{0.3em}
    \begin{subfigure}[b]{0.3\textwidth}
        \centering
        \begin{overpic}[width=\textwidth]{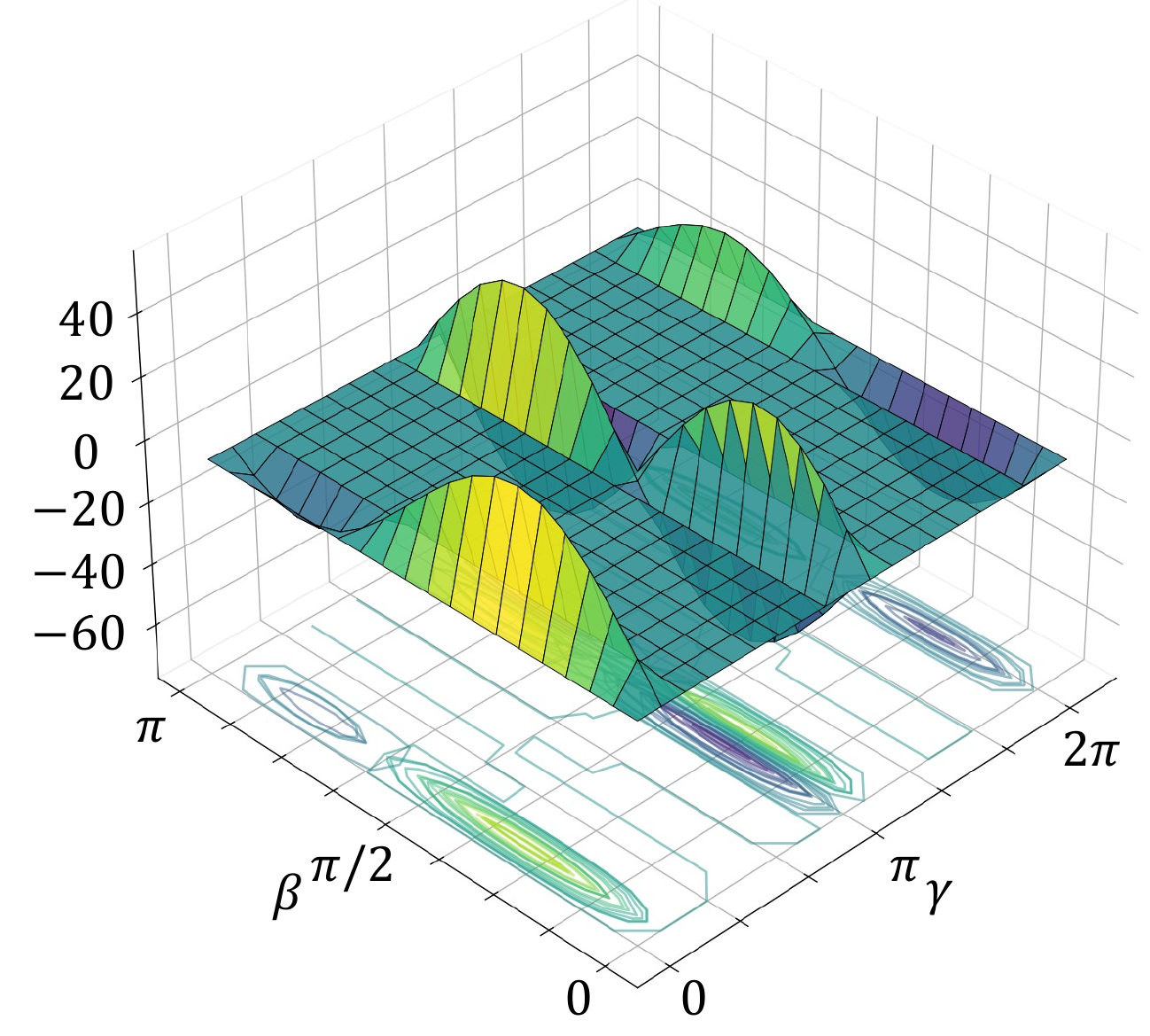}
            \put(5, 130){\textbf{(c)}}
        \end{overpic}
        \vspace{-0.8em}
        \label{fig:low_overlab_c}
    \end{subfigure}

    \begin{subfigure}[b]{0.3\textwidth}
        \centering
        \begin{overpic}[width=\textwidth]{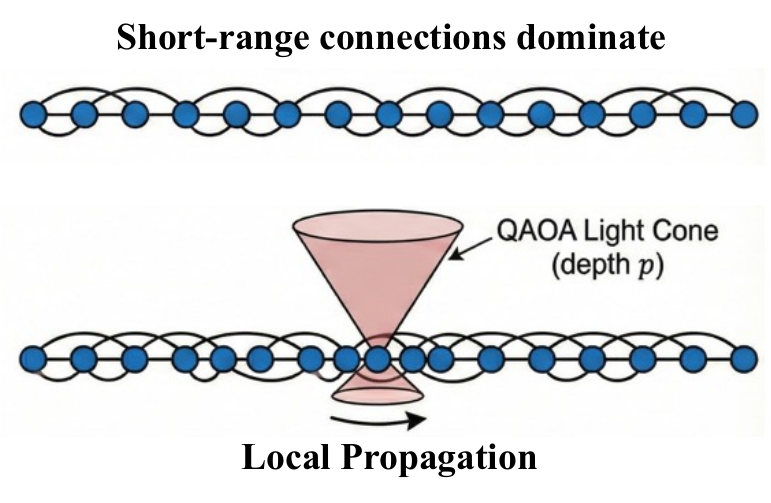}
            \put(5, 130){\textbf{(d)}}
        \end{overpic}
        \vspace{-0.8em}
        \label{fig:high_overlab_a}
    \end{subfigure}
    \hspace{0.3em}
    \begin{subfigure}[b]{0.3\textwidth}
        \centering
        \begin{overpic}[width=\textwidth]{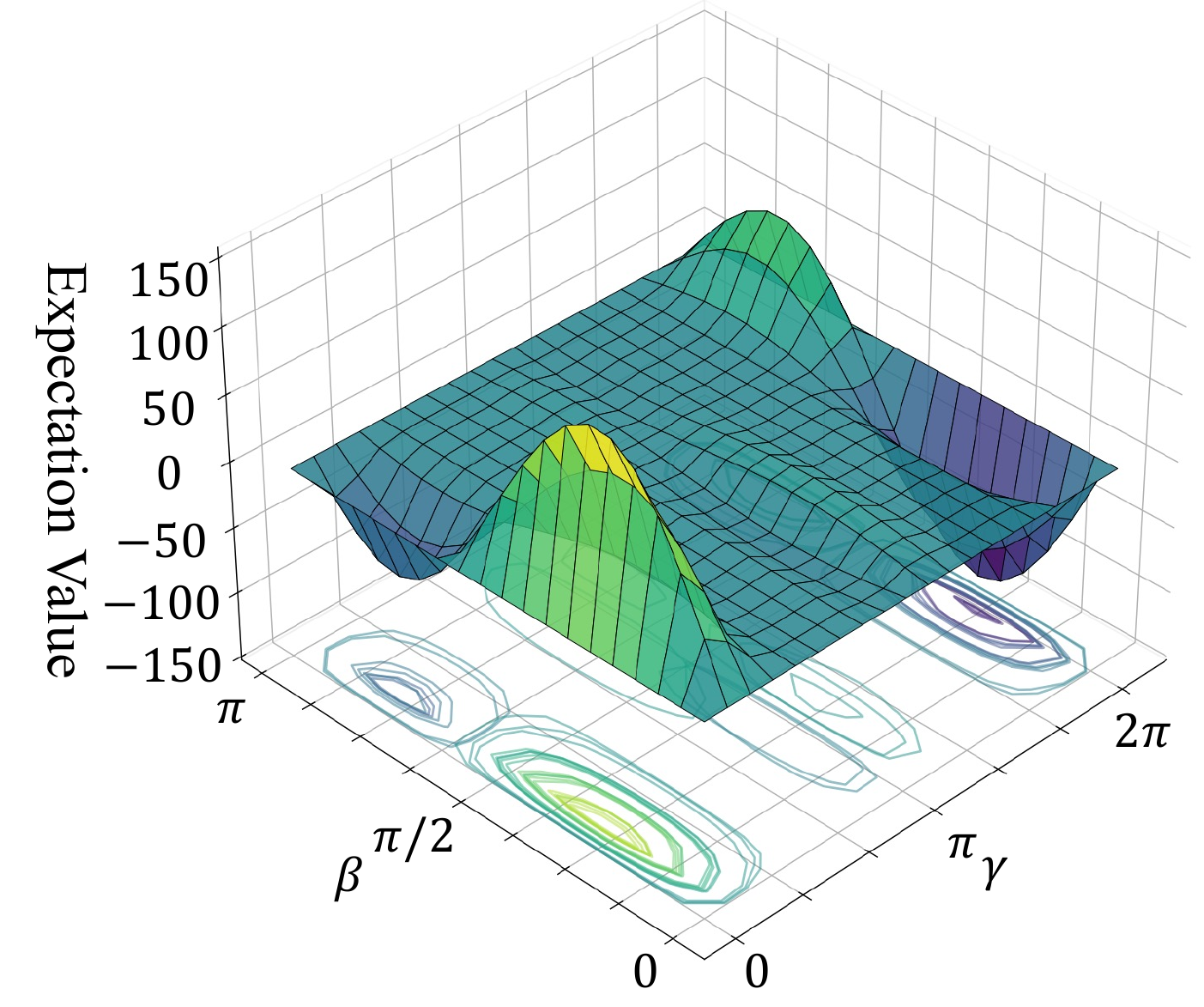}
            \put(5, 130){\textbf{(e)}}
        \end{overpic}
        \vspace{-0.8em}
        \label{fig:high_overlab_b}
    \end{subfigure}
    \hspace{0.3em}
    \begin{subfigure}[b]{0.3\textwidth}
        \centering
        \begin{overpic}[width=\textwidth]{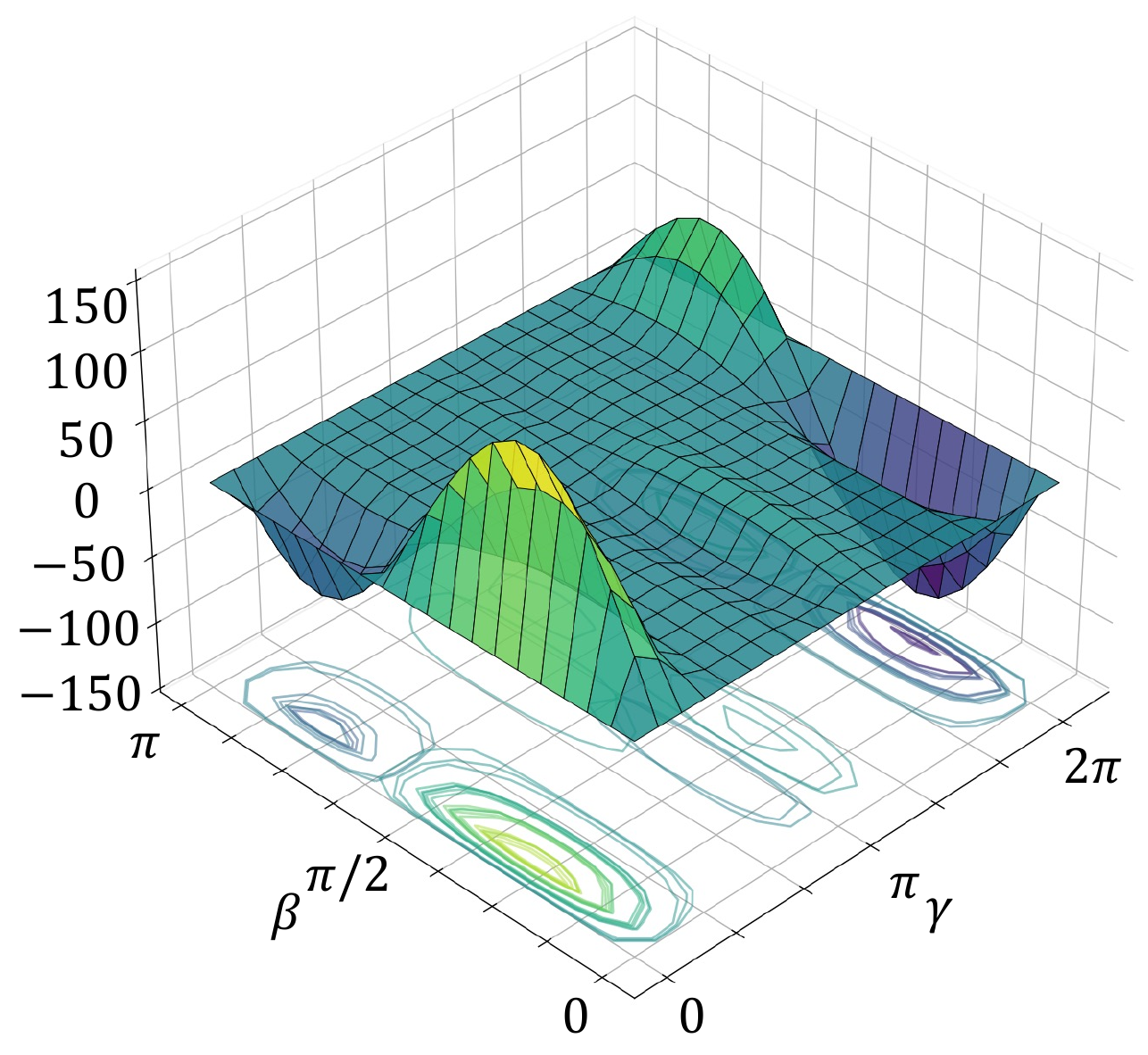}
            \put(5, 130){\textbf{(f)}}
        \end{overpic}
        \vspace{-0.8em}
        \label{fig:high_overlab_c}
    \end{subfigure}

    \vspace{0.2em}
  \caption{Examples of lower- and higher-overlap landscape regimes. Panels (a) and (d) illustrate the corresponding interaction structures. Panels (b--c) and (e--f) show energy landscapes for distinct frozen assignments in the lower- and higher-overlap examples, respectively. The vertical axis is the expectation value $E(\bm{\gamma},\bm{\beta})$ of the cost Hamiltonian. The lower-overlap examples show different landscape shapes with minima in similar parameter regions, while the higher-overlap examples show more closely aligned surfaces. The panels illustrate finite-instance similarities and differences; they do not establish distinct thermodynamic phases or universal basin locations.}
    \label{fig:percolation_phase_boundary_2}
\end{figure*}

\section{Landscape Similarity in QAOA}
\label{sec:motivation}

\subsection{Standard QAOA Ansatz}
\label{subsec:qaoa_decimation}

The QAOA is a variational hybrid algorithm designed to find approximate solutions to combinatorial optimization problems~\cite{quantum_wireless_scheduling, biology_qaoa_advantage, qualify_quantum_hard_industrial, multi_objective_Kotil2025}, with preliminaries provided in Appendix~\ref{sec:preliminaries}.  
The problem is encoded in a cost Hamiltonian $H_C$ acting on $N$ qubits. In this work, we focus on the general Ising model defined on an interaction graph $G=(V,E)$, given by:
\begin{equation}
H_C = \sum_{(i,j) \in E} J_{ij} Z_i Z_j + \sum_{i \in V} h_i Z_i,
\label{eq:cost_hamiltonian}
\end{equation}
where $Z_i$ denotes the Pauli-Z operator on qubit $i$, $J_{ij}$ represents the coupling strength between qubits $i$ and $j$, and $h_i$ represents the local field bias.
The algorithm prepares a parameterized ansatz state $|\psi_p({\gamma}, {\beta})\rangle$ by applying a sequence of alternating unitaries:
\begin{equation}
    |\psi_p({\gamma}, {\beta})\rangle = \prod_{k=1}^p e^{-i\beta_k H_B} e^{-i\gamma_k H_C} |+\rangle^{\otimes N},
    \label{eq:qaoa_ansatz}
\end{equation}
where $H_B = \sum_{j} X_j$ is the transverse-field mixer, and $p$ denotes the circuit depth. The classical optimizer iteratively updates the parameters $({\gamma}, {\beta})$ to minimize the expectation value $E({\gamma}, {\beta}) = \langle \psi_p | H_C | \psi_p \rangle$. 
The objective of QAOA is to find parameters $\gamma$ and $\beta$ that minimizes the expectation value $E({\gamma}, {\beta}) = \bra{\psi({\gamma}, {\beta})} H_C \ket{\psi({\gamma}, {\beta})}$. In the main text, we focus on depth $p=1$ for  visualization and the analysis that follows.

Divide-and-conquer strategies, such as \textit{FrozenQubits}~\cite{FrozenQubits_Ayanzadeh2023} (see Appendix~\ref{sub:divide_and_conquer_preliminary}), mitigate hardware noise by partitioning the graph. Freezing a subset of high-degree nodes $S$ into a configuration $\mathbf{z} \in \{+1, -1\}^{|S|}$ reduces the problem to an active subspace $R$. Crucially, this partitioning decomposes the effective Hamiltonian into two distinct parts:
\begin{equation}
    H^{(\mathbf{z})} = H_{\text{quad}}^{(R)} + H_{\text{lin}}^{(R, \mathbf{z})}.
    \label{eq:effective_hamiltonian}
\end{equation}
Here, $H_{quad}^{(R)} = \sum_{(i,j) \in E_R} Z_i Z_j$ represents the internal interactions of the active subgraph, where $E_R = \{(i,j) \in E \mid i,j \in R\}$ denotes the set of edges connecting qubits strictly within the active partition $R$. 
The variation between sub-problems is entirely contained within the linear bias term $H_{lin}^{(R,z)} = \sum_{j \in R} \tilde{h}_j^{(z)} Z_j$, where the local fields are shifted by the frozen values according to $\tilde{h}_j^{(z)} = h_j + \sum_{k \in \mathcal{N}(j) \cap S} J_{kj} z_k$. Here, $\mathcal{N}(j) = \{k \in V \mid (j,k) \in E\}$ denotes the set of immediate neighbors of node $j$ in the original graph. 
Importantly, this quadratic term $H_{quad}^{(R)}$ remains invariant across all possible frozen configurations.
Standard approaches treat these $2^m$ instances as uncorrelated optimization tasks, necessitating independent optimization loops. Consequently, the total runtime scales as $\mathcal{O}(2^m \cdot N_{\text{shots}} \cdot N_{\text{iter}})$, shifting the barrier from quantum fidelity to classical computational intractability. Even a modest partition of $m=10$ leads to 1024 separate training sessions, negating the throughput benefits of parallelization.

\subsection{Landscape-overlap measurements}\label{subsec:landscape_similarity_transition}

The variational energy $E({\gamma}, {\beta})$ generated by the quantum circuit is expressed as a weighted average of the energy distribution in the spin space as,
\begin{equation}\label{eq:energy_distribution}
    E({\gamma}, {\beta})= \sum_{z \in \{-1,1\}^N}|\langle z \mid \psi({\gamma}, {\beta}) \rangle|^2 H_C(z),
\end{equation}
where $H_C(z)=\langle z|H_C| z\rangle$ represents the spin-configuration landscape in the bitstring space. Physically, Eq.~\ref{eq:energy_distribution} can be interpreted as a convolution process, where the QAOA wavefunction acts as a probabilistic kernel that smooths the discrete spin landscape into a differentiable energy surface over $({\gamma}, {\beta})$. Tuning the variational parameters is equivalent to reshaping the probability distribution to shift the focus of the kernel toward specific low-energy regions of the spin space.

For a fixed parent graph and frozen set $S$, the index $k$
labels a frozen assignment $z^{(k)}$. We denote by
$H^{(z^{(k)})}$ its reduced Hamiltonian and by
$|\psi_{p,k}(\boldsymbol{\theta})\rangle$ the QAOA state
prepared using that reduced Hamiltonian in the cost
unitaries of Eq.~\eqref{eq:qaoa_ansatz}. Its variational
energy is
\[
E_k(\boldsymbol{\theta}) =
\langle\psi_{p,k}(\boldsymbol{\theta})|
H^{(z^{(k)})}|
\psi_{p,k}(\boldsymbol{\theta})\rangle.
\]
At $p=1$, $\boldsymbol{\theta}=(\gamma,\beta)$. The computation of this expectation depends on the experiment: Fig.~\ref{fig:percolation_phase_boundary_1} uses an analytic $p=1$ expression, whereas the higher-depth calculations in Appendix~\ref{app:higher_depths} use state-vector simulation. The overlap compares energy values at a common set of parameter points within each calculation. Its mean and standard deviation are taken over those points, not over graph realizations. Assignment-dependent constant offsets cancel under centering in Eq.~\ref{eq:S_kl_and_E_k}. They are retained when comparing global energies across frozen branches, since the constants can differ between assignments.

For the unit-coupling graphs in Fig.~\ref{fig:percolation_phase_boundary_1}, the analytic evaluation is
\[
\begin{aligned}
E_k(\gamma,\beta)={}&C_k+\sin(2\beta)\sum_{i\in R}h_i^{(k)}v_i c_\gamma^{d_i}\\
&+\sin(4\beta)\sum_{(i,j)\in E_R} A_{ij}\\
&+\sin^2(2\beta)\sum_{(i,j)\in E_R} B_{ij},
\end{aligned}
\]
where $E_R$ is the active edge set, $d_i$ is the active degree, $n_{ij}$ is the number of common active neighbors, $C_k=C^{(z^{(k)})}$, $c_\gamma=\cos(2\gamma)$, $u_i=\cos(2\gamma h_i^{(k)})$, and $v_i=\sin(2\gamma h_i^{(k)})$. The edge terms are
\[
\begin{aligned}
A_{ij}&=\tfrac12\sin(2\gamma)\left(u_i c_\gamma^{d_i-1}+u_j c_\gamma^{d_j-1}\right),\\
B_{ij}&=c_\gamma^{d_i+d_j-2-2n_{ij}}\left(w_{ij}^{+}v_i v_j+w_{ij}^{-}u_i u_j\right),\\
w_{ij}^{\pm}&=\tfrac12\left[1\pm\bigl(\cos(4\gamma)\bigr)^{n_{ij}}\right].
\end{aligned}
\]
The induced fields $h_i^{(k)}$ enter both the node and edge terms; they are not omitted.

We quantify landscape correlation using the overlap measure $q(s)$, defined as the mean pairwise overlap among $M$ frozen replicas $\mathcal Z=\{z^{(1)},\dots,z^{(M)}\}$ sampled on a graph drawn with connectivity parameter $s$,
\begin{equation}
q(s)=\frac{2}{M(M-1)}\sum_{k=1}^{M}\sum_{l=k+1}^{M} S_{kl}.
\end{equation}
Here $S_{kl}$ is the cosine similarity between standardized energy surfaces,
\begin{equation}\label{eq:S_kl_and_E_k}
\begin{aligned}
S_{kl}=\frac{\langle \mathrm{vec}(\hat E_k),\,\mathrm{vec}(\hat E_l)\rangle}
{\|\mathrm{vec}(\hat E_k)\|\,\|\mathrm{vec}(\hat E_l)\|}, \\
\quad
\hat E_k(\gamma,\beta)\!=\!\frac{E_k(\gamma,\beta)-\langle E_k\rangle}{\mathrm{std}(E_k)},
\end{aligned}
\end{equation}
where $\langle\cdot\rangle$ denotes an average over the parameter points used in the stated calculation. We evaluate each $E_k$ on the same $D$ points within that calculation, flatten the values into a vector $E_k\in\mathbb{R}^D$, and standardize it as $\hat E_k=(E_k-\mu_k)/\sigma_k$, with $\mu_k$ and $\sigma_k$ the mean and standard deviation over those points. The dependence on $s$ enters through the graph ensemble: nearest-neighbor edges are always present, while additional edges are sampled with probability $p(r_{ij})$ from Eq.~\ref{eq:lrp}. Varying $s$ therefore changes the interactions entering $E_k$ and the measured overlap statistics.

\begin{figure*}[t] 
    \centering 
    \begin{subfigure}[b]{0.46\textwidth}
        \centering
        \begin{overpic}[width=\textwidth]{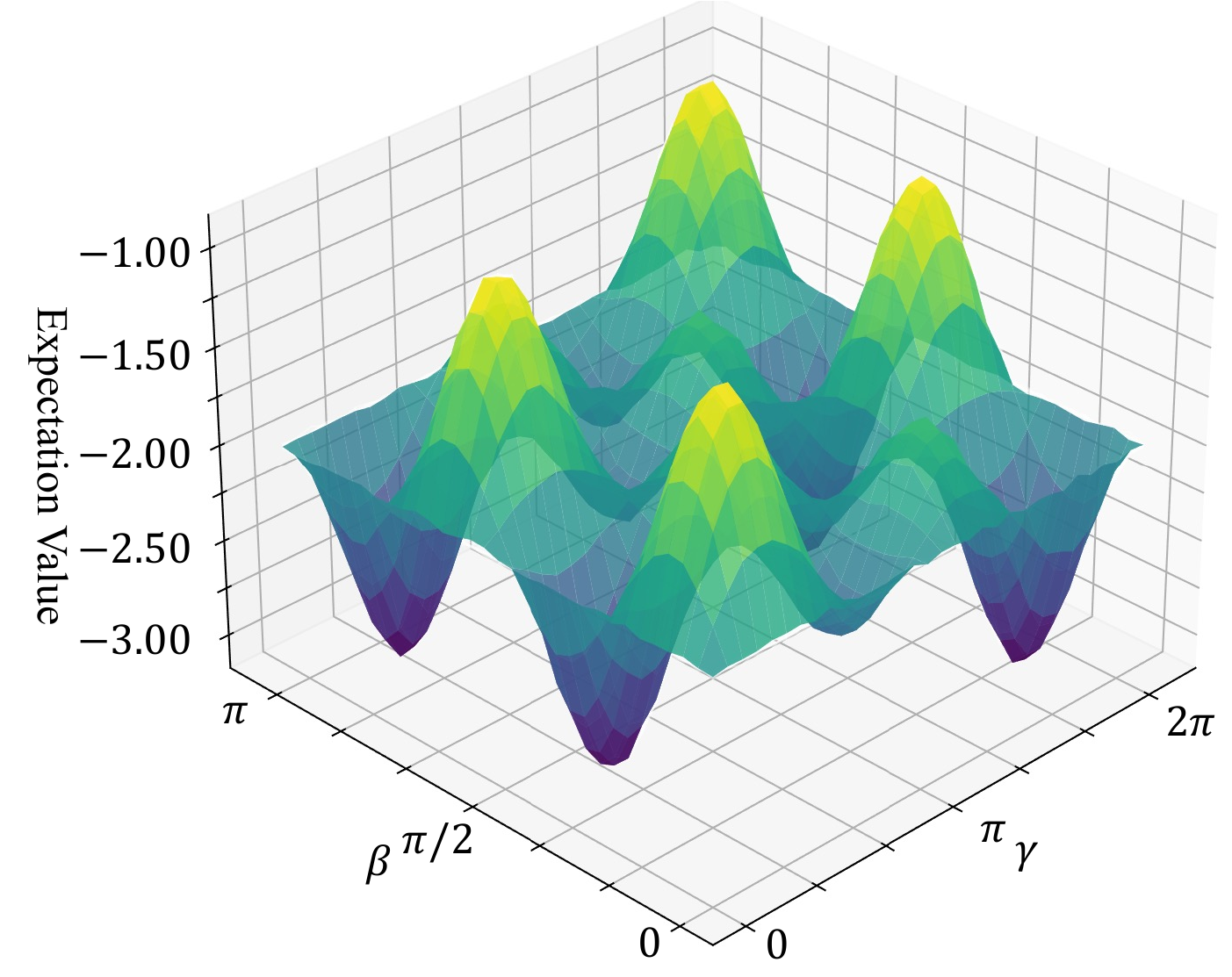}
            \put(2, 160){\textbf{(a)}}
        \end{overpic}
        \label{fig:ideal_landscap}
    \end{subfigure}
     \hfill 
     \begin{subfigure}[b]{0.46\textwidth}
        \centering
        \begin{overpic}[width=\textwidth]{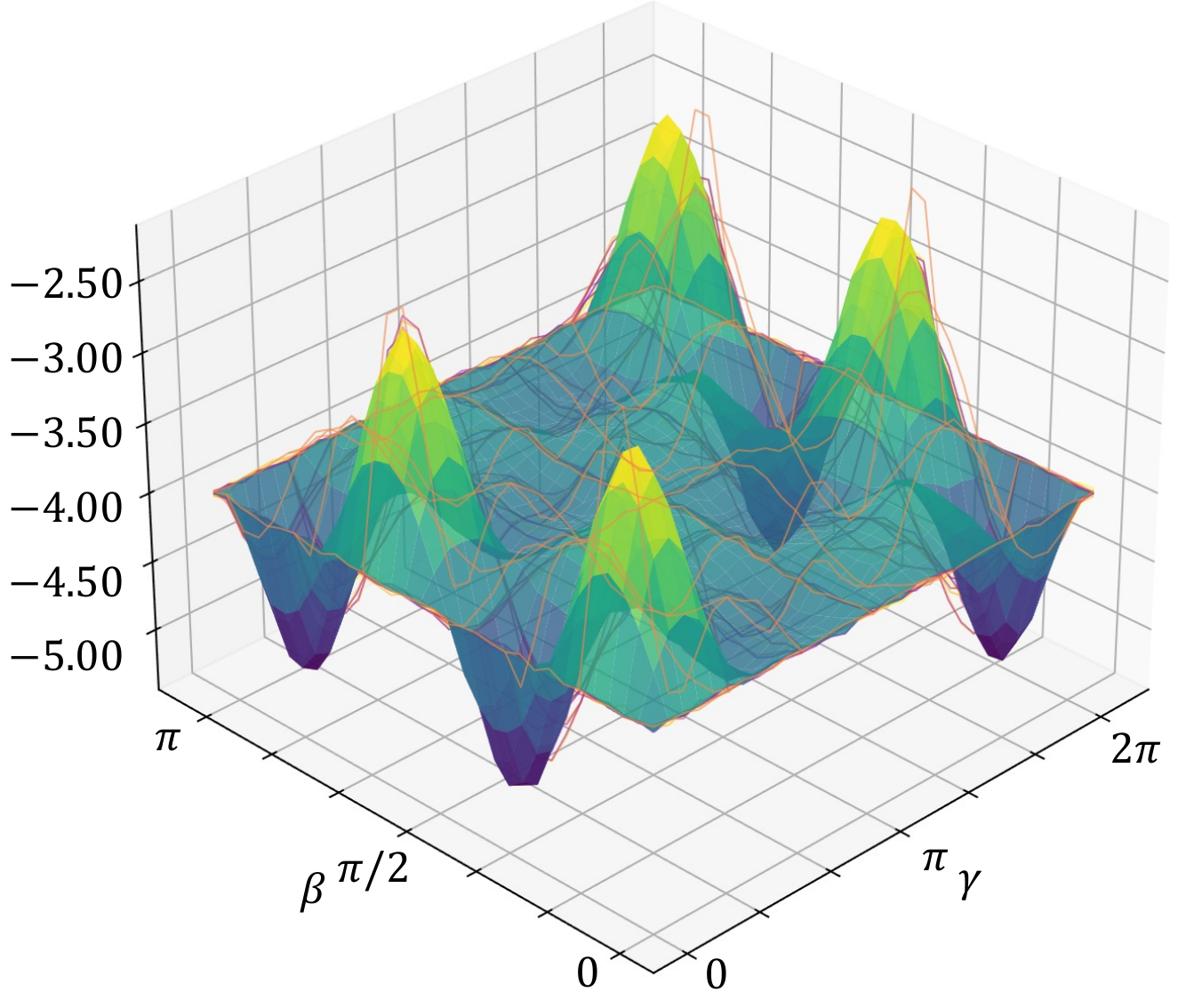}
            \put(0, 160){\textbf{(b)}}
        \end{overpic}
        \label{fig:distort_landscap}
    \end{subfigure}
    \caption{Divide-and-conquer approach and energy landscapes. The energy landscapes of $2^3$ sub-problems (overlaid) show that the geometric features (minima/maxima) align closely, despite the disparate linear biases induced by different frozen bitstrings. (a) In the ideal landscape (No Coefficients), the hypothesis holds, so the eight landscapes are mathematically nearly identical. As a result, the eight wireframes overlap perfectly, making it appear as though only a single landscape is plotted. (b) The induced linear-term coefficients distort the landscape, causing the eight wireframes to separate or “fuzz out,” while still remaining similar to the selected representative sub-problem.}
\label{fig:landscapes}
\end{figure*}

We consider a one-dimensional long-range percolation model~\cite{long_rang_percolation} with nearest-neighbor edges always present. Additional edges between vertices at distance $r$ are sampled with probability 
\bea\label{eq:lrp}
p(r)=1-\exp(-\, r^{-s}).
\eea
Here, $s$ is a structural control parameter for  graph ensembles that defines the decay of interaction range. In other words, this decay exponent $s$ sets the range of nonlocal connections. The long-range percolation model exhibits the percolation transitions. For the graph diameter $D_L$\footnote{The graph diameter is formally defined as the maximum shortest-path distance (geodesic distance) between any pair of nodes in the network: $D_L \equiv \max_{i,j \in V} d(i,j)$, where $d(i,j)$ denotes the minimum number of edges required to connect node $i$ to node $j$.} on a system of size $L$
the diameter remains bounded, $D_L\sim \mathcal{O}(1)$ for $s<1$. For $1<s<2$,$D_L$ shows polylogarithmic growth, $D_L\sim \log L^{\Delta(s)}$. For $s>2$, the diameter grows linearly with system size, $D_L\sim L$.
In this sense, $s$ acts as a structural control parameter. Since a depth-$p$ QAOA circuit only correlates degrees of freedom within a $p$-hop light cone, changing $s$ effectively renormalizes the local neighborhood geometry experienced by each cost term. 

For Fig.~\ref{fig:percolation_phase_boundary_1}, $p=1$ and $L=50,100,200,400$. The couplings are unity on the sampled edges, with no original linear fields. For each of the 20 equally spaced values $s\in[0.2,1.2]$, we average over 100 independent graph realizations. For each graph, a set of $m=\operatorname{round}(0.2L)$ vertices is selected uniformly at random and held fixed, and $M=101$ frozen assignments are sampled with independent equiprobable spins. Thus $m=10,20,40,80$ for the four sizes. This diagnostic freezing scheme differs from the small, highest-degree frozen sets used in the DO-QAOA benchmarks. Each landscape is evaluated analytically on the same $11\times11$ grid, with $\gamma\in[0,\pi]$ and $\beta\in[0,\pi/2]$, including both endpoints. The reduced Hamiltonian includes the induced fields and the frozen-frozen constant offset. We first average the pairwise overlaps within each graph and then average the resulting $q$ over graph realizations.

The displayed $q(s)$ curves cross in a region near $s\approx0.6$. At smaller $s$, the overlap generally decreases with $L$; at larger $s$, it increases toward unity. Here, \textit{crossover} refers only to this change observed in finite systems; its nature in the thermodynamic limit is not determined by the present data. The measurements do not establish a single common landscape class or the scaling of landscape fluctuations with $L$.

Appendix~\ref{app:higher_depths} reports overlap measurements at higher depths. For $p=7,\ldots,12$ on the accessible systems $L\leq20$, the region over which $q$ increases moves toward larger $s$, reaching approximately $s=0.9$ by $p=12$. At $L=20$ and $s\gtrsim1.0$, the plotted values remain approximately $q=0.92$--$0.95$. These are finite-system observations, not evidence of a depth-independent critical point.

\begin{table}[h!]
\centering
\caption{Empirical verification of the Landscape Similarity Hypothesis. Metrics (averaged over all pairwise sub-problems) show that even with induced linear coefficients (\textit{With Coeffs}), the correlation remains near 1.0, indicating the landscape shape is preserved. Here $m$ is the number of frozen qubits, while MSE denotes the mean squared error in evaluating the Correlation, defined as the average geometric fidelity $r$ between two energy landscapes corresponding to distinct sub-problems. The $L_\infty$ Distance measures the average pointwise deviation between two energy surfaces. For a visual comparison, see also Fig.\ \ref{fig:pcontour_results}.}
\label{tab:landscape_verify}
\resizebox{\columnwidth}{!}{%
\begin{tabular}{llccc}
\toprule
\toprule
$m$ & Condition & MSE (Raw) & $L_\infty$ Distance & Correlation \\
\midrule
1 & No Coeffs & 1.17e-03 & 0.1396 & 0.99928 \\
 & With Coeffs & 1.20e-03 & 0.1339 & 0.99916 \\
\midrule
2 & No Coeffs & 8.01e-04 & 0.1013 & 0.99921 \\
 & With Coeffs & 1.76e-01 & 1.0427 & 0.79639 \\
\midrule
3 & No Coeffs & 2.28e-04 & 0.0527 & 0.99940 \\
 & With Coeffs & 8.53e-02 & 0.7709 & 0.80185 \\
\bottomrule
\bottomrule
\end{tabular}%
}
\end{table}

\subsection{Justification for the Landscape Similarity Hypothesis}
\label{subsec:empirical_validation}

We hypothesize that the QAOA energy landscape is primarily determined by the graph 
connectivity (the quadratic terms), with the linear biases acting 
as minor perturbations. This hypothesis is supported by both theoretical bounds 
(Appendix~\ref{app:theorectical_landscape}) and empirical evidence in 
Section~\ref{sub:emperical_evaluation}. To demonstrate this relation, we generated an 
ensemble of 10 random graphs using NetworkX~\cite{network_x}, each consisting of $n = 10$ nodes. 
Specifically, the graphs are drawn from the Erd\H{o}s--R\'{e}nyi $G(n, p)$ 
model~\cite{erdHos1960evolution} with edge probability $p = 0.5$, which produces a moderately 
dense random topology with an expected degree of $\bar{k} = (n-1)p = 4.5$. It is important 
to note that this Erd\H{o}s--R\'{e}nyi ensemble is not parameterized by the LRP decay exponent $s$. We therefore do not assign it to an $s$-defined regime. It provides a separate test of landscape similarity on moderately dense random graphs. For every subgraph, we evaluated the 1-layer QAOA energy 
landscape using a comprehensive grid search with a resolution of $32 \times 32$. This 
approach yielded 1,024 energy data points per instance, enabling a detailed visualization 
of the landscape's critical structural properties. These graphs provide a versatile testing 
platform for our experiments. Fig.~\ref{fig:landscapes} visually confirms this hypothesis.
We plot the energy landscapes for distinct sub-problems generated from the same parent 
graph, despite the varying induced fields, the locations of peaks, valleys, 
and local minima remain strikingly consistent. Consequently, the optimal parameters 
$(\gamma^*, \beta^*)$ for one sub-problem should be highly predictive of the optima of the remaining sub-problems.

To quantitatively validate the landscape similarity predicted by our hypothesis, we 
evaluate the pairwise correspondence between the representative and target sub-problems 
using three complementary metrics, as summarized in Table~\ref{tab:landscape_verify}. 
We first calculate the Pearson correlation coefficient ($r$) to quantify the geometric 
fidelity of the energy manifolds:
\begin{equation}
    r = \frac{\sum_{a=1}^{D} (E_A^{(a)} - \bar{E}_A)(E_B^{(a)} - \bar{E}_B)}
    {\sqrt{\sum_{a=1}^{D} (E_A^{(a)} - \bar{E}_A)^2}
    \sqrt{\sum_{a=1}^{D} (E_B^{(a)} - \bar{E}_B)^2}}.
    \label{eq:pearson_r}
\end{equation}

Here $a$ indexes the common parameter points $\boldsymbol{\theta}_a$, with $E_A^{(a)}=E_A(\boldsymbol{\theta}_a)$ and $E_B^{(a)}=E_B(\boldsymbol{\theta}_a)$. The bars denote averages over these $D$ points; $D=1024$ for the $32\times32$ grid in this $p=1$ comparison. A value of $r$ near unity indicates correlated sampled energy values, but does not by itself establish identical local optima or basin curvature.

As shown in Table~\ref{tab:landscape_verify}, the empirical results strongly support our 
hypothesis. Freezing the highest-degree node ($m = 1$) results in a correlation exceeding 
$0.999$ and a negligible Mean Squared Error (MSE) ($\sim 10^{-3}$) even in the presence 
of induced linear coefficients. This confirms that the invariant quadratic term 
($H_{\mathrm{quad}}^{(R)}$) dominates the landscape structure. Furthermore, while the 
induced linear coefficients introduce stronger perturbations when more qubits are frozen 
($m = 2, 3$) and reduce the similarity of the full landscapes to a moderate level 
(typically $r \approx 0.8$), the quadratic backbone itself remains nearly unchanged, 
exhibiting near-perfect similarity ($r > 0.999$) when the quadratic component is isolated. 
These correlations describe the tested Erd\H{o}s--R\'{e}nyi ensemble and do not assign it to a phase or establish an LRP threshold for parameter transfer.


Appendix~\ref{app:theorectical_landscape} bounds differences between the branch-dependent energy landscapes under stated assumptions. This analysis supplements the empirical correlation metrics but does not establish a universal landscape class or guarantee successful parameter transfer.

\begin{figure*}[t]
  \centering
  \includegraphics[width=0.99\textwidth]{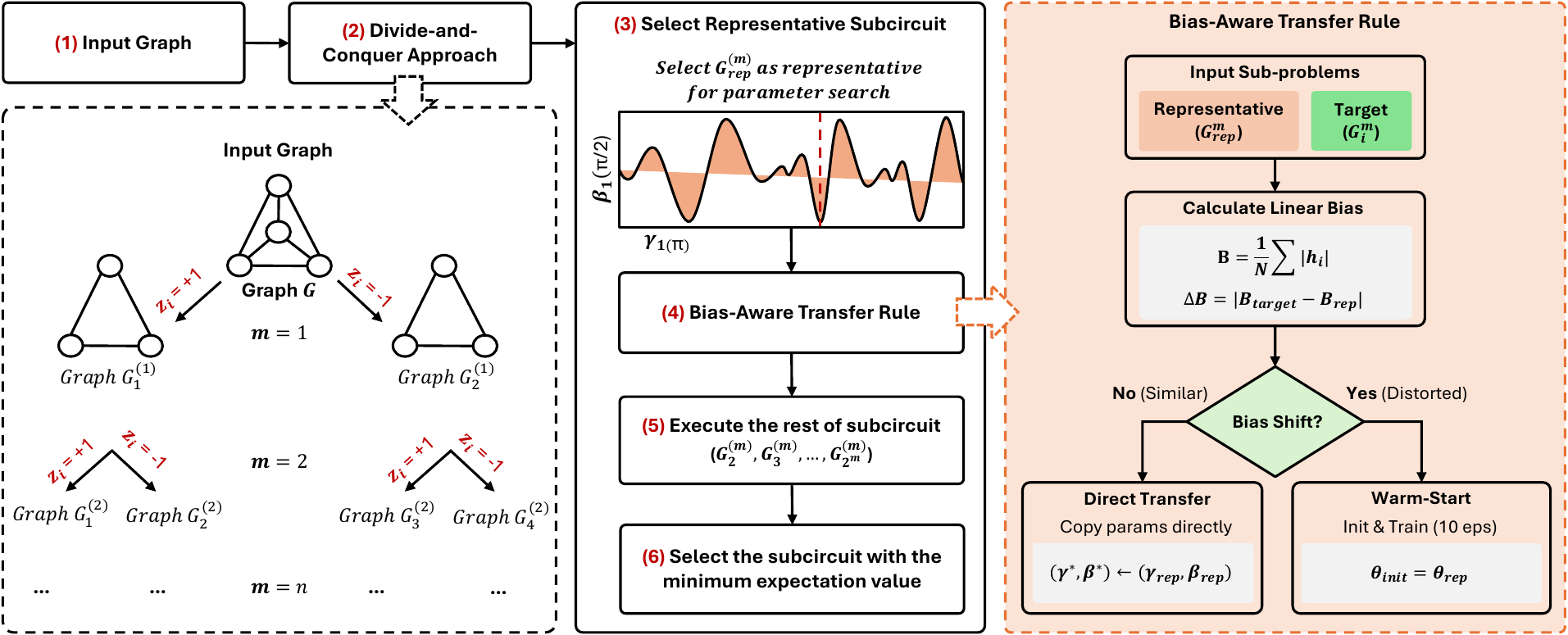}
  \caption{Overview of the DO-QAOA framework. The process begins with partitioning the Input Graph, selecting a Representative Subcircuit for training, and using the Bias-Aware Transfer Rule to efficiently transfer parameters to the remaining $2^m - 1$ sub-problems.}
  \label{fig:method_overview}
\end{figure*}

\section{DO-QAOA Method}
\label{sec:method}

Based on this observation, we introduce a strategy to reduce the optimization cost. Instead of performing $2^m$ redundant searches, we reduce repeated optimization
by fully training one representative sub-problem instead
of independently training every frozen assignment.
The number of fully trained representatives is fixed at
$K_{\mathrm{rep}}=1$. Whether the transferred parameters
provide satisfactory target performance is assessed
empirically and is not guaranteed by this representative
count. Identifying multiple landscape classes and training
one representative per class is left to future work.

The core idea of DO-QAOA is to perform full variational optimization on a single representative sub-problem to extract the optimal parameters ${\bm{\theta}}^*_\mathrm{rep}$. These parameters are then propagated to the remaining instances.

Motivated by the landscape-overlap measurements $q$ in Fig.~\ref{fig:percolation_phase_boundary_1} and the Pearson correlations $r$ in Table~\ref{tab:landscape_verify}, DO-QAOA reduces repeated full training by optimizing one representative and transferring its parameters to the remaining assignments, with refinement when required. Fig.~\ref{fig:method_overview} illustrates the complete pipeline, which consists of two phases and six integrated steps.

\subsection{Phase I: Graph Partitioning and Representative Selection}
The process commences with an arbitrary input graph $G$ (Step 1 of Fig.\ \ref{fig:method_overview}). To mitigate the hardware constraints of NISQ devices, we employ a divide-and-conquer strategy (Step 2 of Fig.\ \ref{fig:method_overview}) that recursively partitions the graph by freezing $m$ highest-degree nodes into classical states $z_i \in \{+1, -1\}$. This decomposition generates a set of $2^m$ distinct sub-problems, denoted as $\{G_1^{(m)}, \dots, G_{2^m}^{(m)}\}$, which share the same underlying connectivity but differ in their induced linear magnetic fields. In the practical execution pipeline, we pick a specific $m$ and stick to it for that entire run. We select nodes strictly according to degree centrality. We sort all nodes by degree and pick the top $m$ nodes. If multiple nodes share the same degree, we typically break ties by their node index (or a stable sort order) to ensure the experiment is reproducible.

Our pipeline selects the first generated sub-problem as the representative instance $G_{rep}^{(m)}$ (Step~3 of Fig.~\ref{fig:method_overview}); for $m=1$, the example is the $+1$ frozen configuration. We optimize this instance to obtain $\theta_{rep}^* = (\gamma_{rep}^*,\beta_{rep}^*)$ and use these parameters for transfer to the other reduced instances. This is an implementation convention; equivalent performance for every possible representative is not assumed.

In the typical operating regime of our method, only a small number of high-degree (hotspot) qubits are frozen (default design uses $m = 1, 2, 3$), as is common for power-law and real-world graphs. 
Empirically, we observe that freezing additional qubits beyond this regime yields diminishing improvements in solution quality, and therefore does not justify the added computational cost (discussed in more detail in Appendix~\ref{sec:appendix_limits}).

In addition, the optimization cost of this representative training step can be further reduced through informed parameter initialization. As detailed in Appendix~\ref{app:initialization}, we employ an initialization strategy that significantly accelerates convergence compared to random initialization, thereby shortening the training loop without compromising solution quality.

\begin{figure*}[thbp]
    \centering
    
    \begin{subfigure}[b]{0.30\textwidth}
        \centering
        \begin{overpic}[width=\textwidth]{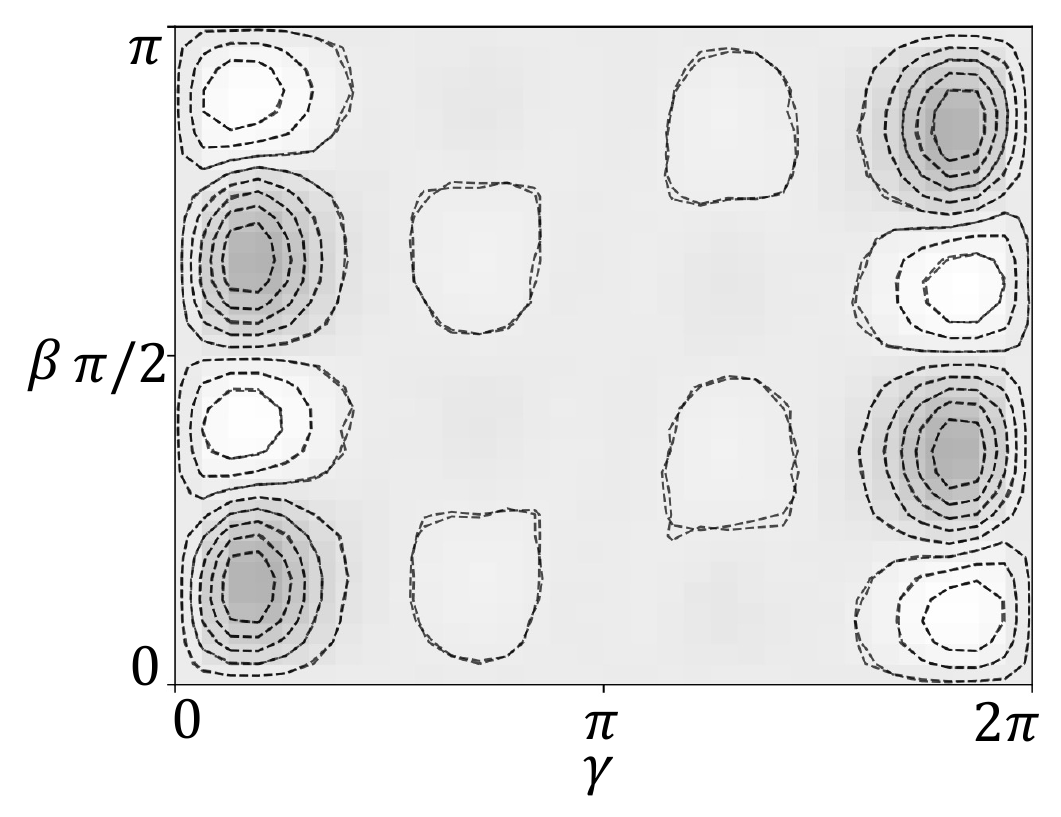}
            \put(0, 110){\textbf{(a)}}
        \end{overpic}
        \label{fig:pcontour_m1_a}
    \end{subfigure}
    \hspace{-0.3em}
    \begin{subfigure}[b]{0.30\textwidth}
        \centering
        \begin{overpic}[width=\textwidth]{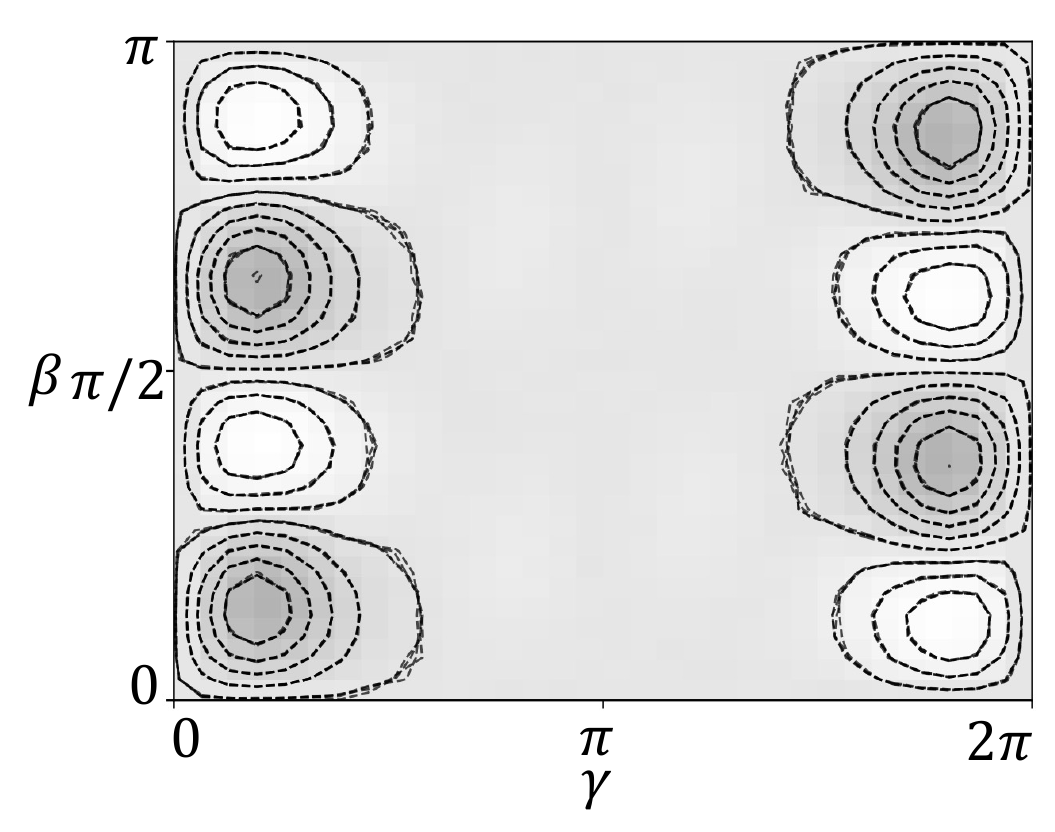}
            \put(0, 110){\textbf{(b)}}
        \end{overpic}
        \label{fig:pcontour_m1_b}
    \end{subfigure}
    \hspace{-0.3em}
    \begin{subfigure}[b]{0.30\textwidth}
        \centering
        \begin{overpic}[width=\textwidth]{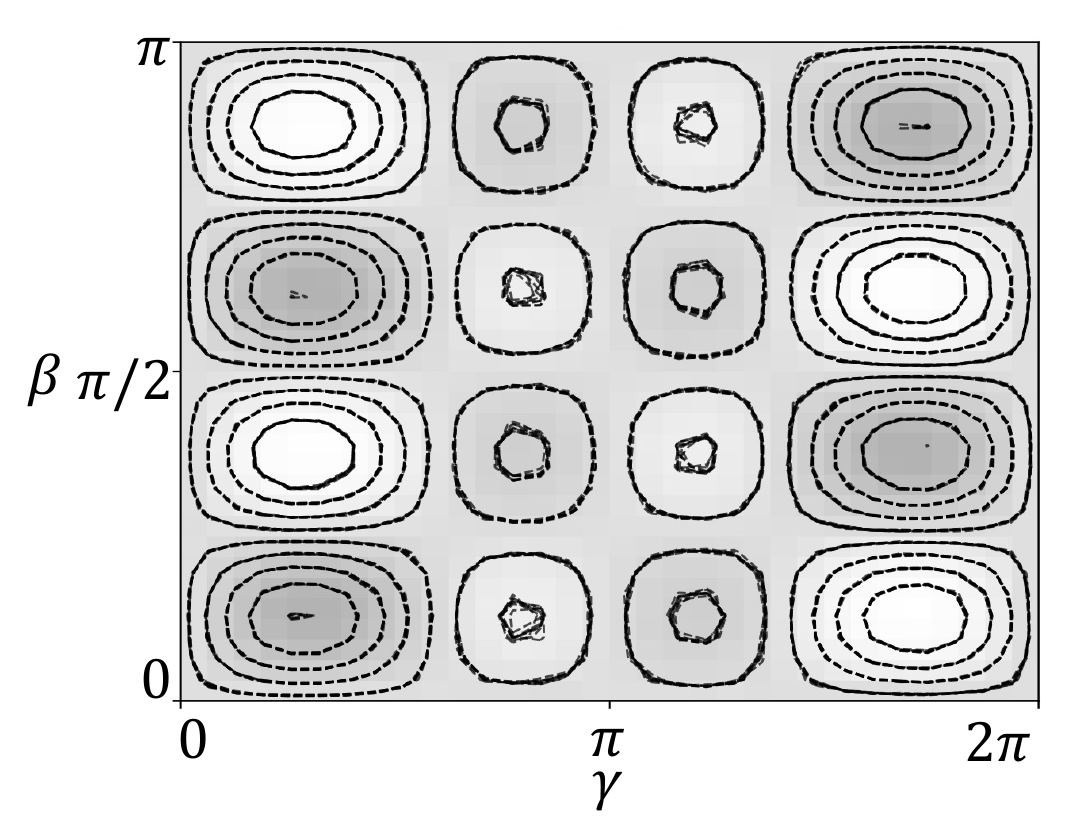}
            \put(0, 110){\textbf{(c)}}
        \end{overpic}
        \label{fig:pcontour_m2_a}
    \end{subfigure}

    \begin{subfigure}[b]{0.30\textwidth}
        \centering
        \begin{overpic}[width=\textwidth]{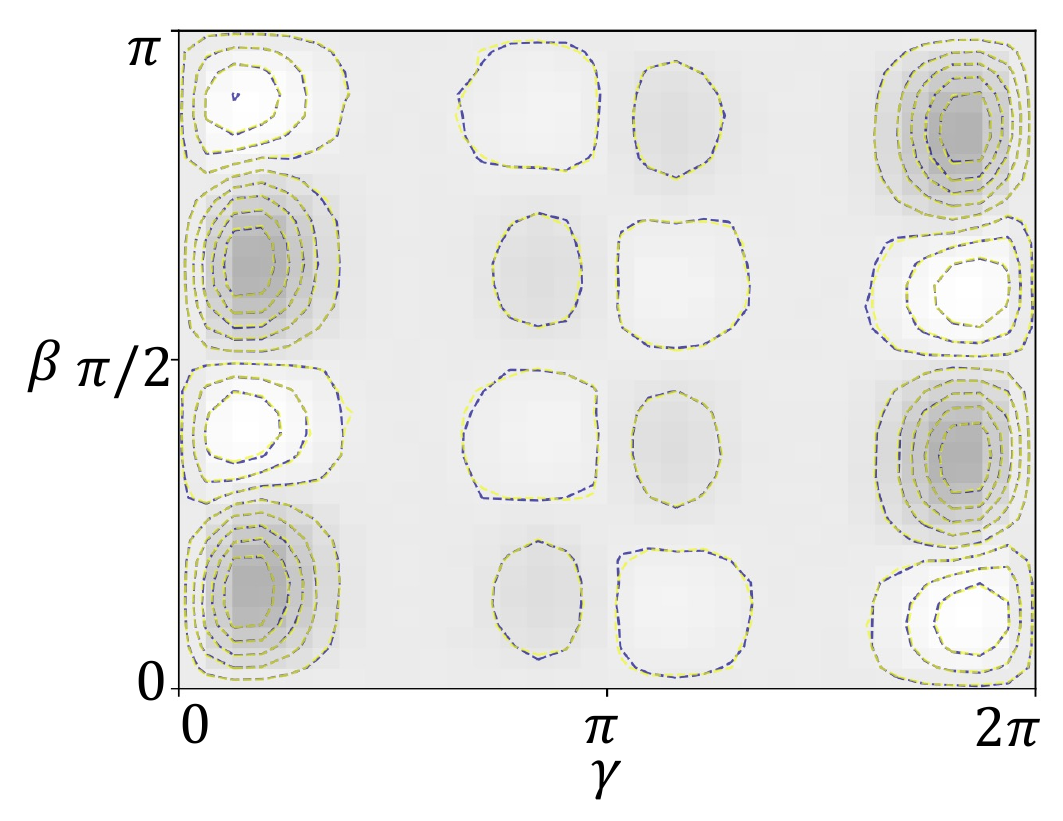}
            \put(0, 110){\textbf{(d)}}
        \end{overpic}
        \label{fig:pcontour_m2_b}
    \end{subfigure}
    \hspace{-0.3em}
    \begin{subfigure}[b]{0.30\textwidth}
        \centering
        \begin{overpic}[width=\textwidth]{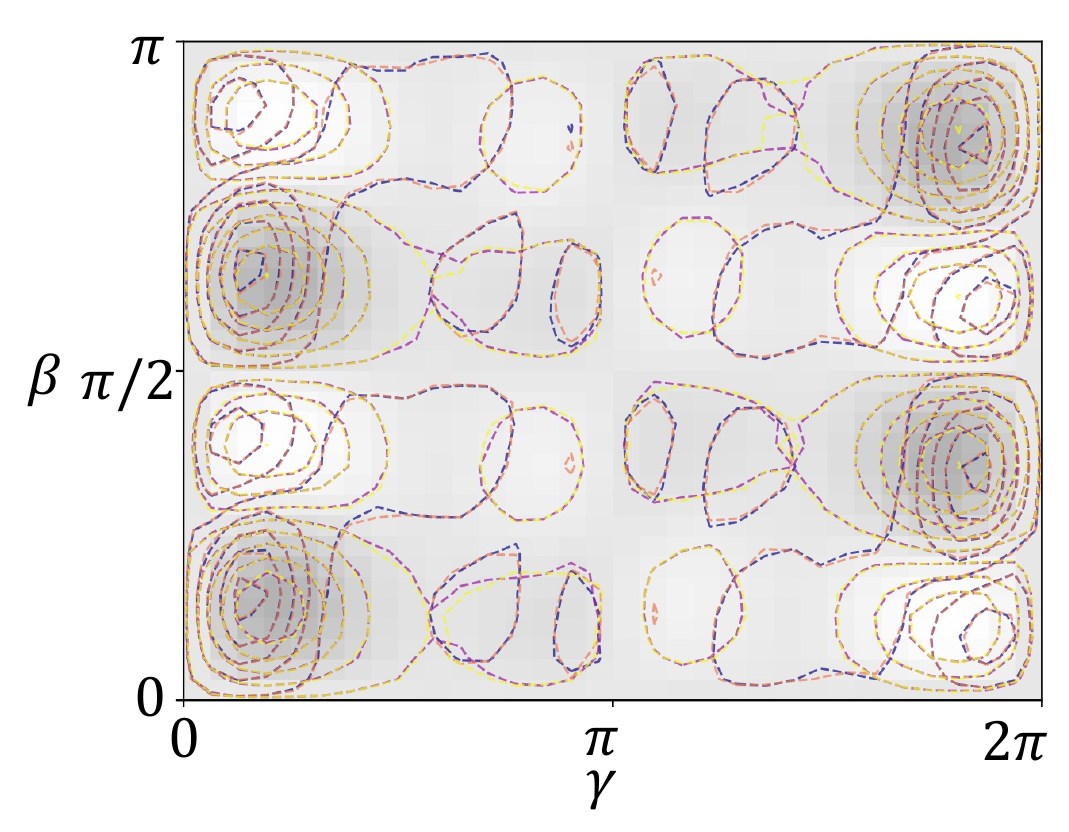}
            \put(0, 110){\textbf{(e)}}
        \end{overpic}
        \label{fig:pcontour_m3_a}
    \end{subfigure}
    \hspace{-0.3em}
    \begin{subfigure}[b]{0.30\textwidth}
        \centering
        \begin{overpic}[width=\textwidth]{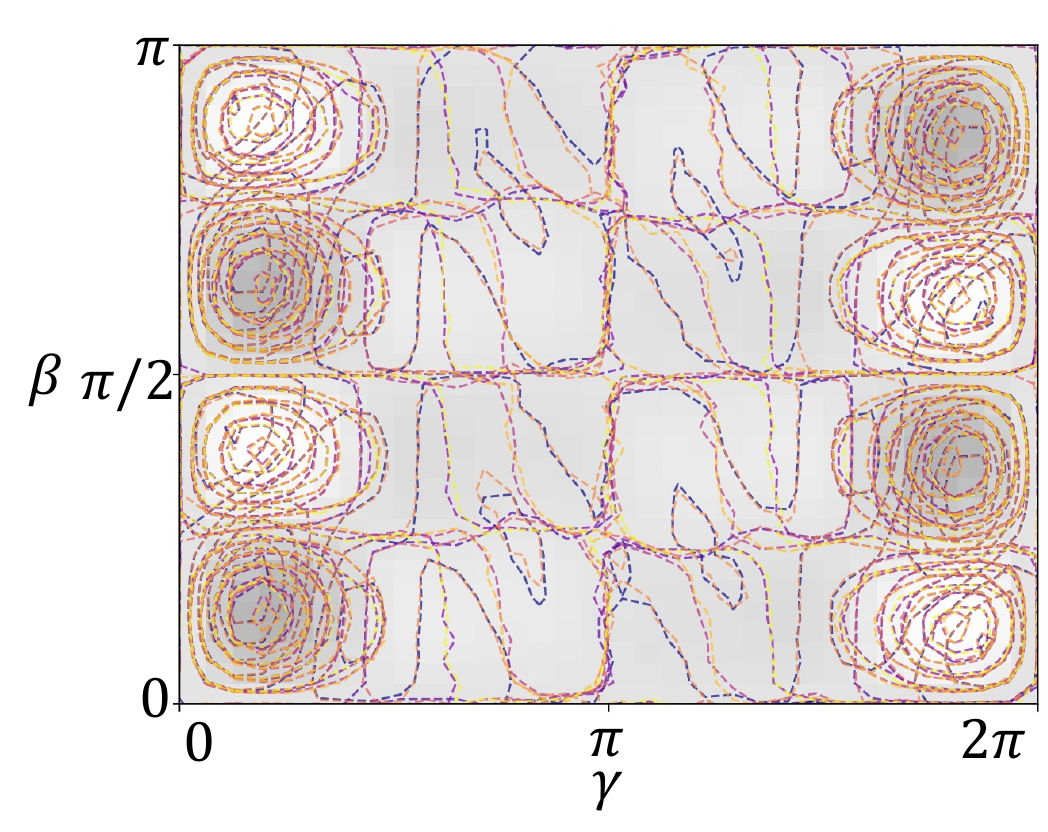}
            \put(0, 110){\textbf{(f)}}
        \end{overpic}
        \label{fig:pcontour_m3_b}
    \end{subfigure}

    \caption{Reference and distorted QAOA energy landscapes with 1-layer. Panels (a--c) show the ideal (reference) energy landscapes without induced coefficients for $m=1,2,3$. Panels (d–f) show overlaid energy contours of sub-problems generated by freezing $m$ nodes, illustrating how the landscape is distorted. However, basin stability is maintained: the optimization minima (light regions) remain localized within the same vicinity of $(\gamma, \beta)$.}
    \label{fig:pcontour_results}
\end{figure*}

\subsection{Phase II: Parameter Transfer and Execution}
Once the representative parameters are obtained, they are propagated to the remaining sub-problems via the Bias-Aware Transfer Rule (Step 4 of Fig.\ \ref{fig:method_overview}). This mechanism determines whether the landscape of a target sub-problem $G_i^{(m)}$ is sufficiently similar to the representative to warrant direct parameter transfer.

We first calculate the aggregate magnitude of the induced local magnetic fields, $B = \frac{1}{N} \sum |h_k|$, for both the representative and the target graphs. We then evaluate the bias distortion $\Delta B = |B_{\text{target}} - B_{\text{rep}}|$. If this distortion is below a critical threshold (experimentally set at 0.3), the landscapes are deemed congruent, and we apply Direct Transfer, simply copying $\bm{\theta}^*_\mathrm{rep}$ to the target. Conversely, if $\Delta B > 0.3$, this indicates a significant difference in the induced linear bias, which suggests a shift in the ground state. 
(We discuss the threshold selection in more detail in Appendix~\ref{sec:appendix_sensitivity}). 
In this scenario, we employ a Warm-Start strategy, initializing the target with $\bm{\theta}^*_\mathrm{rep}$, followed by a brief fine-tuning stage (e.g., 10 epochs) to adapt to the new landscape.

With the parameters established for all instances, we execute the remaining quantum sub-circuits (Step 5 of Fig.\ \ref{fig:method_overview}) to evaluate their expectation values. Finally, we aggregate the results and select the configuration that yields the minimum global energy (Step 6 of Fig.\ \ref{fig:method_overview}) as an approximate solution to the original problem.

\section{Evaluation}
\label{sec:evaluation}
In this section, we empirically validate the theoretical framework of DO-QAOA, transitioning from fundamental landscape analysis to practical performance benchmarking under realistic noise models. Our evaluation follows a systematic progression designed to isolate the physical mechanisms of graph decimation from the complexities of specific problem instances. We begin in Section~\ref{sub:emperical_evaluation} by visualizing the
%
similarity of landscape contours under iterative graph decimation. Following this visual comparison, we evaluate the performance of DO-QAOA under realistic noisy conditions (see Appendix~\ref{app:exp_setup} for details of the experimental setup).
To emphasize the central contributions of this work, we focus on two key performance metrics:
(i) solution quality measured by the ARG, and
(ii) training complexity quantified by the total number of quantum shots.
Additional hardware-level metrics and per-dataset analyses are provided in the Appendix~\ref{app:per_dataset_analysis}.

\subsection{Visual Validation of Landscape Similarity}\label{sub:emperical_evaluation}

To test the Landscape Similarity Hypothesis (Section~\ref{sec:motivation}) in a controlled regime, we first evaluate DO-QAOA on 10 NetworkX random graphs. 
In this simplified setting, we can isolate the effects of graph decimation (node freezing) on the energy landscape without the confounding factors of complex real-world graphs.

Fig.~\ref{fig:pcontour_results} illustrates the QAOA energy landscapes for $m=1,2,3$ in the $p=1$ comparison of 10 random graphs described in Section~\ref{subsec:empirical_validation}. Panels (a--c) show reference landscapes with induced fields removed. Panels (d--f) overlay the contours of the corresponding sub-problems with induced fields included. The separation of these contours shows the effect of the induced fields and complements the correlations reported in Table~\ref{tab:landscape_verify}.

For $m=1$, the contours of the distinct sub-problems in Fig.~\ref{fig:pcontour_results}(d) overlap closely. In Table~\ref{tab:landscape_verify}, the correlation with induced fields included is $r=0.99916$. This is an observation for the tested graphs and parameter grid; it does not establish a universality class or guarantee parameter-transfer fidelity.
As the decimation depth increases to $m=2$ and $m=3$ in Fig.~\ref{fig:pcontour_results} (e) and (f), respectively, 
the contours in panels (e--f) separate, while the light regions containing the minima remain in similar parts of parameter space. The corresponding correlations with induced fields are $r=0.79639$ and $r=0.80185$, respectively (Table~\ref{tab:landscape_verify}). These plots show appreciable landscape distortion and approximate alignment of minima in the displayed examples, rather than unchanged basin geometry.

\begin{figure}[thbp]
    \centering
    
    \begin{subfigure}[b]{\linewidth}
        \centering
        \includegraphics[width=0.99\linewidth]{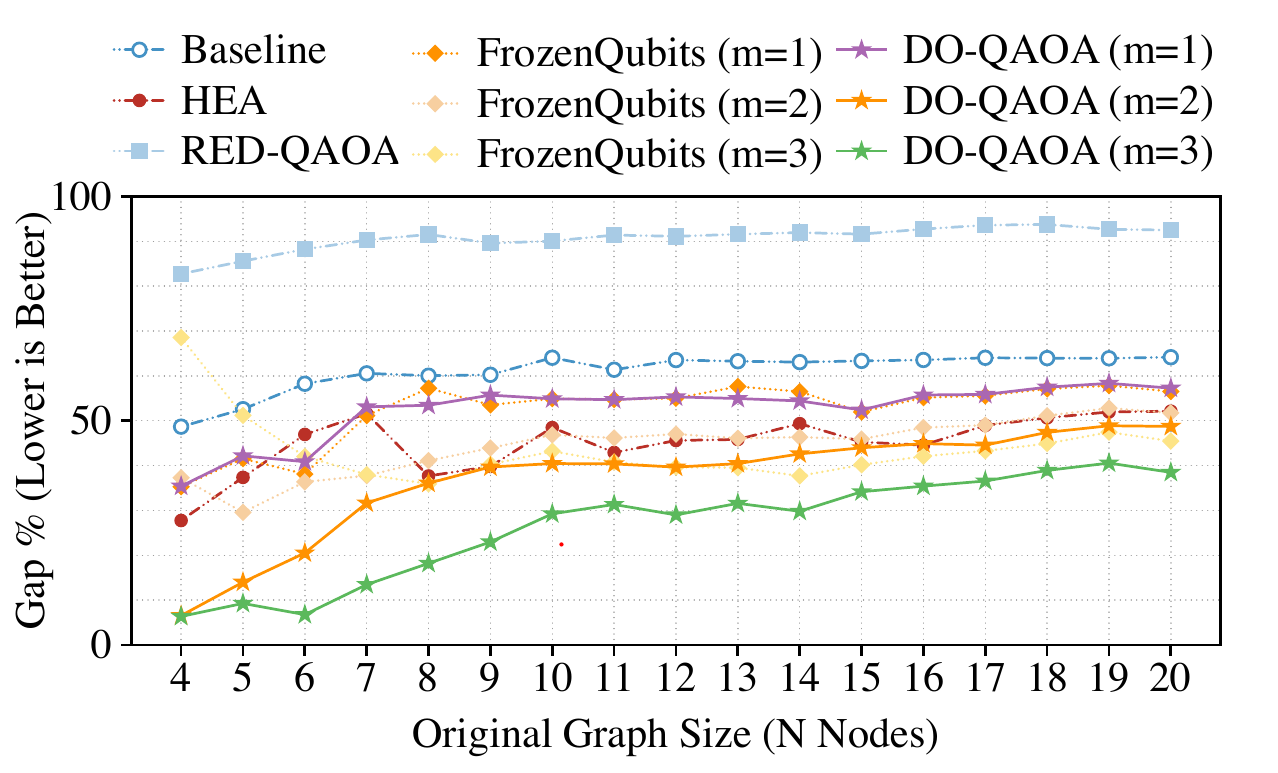}
    \end{subfigure}


    
    \caption{Scaling analysis on Power-Law graphs using ARG}
    \label{fig:powerlaw_results}
\end{figure}

\begin{figure*}[thbp]
    \centering
    \includegraphics[width=0.8\linewidth]{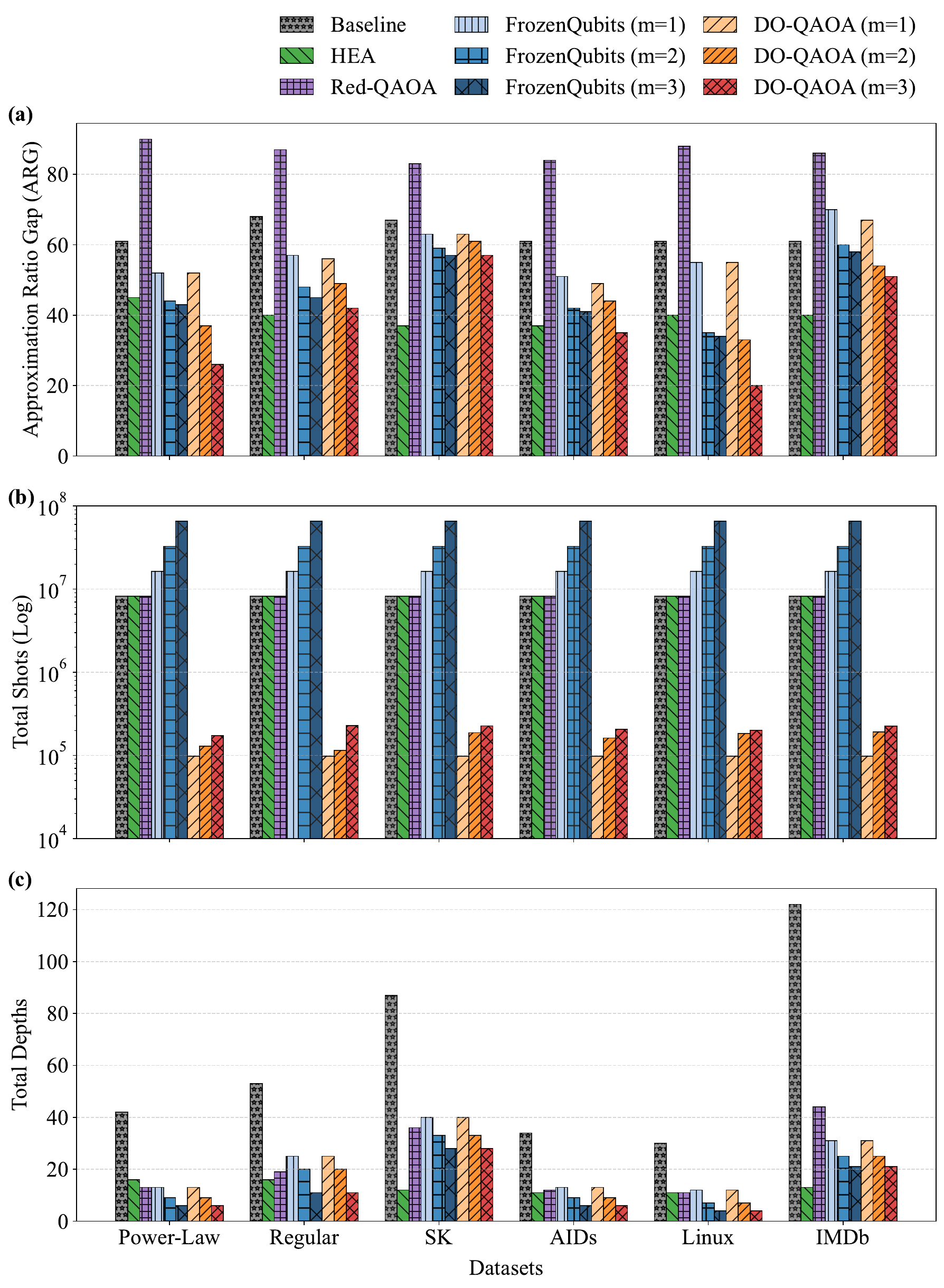}
    
    \caption{Average end-to-end performance metrics across all benchmarks. (a) Average ARG, (b) Total Shots Across All Benchmarks, and (c) Total Circuit Depths Across All Benchmarks $= p\left(\textrm{Cost-Layer Depth} + \textrm{Mixer-Layer Depth} \right) \approx p\left(\textrm{Graph Degree} + 1\right)$. The comparison highlights the trade-off between solution quality (ARG) and computational cost (total shots). 
    %
    The comparison reports the solution-quality and shot-cost tradeoff under the evaluated budgets; it does not establish a change in asymptotic scaling with system size.
    For the $m=3$, $p=1$ comparison under the \texttt{FakeBrisbane} noise model, the dataset-averaged ARG is equal to or lower than FrozenQubits across the six benchmark families. This is not uniform over all freezing settings: at $m=2$, the ARG is slightly higher for Regular, SK, and AIDS graphs (Table~\ref{tab:additional_comparison}).
    }
    \label{fig:end_end_results}
\end{figure*}

\subsection{Generalization to Real-World Datasets}\label{evaluation_arg}
Having established the physical mechanism of landscape similarity on random graphs, we now quantify the solution quality of DO-QAOA on standard synthetic benchmarks.

To quantify the performance degradation relative to the optimal solution, we evaluate algorithmic fidelity using the Approximation Ratio Gap (ARG), following prior works~\cite{arg_1, arg_2, herrman2022multi}. The ARG is defined as
\begin{equation}
    \mathrm{ARG} = 100 \times \left| \frac{E_{\min} - \langle H_C \rangle}{E_{\min}} \right|,
    \label{eq:arg}
\end{equation}
where $\langle H_C \rangle$ denotes the expectation value of the cost Hamiltonian obtained from the QAOA circuit, and $E_{\min}$ is the exact ground-state energy computed via classical brute-force enumeration over all $2^n$ computational basis states. The ARG takes values in $[0, \infty)$, with lower values indicating closer agreement with the optimal solution.

We first evaluate the solution quality of DO-QAOA using the ARG, defined in Eq.~(\ref{eq:arg}).
Fig.~\ref{fig:end_end_results}(a) summarizes the ARG results across all benchmark classes, including synthetic graphs [power-law, 3-regular, and fully connected (Sherrington–Kirkpatrick (SK) model)] and real-world datasets (AIDS, Linux, and IMDb).
Across all graph types, baseline QAOA exhibits large approximation gaps under realistic noise, reflecting the degradation caused by deep circuits and accumulated gate errors (detail in Appendix~\ref{app:ablation_study_increase_p}).
Hardware-Efficient Ansatz (HEA) improves robustness but plateaus at moderate ARG values, while Red-QAOA consistently underperforms due to excessive structural reduction.

Divide-and-conquer approaches based on qubit freezing improve approximation quality as the number of frozen qubits $m$ increases.
In particular, FrozenQubits achieves progressively lower ARG values by reducing circuit depth and suppressing noise.
At $m=3$, DO-QAOA matches or improves upon the dataset-averaged ARG of FrozenQubits in Table~\ref{tab:additional_comparison}. For Power-Law graphs, the ARG decreases from $43\%$ to $26\%$, a relative reduction of approximately $40\%$ at $p=1$ under the \texttt{FakeBrisbane} noise model. Improvement is not uniform across all freezing settings: at $m=2$, the ARG is slightly higher for Regular, SK, and AIDS graphs. These results assess transfer performance under the reported benchmark budgets.

\subsection{Computational Complexity and Scaling}

We also evaluate the computational efficiency of DO-QAOA in terms of total quantum shots.
This metric captures the dominant costs in hybrid quantum-classical workflows, particularly on cloud-based NISQ platforms.

As shown in Fig.~\ref{fig:end_end_results}(b), standard divide-and-conquer methods incur an exponential increase in computational cost with the number of frozen qubits.
FrozenQubits requires $2^{m}$ independent variational optimization loops, leading to total shot counts exceeding $65 \times 10^{6}$ for $m=3$.

In contrast, DO-QAOA fully trains one representative and transfers its parameters to the other frozen assignments. The reduction from $2^m$ to one applies to full independent
optimization runs, rather than to the complete computational workload. The total quantum-shot cost can be decomposed as
\[
Q_{\mathrm{DO}}
= Q_{\mathrm{rep}}
+ \sum_{t\ne\mathrm{rep}} Q_{\mathrm{ref},t}
+ \sum_{z=1}^{2^m} Q_{\mathrm{eval},z},
\]
where the terms account for representative training, target-specific refinement, and branch evaluation, respectively. The cost of each training run and the number of required refinements can depend on the instance. Consequently, one fully trained representative does not imply constant total cost. For fixed $m$, reducing repeated training constitutes a constant-factor improvement with respect to system size $N$, rather than an asymptotic exponential speedup in $N$.

We used COBYLA~\cite{cobyla_optimization} with $8192$ shots per circuit evaluation
and a maximum of $1000$ optimization episodes.
DO-QAOA representative training could terminate before
the maximum when the target was reached (see Appendix~\ref{app:initialization}).
The reported shot totals include state the counted
training, refinement, and final-evaluation components.
With matched full-training budgets per branch, replacing
$2^m$ independent full-training runs by one reduces that
training component by a factor $2^m$, before refinement
and branch-evaluation costs. The larger measured ratios
therefore reflect the evaluated optimization workflow,
including the actual training effort, and cannot be
attributed to parameter transfer alone.

For $m=3$, Table~\ref{tab:additional_comparison} reports total shot counts of $65.54\times10^6$ for FrozenQubits and $0.17$--$0.23\times10^6$ for DO-QAOA across the benchmark datasets ($p=1$, $N\leq20$, \texttt{FakeBrisbane} noise model). For Power-Law graphs, the reported counts are $65.54\times10^6$ and $0.17\times10^6$, respectively, while the ARG decreases from $43\%$ to $26\%$. These measured savings concern the stated training and evaluation budgets.


\subsection{Summary of Results}

For the Power-Law benchmarks at $m=3$ and $p=1$ under the \texttt{FakeBrisbane} noise model, Table~\ref{tab:additional_comparison} reports an ARG of $26\%$ for DO-QAOA versus $43\%$ for FrozenQubits. Total shots decrease from $65.54\times10^6$ to $0.17\times10^6$, and the reported runtime decreases from $1{,}055$~s to $101$~s. These values quantify the benefit for the tested graphs ($N\leq20$) and budgets.
By transferring parameters across related sub-problems, DO-QAOA reduces repeated full optimization while retaining the circuit structure of the underlying divide-and-conquer
method. Its total cost still includes screening, refinement, execution of the frozen assignments, and aggregation.
The reported resource reductions do not establish a change in asymptotic complexity.

The Long-Range Percolation (LRP) model provides a qualitative analogy for connectivity-dependent landscape correlations, rather than a mapping from the benchmark graphs to a numerical value of $s$. For the $m=3$, $p=1$ benchmarks under the \texttt{FakeBrisbane} noise model, Table~\ref{tab:additional_comparison} reports ARG changes from $43\%$ to $26\%$ on Power-Law graphs, $45\%$ to $42\%$ on Regular graphs, $41\%$ to $35\%$ on AIDS, and $34\%$ to $20\%$ on Linux, when replacing independent FrozenQubits training with DO-QAOA. On SK, both methods give $57\%$; on IMDb, the ARG decreases from $58\%$ to $51\%$. These measured values support the benchmark comparison. They do not establish that an LRP threshold ensures parameter-transfer success on unstructured graphs.

\section{Discussion}\label{sec:discussion}
The results in Section~\ref{sec:evaluation} and Appendix~\ref{app:per_dataset_analysis} show that representative parameter transfer can reduce repeated training in divide-and-conquer QAOA. In the tested pipeline, one representative is fully optimized and the remaining instances receive transferred parameters, with refinement when triggered by the bias-aware rule. The measured benefit depends on transfer quality and on the training and evaluation budgets. It is therefore a practical resource reduction in the tested regimes, rather than a general change from exponential to constant total complexity.
DO-QAOA retains the reduced-circuit structure of the underlying divide-and-conquer method. Its additional steps seek to reuse useful parameters while limiting target-specific retraining. Screening, refinement, target execution, and aggregation remain part of the computational workload.


A notable empirical observation from our evaluation is the presence of diminishing returns as the number of frozen qubits increases (Appendix~\ref{sec:appendix_limits}). For small values of \(m\), freezing additional high-degree nodes yields consistent improvements by suppressing noise and reducing circuit depth. However, beyond this regime ($m = 3$), further freezing provides limited benefit, and additional cuts no longer meaningfully improve approximation quality. Instead, they primarily increase the overhead of classical bookkeeping.

These observations motivate a practical best practice of $m \le 3$ for most realistic problem instances considered in this work. At this operating point, DO-QAOA achieves a favorable balance between circuit simplification and landscape stability, delivering substantial improvements in approximation quality while maintaining minimal quantum and classical overhead. Importantly, this recommendation is empirical rather than prescriptive; it reflects the structure of the benchmark graphs studied here and the noise characteristics of current NISQ devices.

The present implementation uses one representative. Where landscape correlations weaken, transfer may become less effective and more target-specific optimization may be necessary. An extension could group the reduced instances and train one representative per group. With $K$ representatives, this would require $K$ full training runs plus clustering, screening, refinement, and evaluation costs. We have not established a system-size-independent bound on $K$ or benchmarked this extension; its benefit depends on these additional costs and the quality of the resulting groups.

\section{Conclusion}\label{conclusion}


We introduced DO-QAOA, a divide-and-conquer framework that reduces repeated variational optimization by transferring parameters from one trained representative to related sub-problems. In the tested regimes, this approach yields substantial empirical savings in quantum shots and runtime. These savings concern the evaluated pipeline and do not imply constant total cost or an asymptotic exponential speedup in problem size.

Finally, our findings highlight the importance of understanding variational energy landscapes in divide-and-conquer QAOA. Landscape-aware optimization strategies, such as DO-QAOA, provide a pathway to scalable hybrid quantum-classical algorithms under realistic hardware constraints.

\section*{Acknowledgement}
This work was partly supported by Institute for Information \& communications Technology Planning \& Evaluation (IITP) grant funded by the Korea government (MSIT) (No. 2020-0-00014, A Technology Development of Quantum OS for Fault-tolerant Logical Qubit Computing Environment) and Creation of the quantum information science R\&D ecosystem(based on human resources) through the National Research Foundation of Korea(NRF) funded by the Korean government (Ministry of Science and ICT(MSIT)) (No. RS-2023-00256050).

\appendix
\begin{figure*}[t] 
    \centering 
    \begin{subfigure}[b]{0.30\textwidth}
        \centering
        \begin{overpic}[width=\textwidth]{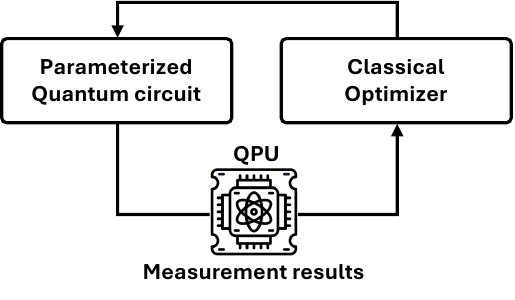}
            \put(2, 90){\textbf{(a)}}
        \end{overpic}
        \label{fig:overall_vqa}
    \end{subfigure}
    \hfill 
    \begin{subfigure}[b]{0.60\textwidth}
        \centering
        \begin{overpic}[width=\textwidth]{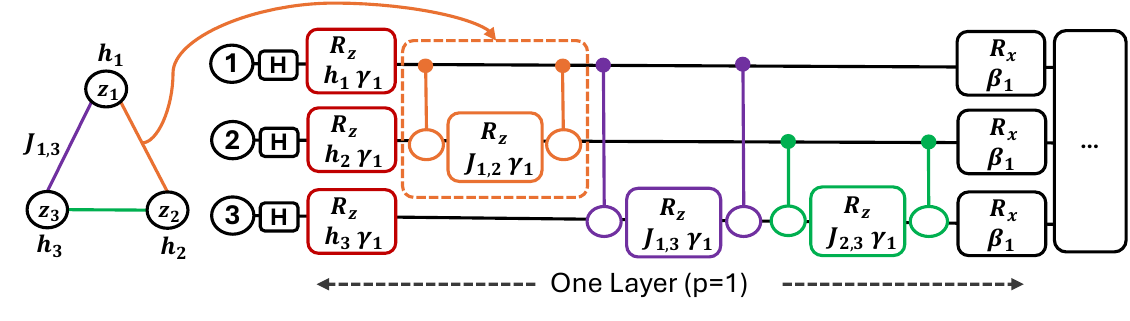}
            \put(2, 90){\textbf{(b)}}
        \end{overpic}
        \label{fig:graph_to_qaoa}
    \end{subfigure}
    \begin{subfigure}[b]{0.35\textwidth}
        \centering
        \begin{overpic}[width=\textwidth]{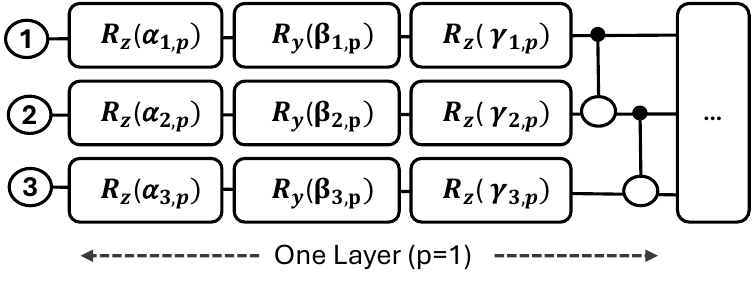}
            \put(2, 80){\textbf{(c)}}
        \end{overpic}
        \label{fig:hea_circuit}
    \end{subfigure}
     \hfill 
     \begin{subfigure}[b]{0.60\textwidth}
        \centering
        \begin{overpic}[width=\textwidth]{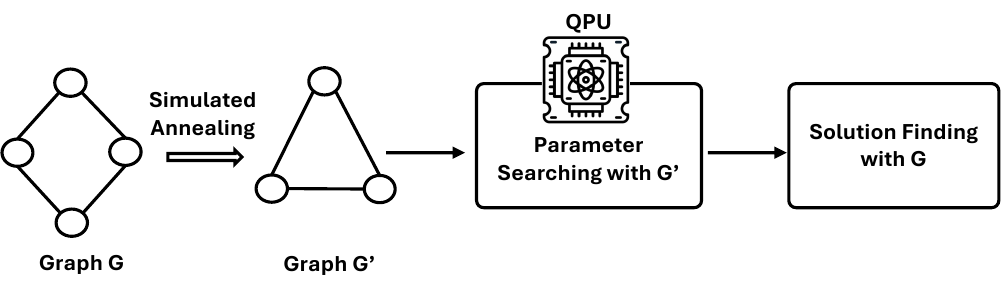}
            \put(2, 80){\textbf{(d)}}
        \end{overpic}
        \label{fig:red_qaoa}
    \end{subfigure}
    \caption{Overview of the Variational Quantum Algorithm. (a) The hybrid classical-quantum feedback loop. (b) Mapping a problem graph to a parameterized quantum circuit. (c) The Hardware Efficient Ansatz~\cite{hea_qaoa} (HEA) adapts to native gate sets to mitigate noise. (d) Graph reduction techniques~\cite{red_qaoa} (Red-QAOA) simplify the problem prior to circuit compilation.}
\label{fig:workflow_vqa}
\end{figure*}

\section{Preliminaries}
\label{sec:preliminaries}
In this section, we outline the fundamental building blocks of variational quantum algorithms and rigorously define the divide-and-conquer strategy. We introduce notation for the Ising Hamiltonian and explicitly derive how qubit freezing induces local magnetic fields, a mechanism central to our landscape similarity hypothesis.

\begin{figure*}[t] 
    \centering 
    \begin{subfigure}[b]{0.95\linewidth}
        \centering
        \begin{overpic}[width=\linewidth]{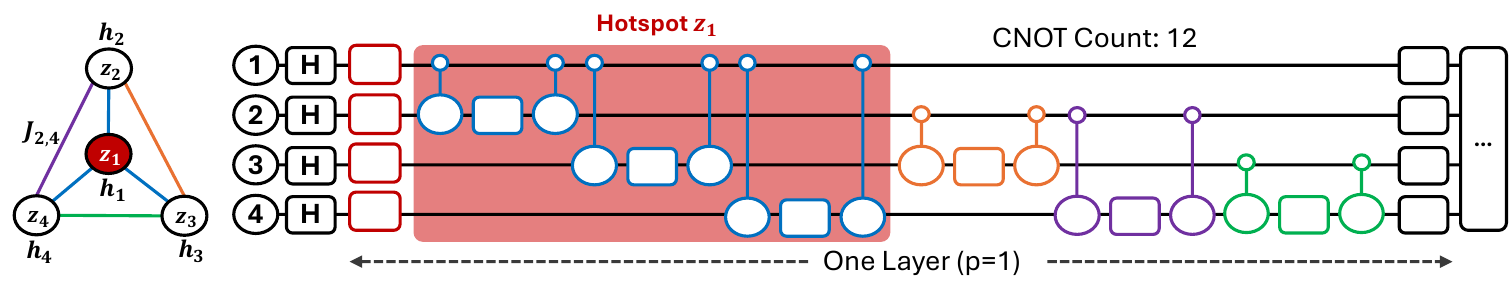}
            \put(-3, 90){\textbf{(a)}}
        \end{overpic}
        \phantomcaption
        \label{fig:identify_hotspot}
    \end{subfigure}
    \begin{subfigure}[b]{0.95\linewidth}
        \centering
        \begin{overpic}[width=\linewidth]{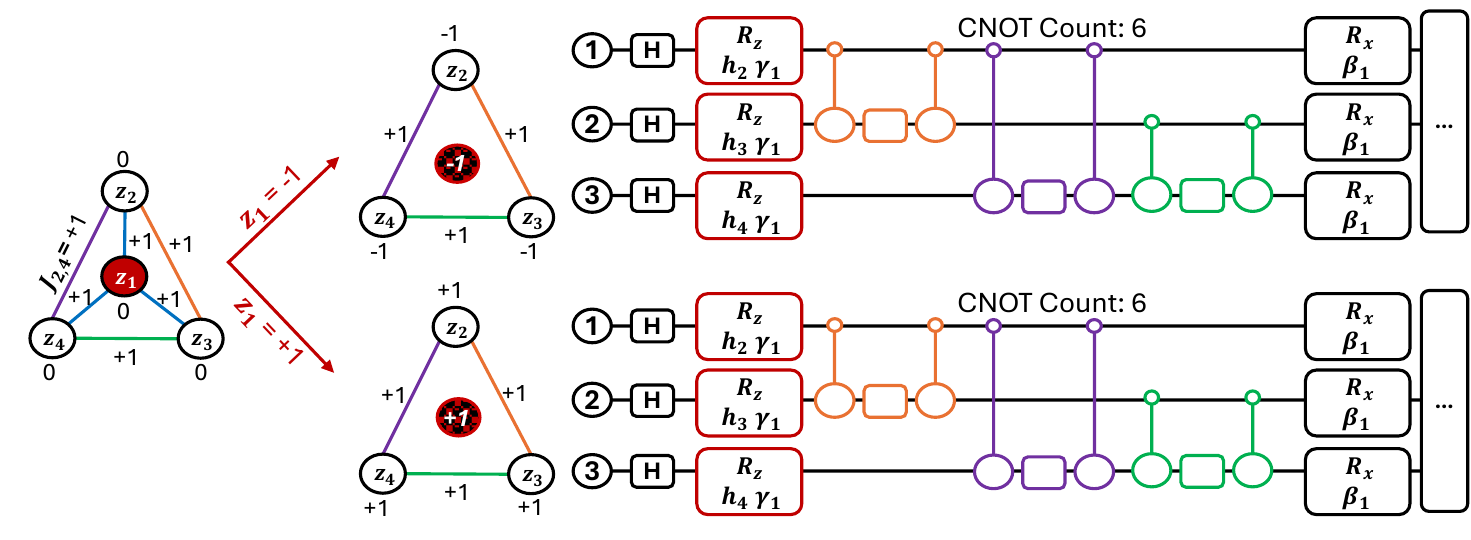}
            \put(-3, 130){\textbf{(b)}}
        \end{overpic}
        \phantomcaption
        \label{fig:divide_conquer}
    \end{subfigure}
    \caption{Workflow of the divide-and-conquer approach for QAOA optimization. (a) High-degree nodes (hotspots, red) are identified. (b) Freezing the hotspot $z_1$ to $+1$ or $-1$ decomposes the interaction graph into independent sub-problems. Note that the choice of $z_1$ alters the linear coefficients ($h_i$) of the remaining nodes, creating distinct energy landscapes.}
    
\label{fig:divide_conquer}
\end{figure*}

\subsection{Variational Quantum Algorithms and QAOA}\label{sub:vqa_qaoa_preliminary}

The Quantum Approximate Optimization Algorithm (QAOA) is designed to find approximate solutions to combinatorial optimization problems defined on a graph $G (V, E)$. The objective is typically encoded in a cost Hamiltonian $H_C$, which is diagonal in the computational basis. 

In its most general form, the problem is represented as an Ising Hamiltonian:
\begin{equation}
    H_{\text{Ising}} = \sum_{i} h_i Z_i + \sum_{i<j} J_{ij} Z_i Z_j + C,
    \label{eq:ising_general}
\end{equation}
where $Z_i$ is the Pauli-Z operator acting on qubit $i$, $h_i$ represents the local magnetic field (bias), $J_{ij}$ denotes the interaction strength (coupling) between qubits $i$ and $j$, and $C$ is a constant energy offset.

For the specific case of the MaxCut problem, we map the graph edges to interactions with $J_{ij} = 0.5$ and $h_i = 0$. The Hamiltonian is thus given by:
\begin{equation}
    H_C = \frac{1}{2} \sum_{(i,j) \in E} (I - Z_i Z_j).
\end{equation}
Solving MaxCut corresponds to finding the ground state that minimizes this energy.

The algorithm proceeds by preparing a parameterized ansatz state $\ket{\psi(\vec{\gamma}, \vec{\beta})}$ by applying alternating layers of the cost unitary $U_C(\gamma) = e^{-i \gamma H_C}$ and a mixing unitary $U_B(\beta) = e^{-i \beta H_B}$, where $H_B = \sum_i X_i$. For a depth-$p$ circuit, the state is:
\begin{equation}
    \ket{\psi(\vec{\gamma}, \vec{\beta})} = \prod_{k=1}^p U_B(\beta_k) U_C(\gamma_k) \ket{+}^{\otimes n}.
\end{equation}
A classical optimizer then iteratively minimizes the expectation value $E(\vec{\gamma}, \vec{\beta}) = \bra{\psi} H_C \ket{\psi}$, as illustrated in Fig.~\ref{fig:workflow_vqa}(a).

\subsection{Noise and Hardware Efficiency}\label{sub:noise_and_hardware_efficient}

In the NISQ era, the depth of quantum circuits is strictly limited by coherence times and gate error rates. As shown in Fig.~\ref{fig:workflow_vqa}(c), Hardware Efficient Ansätze (HEA)~\cite{hea_qaoa} attempt to mitigate this by utilizing native gate sets and minimizing SWAP operations~\cite{nisq_hardware_contrain,sabre,not_all_swap_gate}. However, standard QAOA circuits often require extensive connectivity as shown in Fig.~\ref{fig:workflow_vqa}(b), which does not match the hardware topology, necessitating deep compilation paths that accumulate noise~\cite{sabre,not_all_swap_gate,nisq_hardware_contrain}. This motivates the need for graph reduction (Fig.~\ref{fig:workflow_vqa}(d)) or partitioning strategies to fit significant problems onto noisy, limited-connectivity devices.

\subsection{The Divide-and-Conquer Strategy}\label{sub:divide_and_conquer_preliminary}

To address both connectivity and depth constraints (see Appendix~\ref{sub:noise_and_hardware_efficient}), the \textit{FrozenQubits} approach \cite{FrozenQubits_Ayanzadeh2023} employs a divide-and-conquer strategy based on graph partitioning. 

Let $S \subset V$ be a set of $m$ ``hotspot'' nodes with high degree centrality (Fig.~\ref{fig:divide_conquer}(a)). By ``freezing'' these nodes into classical states $z_k \in \{+1, -1\}$ for all $k \in S$, we remove them from the quantum circuit. This partitioning effectively splits the original graph into smaller, independent sub-graphs that can be solved in parallel on smaller quantum processors (or sequentially on one device) as shown in Fig.~\ref{fig:divide_conquer}(b).

Crucially, as per the general Ising formulation in Eq.~(\ref{eq:ising_general}), freezing a node modifies the Hamiltonian of the remaining active nodes. The interaction between a frozen node $k$ (with value $z_k$) and an active node $j$ transforms from a two-body interaction $J_{kj} Z_k Z_j$ into a one-body linear term:
\begin{equation}
    J_{kj} Z_k Z_j \xrightarrow{Z_k \to z_k} (J_{kj} z_k) Z_j = h_j^{\mathrm{induced}} Z_j.
\end{equation}
Thus, the Hamiltonian for each sub-problem becomes:
\begin{equation}
    H_{\text{sub}} = \sum_{(i,j) \in E'} J_{ij} Z_i Z_j + \sum_{j \in V'} \left( \underbrace{\sum_{k \in S} J_{kj} z_k}_{h_j^{\mathrm{induced}}} \right) Z_j.
    \label{eq:induced_field}
\end{equation}
The term in the parentheses represents the induced linear bias ($h_j$) resulting from the specific configuration of the frozen qubits, as illustrated in Fig.~\ref{fig:divide_conquer}(b). While MaxCut usually assumes $h_i=0$, the divide-and-conquer process \textit{dynamically generates} these linear terms.

The fundamental bottleneck of this approach is the classical overhead. Since there are $m$ frozen qubits, there are $2^m$ possible classical configurations (e.g., $00\dots0$ to $11\dots1$). Existing methods treat these as $2^m$ unrelated optimization problems, requiring $2^m$ independent training loops. As depicted in Fig.~\ref{fig:divide_conquer}, even a small cut of $m=10$ results in 1024 separate training sessions, creating the exponential barrier we aim to break.

\section{Theoretical Bounds on Landscape Similarity}\label{app:theorectical_landscape}

We quantify Hamiltonian differences induced by freezing and distinguish them from differences between the corresponding QAOA landscapes. The latter also depend on the branch-dependent circuit states, the variational angles, and the circuit depth $p$.

\subsection{Hamiltonian Decomposition}
Let the original problem Hamiltonian be $H$. Freezing a subset of qubits $S$ into configuration $z \in \{ \pm 1 \}^{|S|}$ yields a reduced Hamiltonian $H^{(z)}$ acting on the remaining qubits $R$:
\begin{equation}\label{eq:hamiltonian_decompose}
    H^{(z)} = H_{\text{quad}}^{(R)} + H_{\text{lin}}^{(R, z)} + C^{(z)},
\end{equation}
where $H_{\text{quad}}^{(R)}$ contains the quadratic $Z_i Z_j$ terms (invariant across all sub-problems), $H_{\text{lin}}^{(R, z)} = \sum_{i \in R} h_i^{(z)} Z_i$ contains the induced linear fields, and $C^{(z)}$ is a constant energy offset.

\subsection{Branch-dependent landscape bound}


Remove the branch constants and write
\[
\begin{aligned}
A&=H_{\mathrm{quad}}^{(R)}
   +H_{\mathrm{lin}}^{(R,z)},\\
B&=H_{\mathrm{quad}}^{(R)}
   +H_{\mathrm{lin}}^{(R,z')}.
\end{aligned}
\]
Define the shifted branch-dependent landscapes by
$F^{(z)}=E^{(z)}-C^{(z)}$ and
$F^{(z')}=E^{(z')}-C^{(z')}$.
Both circuits use the same initial state, mixer, and
parameter vector, but their cost unitaries use $A$
and $B$, respectively. Set
\[
\begin{aligned}
d&=\|A-B\|
  =\sum_{i\in R}|h_i^{(z)}-h_i^{(z')}|,\\
G&=\sum_{\ell=1}^{p}|\gamma_\ell|,\\
R_B&=\lambda_{\max}(B)-\lambda_{\min}(B).
\end{aligned}
\]

\begin{theorem}[Branch-dependent landscape bound]
\label{thm:stability}
For ideal QAOA circuits under these assumptions,
\begin{equation}
\begin{aligned}
&|F^{(z)}(\boldsymbol{\gamma},\boldsymbol{\beta})
  -F^{(z')}(\boldsymbol{\gamma},\boldsymbol{\beta})|\\
&\qquad\le d+R_B\min\{1,Gd\}.
\end{aligned}
\label{eq:landscape_bound}
\end{equation}
\end{theorem}

\textit{Proof sketch.}
The cost-unitary difference obeys
$\|e^{-i\gamma A}-e^{-i\gamma B}\|\le|\gamma|d$.
Telescoping the circuit layers therefore gives
$\|\psi_A-\psi_B\|_2\le Gd$.
The trace distance of the resulting states satisfies
$D=\tfrac12\|\rho_A-\rho_B\|_1\le\min\{1,Gd\}$.
Splitting the energy difference into
$\operatorname{Tr}[(A-B)\rho_A]$ and
$\operatorname{Tr}[B(\rho_A-\rho_B)]$ bounds the first
contribution by $d$ and the second by $R_BD$,
proving the result.
The bound does not establish preserved basin positions
or a depth-independent landscape deviation.

\paragraph{Heuristic Comparison of Freezing and Interaction Scales.}
To discuss the effect of freezing, we compare a bound on the induced linear perturbation with the spectral range of the remaining quadratic Hamiltonian. This comparison is a diagnostic scale estimate, not a sufficient condition for preserving the QAOA landscape or a universal upper bound on $m$.

First, we define the \emph{quadratic energy scale} $\Delta E_{\mathrm{quad}}$
as the spectral range of $H_{\mathrm{quad}}^{(R)}$ on the active subgraph,
\begin{equation}
\label{eq:DeltaEquad}
\begin{aligned}
\Delta E_{\mathrm{quad}}
\equiv
\lambda_{\max}\!\left(H_{\mathrm{quad}}^{(R)}\right)
- \lambda_{\min}\!\left(H_{\mathrm{quad}}^{(R)}\right) \\
\le 2\sum_{(i,j)\in E_R}|J_{ij}|,
\end{aligned}
\end{equation}
%
The upper bound follows from
$\|H_{\mathrm{quad}}^{(R)}\|
\le\sum_{(i,j)\in E_R}|J_{ij}|$.
Physically, $\Delta E_{\mathrm{quad}}$ sets the
natural energy unit against which linear perturbations are measured: it
is the full vertical extent of the \emph{unperturbed} spin-configuration
landscape $H_{C}(z)$ on $R$. The QAOA expectation value $E(\gamma,\beta)$
in Eq.~\ref{eq:energy_distribution} is a convex combination of these spin energies, so
$\Delta E_{\mathrm{quad}}$ likewise bounds the dynamic range of the
variational landscape generated by the quadratic backbone alone.

Second, we bound the \emph{induced linear perturbation}. Flipping any
single frozen qubit $z_k \to -z_k$ shifts $h_i$ by at most $2|J_{ki}|$
for each active neighbor $i \in \mathcal{N}(k)\cap R$. For two
configurations $z, z'$ differing on all $m$ frozen sites in the worst
case, the cumulative shift is bounded by the cut weight between $S$
and $R$:
%
\begin{equation}
\begin{aligned}
\sum_{i\in R}\left|h_i^{(z)} - h_i^{(z')}\right|
\;\leq\;
2 \sum_{k\in S}\sum_{i\in\mathcal{N}(k)\cap R}|J_{ki}| \\
\;\leq\;
2\,m\,\langle k_S\rangle\,J_{\max},
\label{eq:linear_bound}
\end{aligned}
\end{equation}
where $\langle k_S\rangle$ is the average degree of the frozen nodes
(counting only edges crossing into $R$) and the factor of $2$ arises
from the $\pm 1 \to \mp 1$ swing. This is the cut-size argument made
explicit.

Third, we consider the heuristic scale comparison
\begin{equation}
\underbrace{2\,m\,\langle k_S\rangle\,J_{\max}}_{\text{linear perturbation}}
\;<\;
\underbrace{\Delta E_{\mathrm{quad}}}_{\text{quadratic backbone}}.
\label{eq:similarity_condition}
\end{equation}
%
Equation~\eqref{eq:similarity_condition} compares an upper
bound on the Hamiltonian perturbation with the quadratic
spectral range. It is a diagnostic scale comparison,
not an established necessary or sufficient condition
for correlated QAOA landscapes or successful parameter
transfer. In particular, it does not guarantee unchanged
basins; the actual-landscape bound in
Eq.~\eqref{eq:landscape_bound} also depends on the
circuit angles and depth.

When the denominator is nonzero, the heuristic scale
comparison in Eq.~\eqref{eq:similarity_condition}
can be rearranged as
\begin{equation}
m \;<\; \frac{\Delta E_{\mathrm{quad}}}{2\,\langle k_S\rangle\,J_{\max}}.
\label{eq:m_bound}
\end{equation}
%
Equation~\eqref{eq:m_bound} is a conditional rearrangement
of the heuristic scale comparison, not an independently
established bound on the frozen count. Its right-hand
side depends on the graph and frozen set, and generally
changes with both $m$ and system size. It therefore does
not imply $m=\mathcal O(1)$. For small or sparse active
graphs, the quadratic spectral range can be small or
zero, making this diagnostic restrictive or inapplicable.
Failure of the scale comparison does not itself imply
failure of landscape similarity. The practical choice
of a small frozen set is assessed empirically in the
reported benchmarks.

Landscape similarity does not imply
Eq.~\eqref{eq:similarity_condition}.
For example, freezing one vertex of a zero-field,
unit-coupling $K_4$ graph leaves an active triangle
with $\Delta E_{\mathrm{quad}}=4$, whereas the left-hand
side of Eq.~\eqref{eq:similarity_condition} is 6.
Nevertheless, the two frozen assignments have identical
ideal QAOA landscapes by global spin-flip symmetry.
Thus the scale comparison is not a necessary condition
for landscape similarity.

\subsection{Single-vertex Hamiltonian perturbation}
If two frozen assignments differ only at vertex $k$,
their Hamiltonian difference satisfies
\begin{equation}
d\le
2\sum_{i\in\mathcal N(k)\cap R}|J_{ki}|.
\end{equation}
This controls the Hamiltonian-distance term entering
Eq.~\eqref{eq:landscape_bound}. It does not make the
actual QAOA landscape deviation independent of depth
or guarantee unchanged optimal parameters.

\begin{figure*}[t]
\centering
\includegraphics[width=\linewidth]{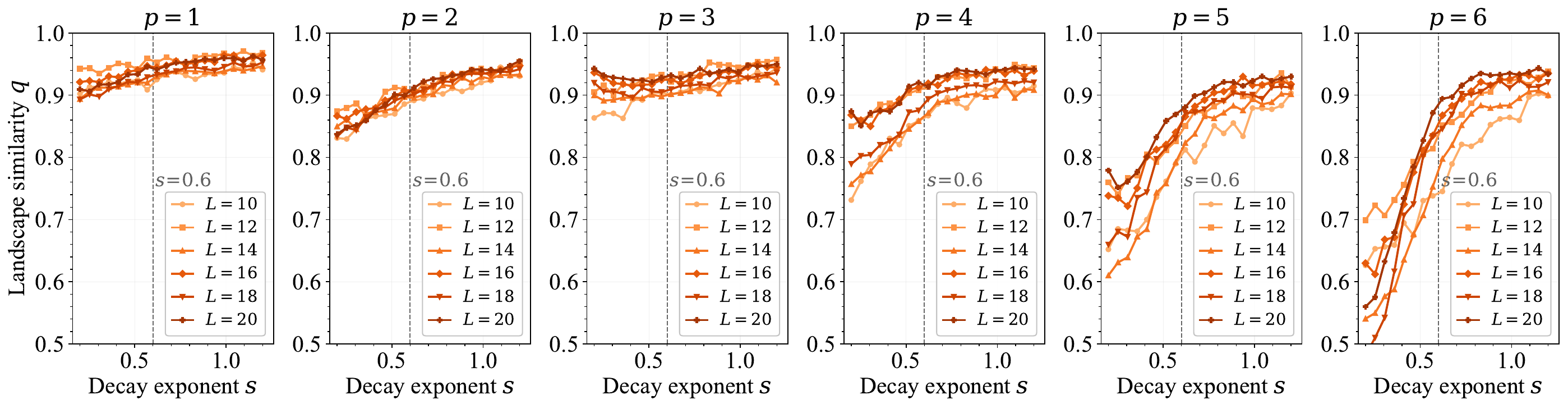}
\caption{Landscape overlap $q$ versus decay exponent $s$ for $p=1,\ldots,6$ and $L=10,12,14,16,18,20$. The line at $s=0.6$ is a reference for comparison with the shallow-depth curves in Fig.~\ref{fig:percolation_phase_boundary_1}. The dependence of overlap on $s$ and $L$ varies with circuit depth; differences between system sizes become more pronounced in several of the larger-depth panels. These observations describe finite systems and do not establish a thermodynamic transition or a universal reference value.}
\label{fig:higher_p}
\end{figure*}

\begin{figure*}[t]
\centering
\includegraphics[width=\linewidth]{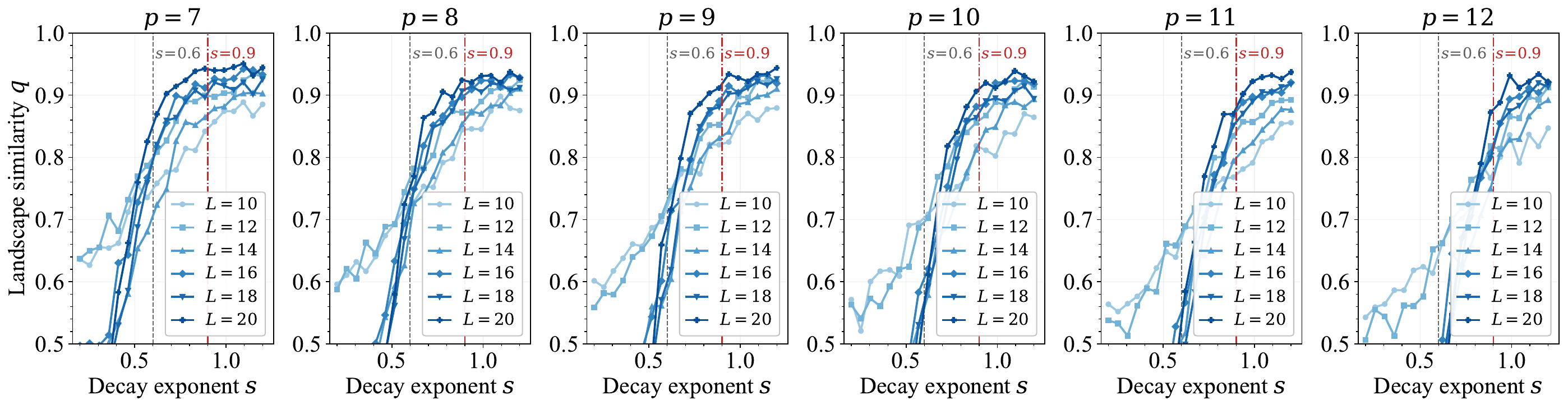}
\caption{Landscape overlap $q$ versus decay exponent $s$ for $p=7,\ldots,12$ and $L=10,12,14,16,18,20$. The region over which overlap increases shifts toward larger $s$ with increasing depth. For $s\gtrsim1.0$, several larger-system curves approach $q\approx0.92$--$0.95$, although residual size dependence remains. The lines at $s=0.6$ and $s=0.9$ are reference markers for comparing the finite-system curves, not estimates of critical points.}
\label{fig:higher_p_7_12}
\end{figure*}

\section{Landscape Similarity at Higher Depths ($p \ge 2$)}
\label{app:higher_depths}

\begin{figure}[h]
    \centering
    \includegraphics[width=\linewidth]{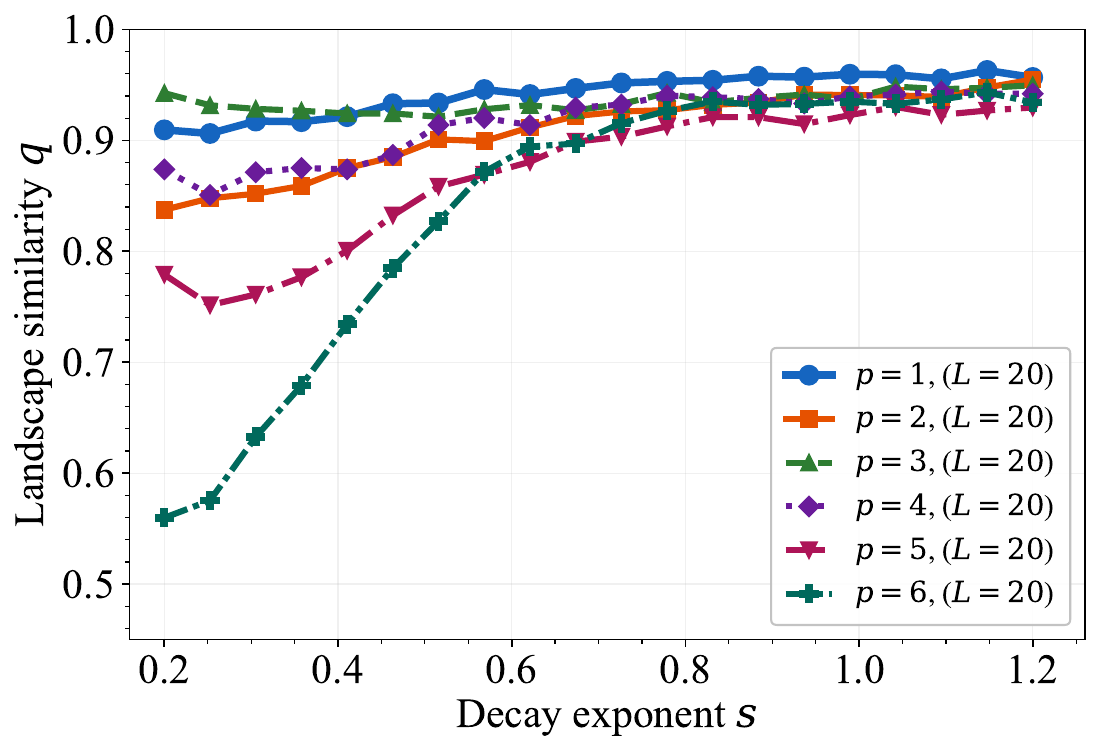}
    \caption{Landscape overlap $q$ versus decay exponent $s$ for $p=1,\ldots,6$ at $L=20$. At larger $s$, the plotted overlaps remain high across the tested depths. The curves report finite-system landscape correlations; they do not establish unchanged basin locations or a thermodynamic interpretation.}
    \label{fig:high_p_overlap}
\end{figure}

\begin{figure}[h]
    \centering
    \includegraphics[width=\linewidth]{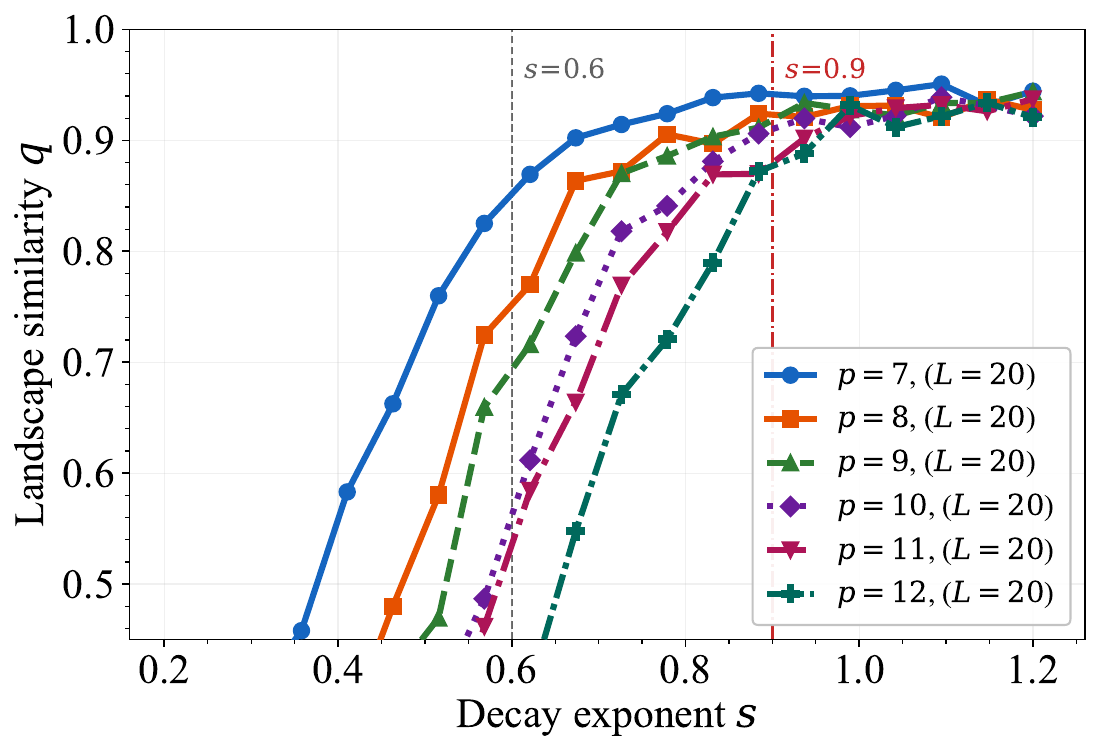}
    \caption{Landscape overlap $q$ versus decay exponent $s$ for $p=7,\ldots,12$ at $L=20$. The region over which $q$ increases moves toward larger $s$ as depth increases. For $s\gtrsim1.0$, the plotted overlaps are approximately $0.92$--$0.95$. The lines at $s=0.6$ and $s=0.9$ are reference markers for comparing the curves, not estimates of critical points.}
    \label{fig:high_p_overlap_7_12}
\end{figure}

\begin{figure*}[t]
    \centering
  \includegraphics[width=\textwidth]{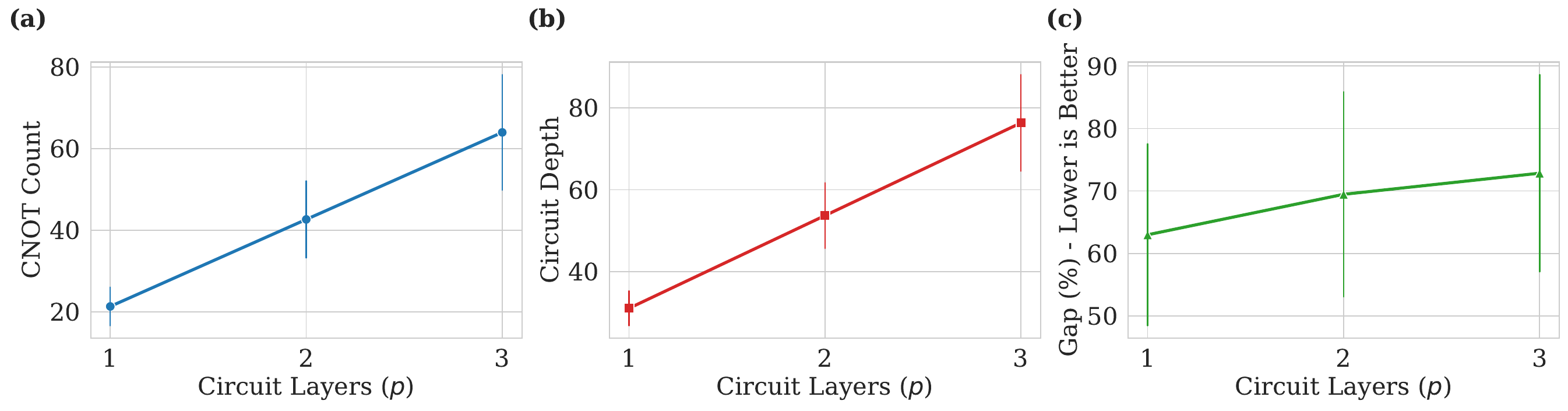}
    \caption{Study on QAOA Depth ($p$). (a) The CNOT count and (b) Circuit depth scale linearly with $p$, significantly increasing the error rate. (c) The results confirm that for current NISQ devices, lower depth ($p=1$) provides a superior trade-off between noise resilience and optimization quality.}
    \label{fig:ablation_layer}
\end{figure*}

To assess the robustness of the Landscape Similarity Hypothesis 
beyond the shallow-circuit regime analyzed in the main text, we 
investigate how the landscape-overlap measure $q(s)$ 
behaves as the QAOA depth $p$ is increased. As $p$ grows, the 
QAOA light cone enlarges and the variational ansatz becomes 
increasingly expressive, potentially amplifying the sensitivity 
of the landscape to the linear perturbations induced by qubit 
freezing.

\subsection{Computational scope and methodology}

In the main text, Fig.~\ref{fig:percolation_phase_boundary_1} reports overlap measurements at $p=1$ for systems up to $L=400$, using the analytic evaluation described in Section~\ref{subsec:landscape_similarity_transition}. At higher depths 
$p \geq 2$, however, the same exact evaluation is no longer 
tractable: deeper QAOA circuits generate increasingly entangled 
many-body states, and faithful classical simulation requires 
storing and evolving the full $2^{n_{\mathrm{active}}}$ amplitude 
vector. The memory and runtime cost therefore grows exponentially 
with the active-system size, restricting full state-vector 
simulation to modest $L$. Also with increase $p$, however, an exhaustive search of the QAOA parameters becomes inefficient due to the curse of dimensionality. If we discretize so that each parameter can take on $m$ values, the exhaustive search of the optimum takes exponential steps in $p$ as $m^{2p}$~\cite{qaoa_fermionic_view_Wang2018}.

Accordingly, we evaluate $q(s)$ for $p \in \{ 1, 2,3,4, 5, 6\}$ on system 
sizes $L \in \{10, 12, 14, 16, 18, 20\}$, with $20$ values of the 
decay exponent $s \in [0.2, 1.2]$ per configuration. 
Crucially, by 
re-evaluating $p = 1$ on the same restricted $L$ range used at 
higher depths, we establish a direct, apples-to-apples 
comparison that controls for finite-size effects and isolates 
the role of circuit depth.
The smaller systems permit comparison of the overlap curves across depths and system sizes, without determining their thermodynamic interpretation.

To probe whether the observed behavior continues at yet larger
depths, we further extend the evaluation to $p \in \{7, 8, 9, 10, 11, 12\}$
on the same system-size range $L \in \{10, \ldots, 20\}$.  At these depths
the exponential cost of state-vector simulation is particularly severe, so
no additional increase in $L$ is attempted.  The results for this extended
range are presented separately in Section~\ref{subsec:p_geq_7} and
Figs.~\ref{fig:higher_p_7_12} and \ref{fig:high_p_overlap_7_12}.

\subsection{Overlap observations at $1\leq p\leq6$}

Figure~\ref{fig:higher_p} shows $q(s)$ for $p=1,\ldots,6$ on the common system-size range $L=10,12,14,16,18,20$. The following observations describe these finite systems.

\paragraph{Overlap at larger $s$.} In the larger-$s$ part of the plotted interval, the overlaps generally approach values near $0.9$--$0.95$, with residual depth and size dependence. High overlap indicates correlated sampled energy values; it does not by itself establish self-averaging, a common landscape class, or successful parameter transfer.

\paragraph{Size dependence at smaller $s$.} At smaller $s$, the curves for different system sizes separate. Their ordering varies with $p$ and $s$, so these data do not establish a uniform direction of size dependence or its thermodynamic limit.

\paragraph{Depth dependence of the curves.} Differences between system sizes become more pronounced in several larger-depth panels. For example, at $p=4$ near $s=0.2$, the spread between the $L=10$ and $L=20$ curves exceeds $0.10$. The reference line at $s=0.6$ helps compare the panels; it does not establish a depth-independent transition location or a value fixed solely by graph geometry.

\subsection{Depth dependence of the overlap curves at $p\geq7$}
\label{subsec:p_geq_7}

Extending the evaluation to $p \in \{7, \ldots, 12\}$ reveals a
qualitatively new feature absent from the $p \leq 6$ data: the
region over which the measured overlap increases moves toward higher values of
$s$ as $p$ increases on these finite systems.
Figs.~\ref{fig:higher_p_7_12} and \ref{fig:high_p_overlap_7_12} document
this behavior in the panel and overlay representations, respectively.  Four
observations are noteworthy.

\paragraph{Shift in the region of increasing overlap.} For $p=7$, the overlap increases near the reference region $s\approx0.6$. As $p$ increases toward $12$, this region moves to larger $s$, reaching approximately $s=0.9$ by $p=12$. The displacement is visible in Fig.~\ref{fig:high_p_overlap_7_12} for $L=20$. These positions describe the plotted curves, not established critical values.

\paragraph{Substantial landscape overlap persists at large $s$.}
For $p=7,\ldots,12$ and $s\gtrsim1.0$, several larger-system curves approach
$q\approx0.92$--$0.95$, including the $L=20$ curves, although residual size dependence remains.
These results demonstrate persistent finite-size landscape
correlation beyond shallow depth. They motivate parameter reuse
but do not, by themselves, establish thermodynamic self-averaging
or guarantee successful parameter transfer.

\paragraph{Possible influence of the finite system size.}
The depth-dependent changes in overlap may be influenced by the interplay between
the QAOA light cone and the finite system size.
For a depth-$p$ circuit, the expectation of a local cost term
depends on graph structure within its circuit light cone
\cite{finite_group_velocity_quantum_spin_system,
qaoa_fermionic_view_Wang2018}.
On the accessible systems ($L\leq20$), increasing $p$ can expose
each cost term to a larger fraction of the active graph,
potentially increasing its sensitivity to the induced fields.
This provides a possible explanation for the observed drift;
the present calculations do not establish its thermodynamic origin.

\paragraph{Interpretation and finite-size limitations.}
The observed changes in the overlap curves may depend on the fraction of the active graph reached by the QAOA light cone. The present calculations do not separate circuit-depth effects from system-size effects or determine the nature of the changes in the thermodynamic limit.
The landscape bound in Eq.~\eqref{eq:landscape_bound} depends on
the circuit angles and depth and does not fix a critical value of $s$.
The results establish substantial landscape correlation in the tested
large-$s$ regime. Resolving the thermodynamic behavior requires further
finite-size analysis beyond the present calculations.

\subsection{Implications for parameter transfer}

At $L=20$, the reported overlaps remain high at larger $s$ for the tested depths. This measures correlation between sampled energy values; it does not isolate the contribution of the quadratic backbone or establish a connectivity threshold for successful transfer. 

To quantify this effect, we compute the pairwise landscape overlap $q$ between sub-problems for varying circuit depths $p$. Figs.~\ref{fig:high_p_overlap} and \ref{fig:high_p_overlap_7_12} summarize the resulting overlap statistics. 
Interestingly, we observe that the landscape overlap remains consistently high even at larger depths for $L=20$. Although shallow circuits ($p=2,3$) exhibit the strongest overlap values, the correlation does not collapse at larger depths. Instead, the optimization landscapes continue to preserve a substantial degree of geometric similarity across sub-problems.

The observed correlations motivate transferring parameters
between related sub-problems beyond shallow depth.
However, the overlap statistic alone does not establish
unchanged basin locations or successful transfer at every depth.
Transfer performance must be assessed separately from the
landscape-overlap measurement.
%
Determining the thermodynamic interpretation of the observed
changes in overlap requires a systematic finite-size analysis at
larger system sizes.
Such calculations become increasingly demanding at higher depth.
Approximate approaches, including tensor-network methods,
could provide a direction for investigating this question
in future work.

The finite-system data show substantial landscape correlation beyond shallow depth. For $p\leq6$, the reference value $s=0.6$ permits comparison of the plotted curves. For $p=7,\ldots,12$, the region over which overlap increases moves toward larger $s$ on systems with $L\leq20$, while at $L=20$ and $s\gtrsim1.0$ the plotted values are approximately $q=0.92$--$0.95$.
These observations do not establish a depth-independent
thermodynamic critical point or a universal locality threshold
for successful parameter transfer. They motivate parameter reuse
in the tested correlated regimes; its practical benefit is
assessed by the reported benchmarks.

\section{Impact of Circuit Depth in NISQ Regime}\label{app:ablation_study_increase_p}

\begin{figure*}[thbp]
  \centering
    \includegraphics[width=\textwidth]{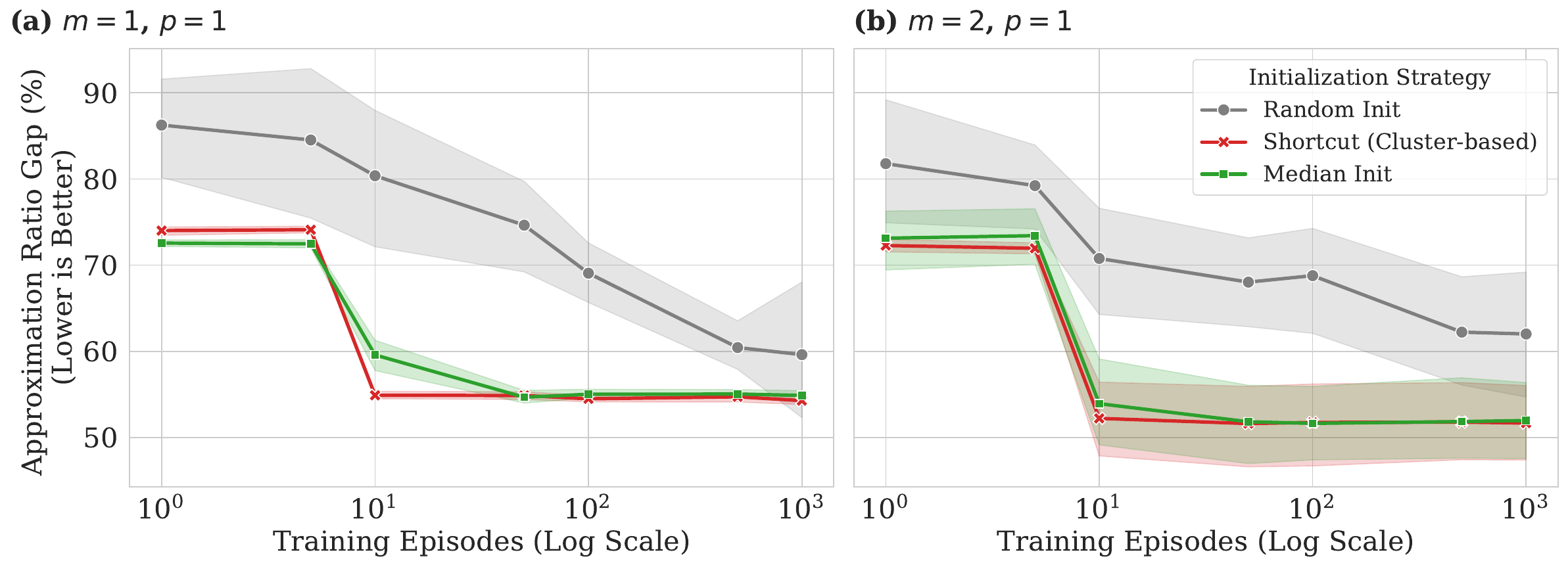}
    \includegraphics[width=\textwidth]{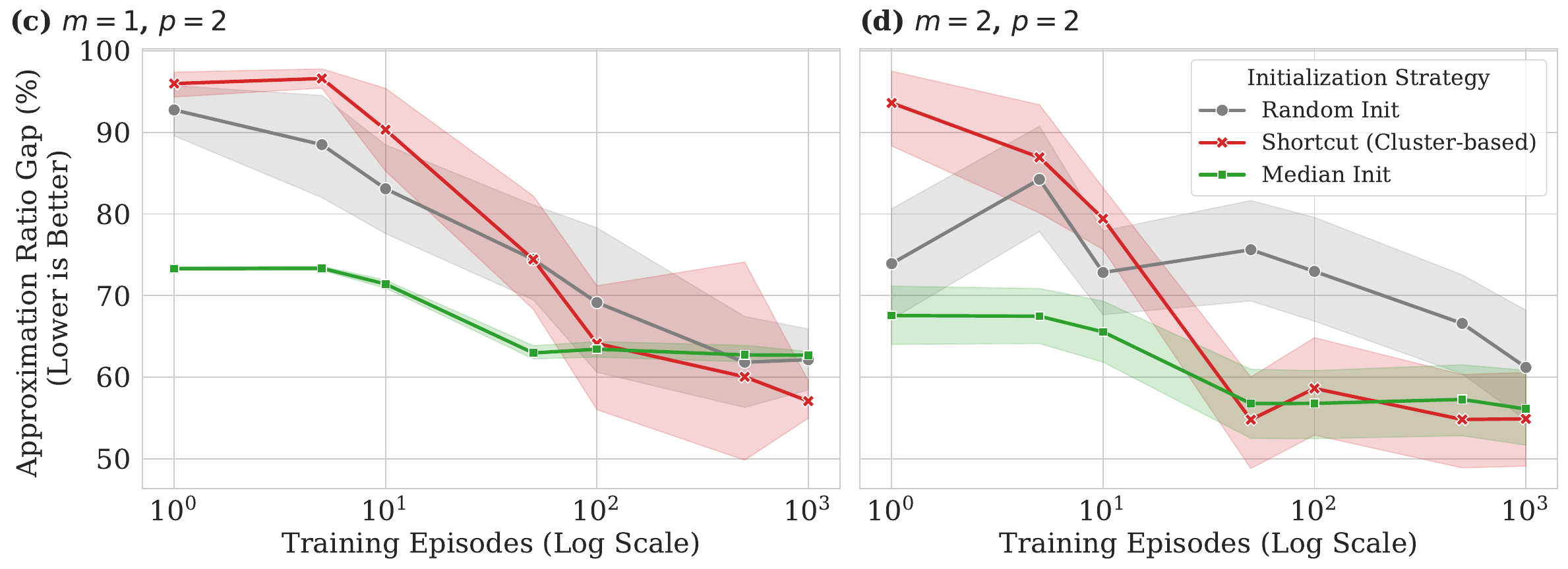}
    \includegraphics[width=\textwidth]{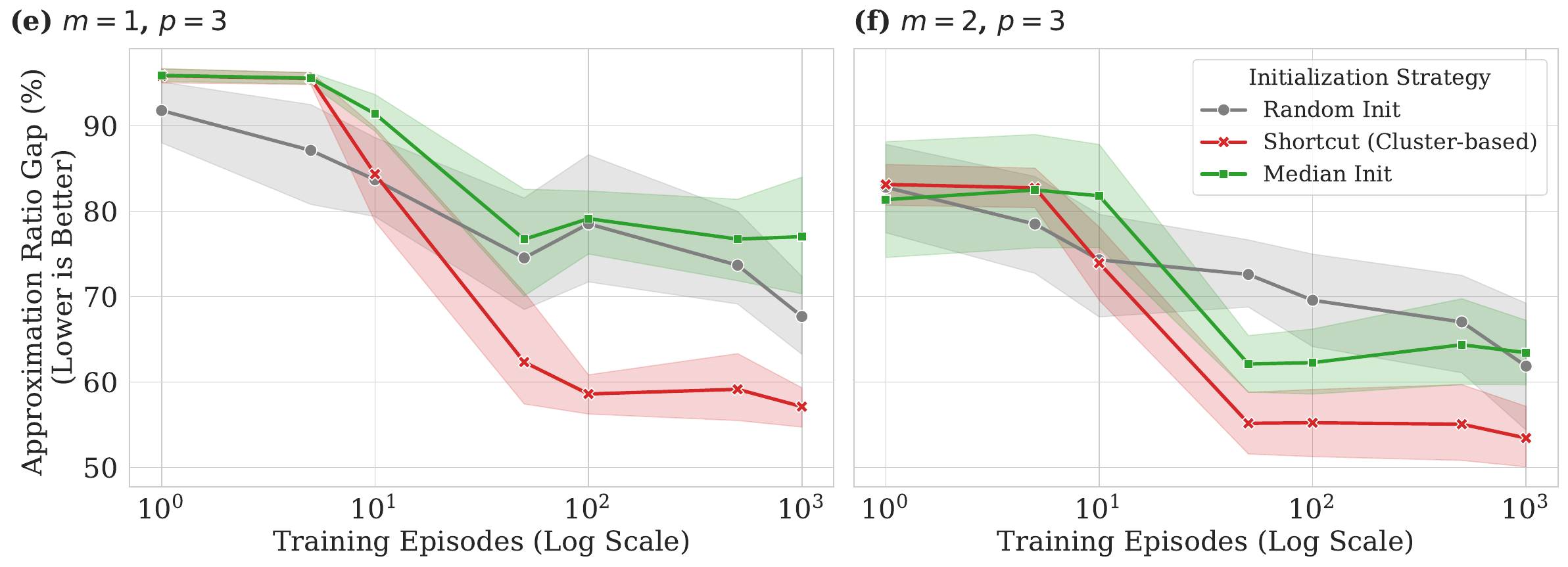}
  \caption{Impact of initialization strategies on optimization convergence across varying problem configurations defined by the number of frozen nodes ($m$) and QAOA depth ($p$). Here, $m$ represents the number of high-degree nodes frozen during the graph partitioning phase, and $p$ denotes the number of QAOA ansatz layers. The results show that the ``Shortcut'' (cluster-based) initialization (red) converges faster and achieves better solution quality (lower ARG) than Median initialization (green) and Random initialization (grey), demonstrating robustness across different partition sizes ($m$) and circuit depths ($p$)}
  \label{fig:convergence_comparison}
\end{figure*}
To empirically validate the decision to focus on the shallow-circuit regime ($p=1$), we conducted an ablation study that analyzes the trade-off between theoretical expressibility and noise-induced degradation as the circuit depth increases.

We evaluated the performance of the QAOA ansatz on a set of random regular graphs with varying circuit depths $p \in \{1, 2, 3\}$. The experiments were conducted using a noise model derived from the \texttt{FakeBrisbane} backend to simulate realistic NISQ conditions~\cite{ibm_fake_provider}.

As illustrated in Figure \ref{fig:ablation_layer}, increasing the number of layers $p$ results in a linear increase in both CNOT count and total circuit depth.

Resource Scaling: For $p=1$, the circuit remains compact with manageable two-qubit gate overhead. However, at $p=3$, the CNOT count triples, and the circuit depth exceeds the coherence time of the simulated qubits for larger graphs.
 
Noise Domination: While theory suggests that $p \to \infty$ allows for adiabatic evolution to the true ground state ~\cite{farhi2014quantumapproximateoptimizationalgorithm, Wurtz_2022}, our noisy simulations show a divergence. The ARG does not improve monotonically with $p$; instead, the noise accumulation from the additional gates outweighs the benefits of increased ansatz expressibility.

The study confirms that finding the lowest energy state becomes increasingly difficult at greater depths because the ``noise floor'' rises faster than the theoretical overlap improves. Consequently, our DO-QAOA methodology and the associated theoretical bounds are preferred optimized for and analyzed in the $p=1$ regime to maximize fidelity on near-term hardware.

\begin{figure*}[t]
    \centering
    \includegraphics[width=\textwidth]{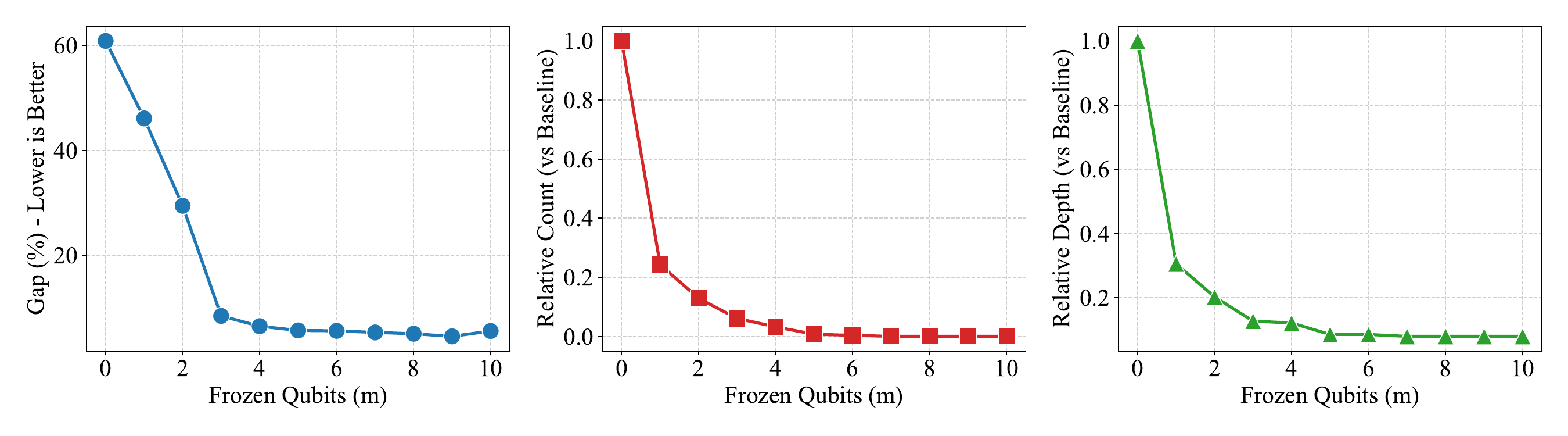} 
    \caption{Analysis of Diminishing Returns ($N=15$). The ARG improves rapidly for small $m$ but saturates beyond $m > 3$, indicating diminishing returns. While quantum resources (CNOTs and Depth) continue to decrease linearly, the sub-problems eventually become trivial, providing no further optimization benefit.}
    \label{fig:diminishing_returns_panel}
\end{figure*}

\section{Impact of Initialization Strategy on Optimization Convergence}
\label{app:initialization}

In the representative optimization step (Step 3) in Fig.~\ref{fig:method_overview}, the choice of initial variational parameters $\bm{\theta}_\mathrm{init}$ significantly influences the convergence speed and the final quality of the solution for the representative subcircuit $G_\mathrm{rep}^{(m)}$. Standard optimization approaches often employ random initialization, which can lead to getting stuck in local minima or requiring extensive training epochs to converge, especially in the rugged landscapes characteristic of NISQ problems.

However, our landscape analysis reveals that the optimal parameters for our problem instances are not uniformly distributed but rather cluster around specific regions in the energy landscape. To exploit this, we compare three distinct initialization strategies:
(1) Random Initialization: Parameters are sampled uniformly from the domain $[0, 2\pi]$. This serves as the baseline.
(2) Median Initialization: This method utilizes the heuristic described in ~\cite{qaoa_kit, ParameterTransfer_Shaydulin2023, fixed_angle_PhysRevA}, which initializes parameters based on the median of known optimal values for similar graph classes for Max-cut problems.
(3) Shortcut Initialization (Ours): Based on our empirical observations, we identify a cluster of optimal parameters near specific angles. Specifically, we initialize parameters near $\gamma_1 \approx -\pi/6$ and $\beta_1 \approx -\pi/8$. Note that these values are specific to the Ising Hamiltonian formulation used in our study, which differs from the standard Max-Cut objective definition where angles might appear shifted or scaled.
Fig.~\ref{fig:convergence_comparison} illustrates the impact of these strategies on the convergence of the Approximation Ratio Gap (ARG). The results demonstrate that searching for optimal parameters within a cluster around these specific angles allows the optimizer to converge significantly faster and more accurately than with random initialization. The ``Shortcut'' method consistently achieves a lower ARG in fewer training episodes, validating the benefit of physics-informed initialization for the representative subcircuit.

Additionally, our analysis highlights the trade-off between circuit depth and noise resilience. While theoretically increasing the number of layers ($p$) improves the ansatz's expressibility, our results indicate that $p=1$ remains the most effective choice for NISQ implementations. As observed in Fig.~\ref{fig:convergence_comparison}, deeper circuits ($p=2, p=3$) do not yield significant improvements in the final ARG compared to $p=1$. Instead, the increased gate count and circuit depth introduce additional noise channels that can degrade solution quality and hinder convergence. Thus, maintaining a shallow depth ($p=1$) strikes a critical balance, ensuring robust performance by minimizing noise accumulation while sufficiently capturing the problem's energy landscape.

\section{Practical Limits and Diminishing Returns of Qubit Freezing}
\label{sec:appendix_limits}

\begin{figure}[t]
    \centering
    \includegraphics[width=\linewidth]{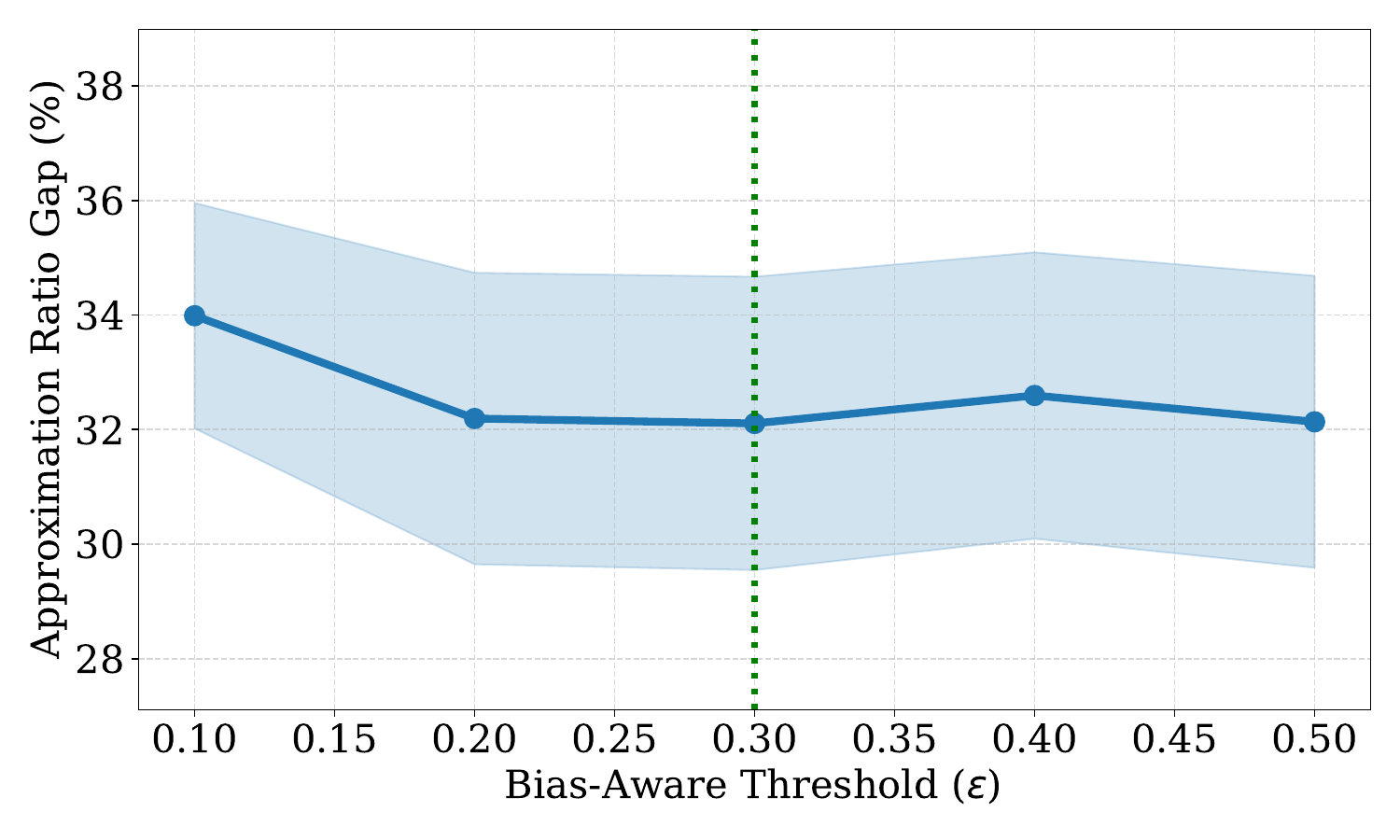}
    \caption{Sensitivity Analysis of Transfer Threshold ($\epsilon$). Evaluated on 10 random Power-Law graphs ($N=15, m=3$). The Approximation Ratio Gap (blue) remains stable across the range $[0.1, 0.5]$. The green dotted line indicates the chosen default ($\epsilon=0.3$)}
    \label{fig:sensitivity_analysis}
\end{figure}

To rigorously define the operational limits of the Divide-and-Conquer (D\&C) strategy, we extended our experimental analysis beyond the standard regime ($m \le 3$) to aggressive decimation levels up to $m=10$. This study was conducted on Power-Law graphs with $N=15$ nodes, involving over 20,000 circuit evaluations. 

As illustrated in Fig.~\ref{fig:diminishing_returns_panel}, the ARG exhibits a distinct saturation behavior. 
Rapid Improvement ($m \le 3$): For small values of $m$, the ARG drops significantly (from baseline($m=0$) $\approx 60\%$ to $\approx 45\%$ at $m=1$ and $\approx 9\%$ at $m=3$). In this regime, freezing high-degree "hotspot" nodes effectively reduces circuit depth, directly mitigating noise-induced errors.
Diminishing Returns ($m \ge 4$): Beyond $m=4$, the improvement in solution quality plateaus. The marginal gain in accuracy for each additional frozen qubit becomes negligible. This indicates that once the dominant topological complexity is resolved, further fragmentation of the graph yields sub-problems that are already sufficiently simple for the optimizer, rendering additional cuts redundant.

\section{Sensitivity Analysis of the Bias-Aware Transfer Rule Threshold}
\label{sec:appendix_sensitivity}

The Bias-Aware Transfer Rule is governed by a critical hyperparameter, the bias distortion threshold $\epsilon$, which determines whether to perform a direct parameter transfer or trigger a warm-start optimization. To justify our default choice of $\epsilon = 0.3$ and assess the algorithm's robustness, we conducted a sensitivity analysis on Power-Law graphs ($N=15$) at a decimation depth of $m=3$. We varied $\epsilon$ from $0.1$ to $0.5$ across 10 distinct graph instances to ensure statistical significance.

Fig.~\ref{fig:sensitivity_analysis} presents the aggregated results for ARG and Total Quantum Shots. Stability of Accuracy: The average ARG remains remarkably stable ($\approx 32\% \pm 12\%$) across the entire sweep range. Varying $\epsilon$ does not lead to significant fluctuations in solution quality, indicating that the algorithm effectively identifies the correct transfer strategy regardless of minor threshold adjustments.
The stability observed in Fig.~\ref{fig:sensitivity_analysis} provides strong empirical support for the Landscape Similarity Hypothesis. Even at $m=3$, the landscape geometries of the sub-problems remain sufficiently correlated.

We selected $\epsilon = 0.3$ (marked by the green dotted line) as the default operating point because it lies in the center of this stable regime. This choice balances the risk of: (1) Over-sensitivity ($\epsilon < 0.2$): Being too strict might trigger unnecessary warm-start loops for sub-problems that are actually similar enough for direct transfer, potentially increasing runtime without improving accuracy. (2) Under-sensitivity ($\epsilon > 0.4$): Being too loose might blindly transfer parameters to sub-problems with significant landscape distortions, potentially degrading the approximation ratio.
Thus, $\epsilon = 0.3$ serves as a robust, distinct-from-zero heuristic that generalizes well across different graph topologies without requiring instance-specific fine-tuning.

\section{Additional Experimental Details}\label{app:additional_compar}
\subsection{Experimental Setup}\label{app:exp_setup}

\textbf{Benchmarks.}
We evaluated our method on over 6{,}000 circuits across a diverse set of graph benchmarks to assess robustness on synthetic and real-world problem instances. Synthetic benchmarks include power-law graphs generated using the Barabási-Albert preferential attachment model, as well as random regular graphs and instances of the Sherrington-Kirkpatrick (SK) model~\cite {albert2005scale, barabasi1999emergence, barabasi2000scale, super_blockers, sparsity_aware, wang2019complex, zadorozhnyi2012structural, zbinden2020embedding}. 

\begin{figure*}[!htbp]
    \centering
    \begin{subfigure}[b]{0.9\textwidth}
        \centering
        \includegraphics[width=\textwidth]{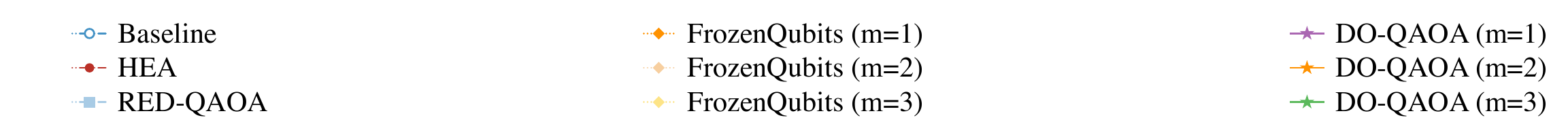}
    \end{subfigure}

    \vspace{1em} 

    \begin{subfigure}[b]{0.48\textwidth}
        \centering
        \begin{overpic}[width=\textwidth]{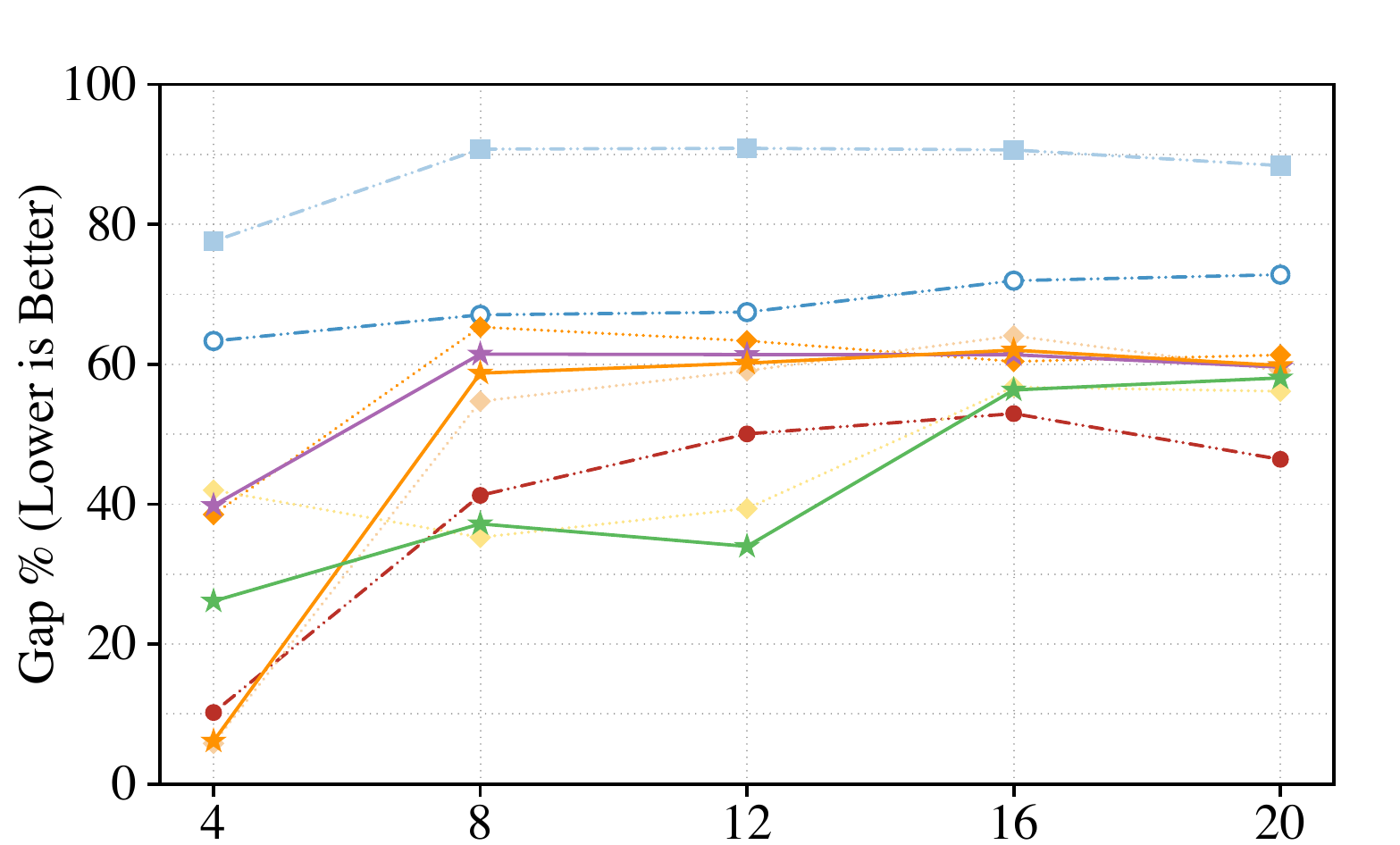}
            \put(0, 140){\textbf{(a)} 3-regular}
        \end{overpic}
        \label{fig:3_regular_results_arg}
    \end{subfigure}
    \hfill
    \begin{subfigure}[b]{0.48\textwidth}
        \centering
        \begin{overpic}[width=\textwidth]{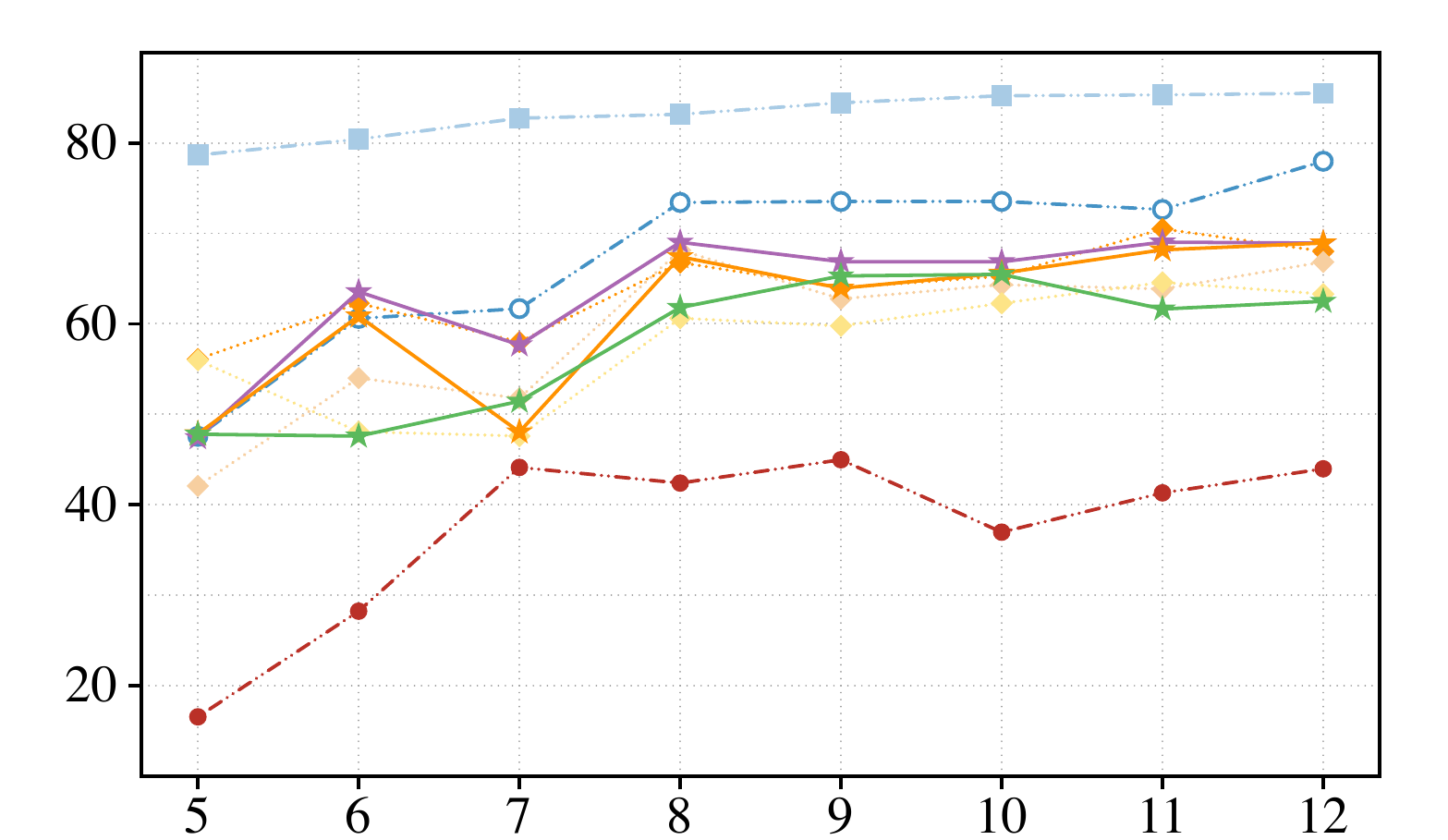}
            \put(0, 140){\textbf{(b)} SK}
        \end{overpic}
        \label{fig:sk_results_arg}
    \end{subfigure}

    \vspace{1em} 

    \begin{subfigure}[b]{0.48\textwidth}
        \centering
        \begin{overpic}[width=\textwidth]{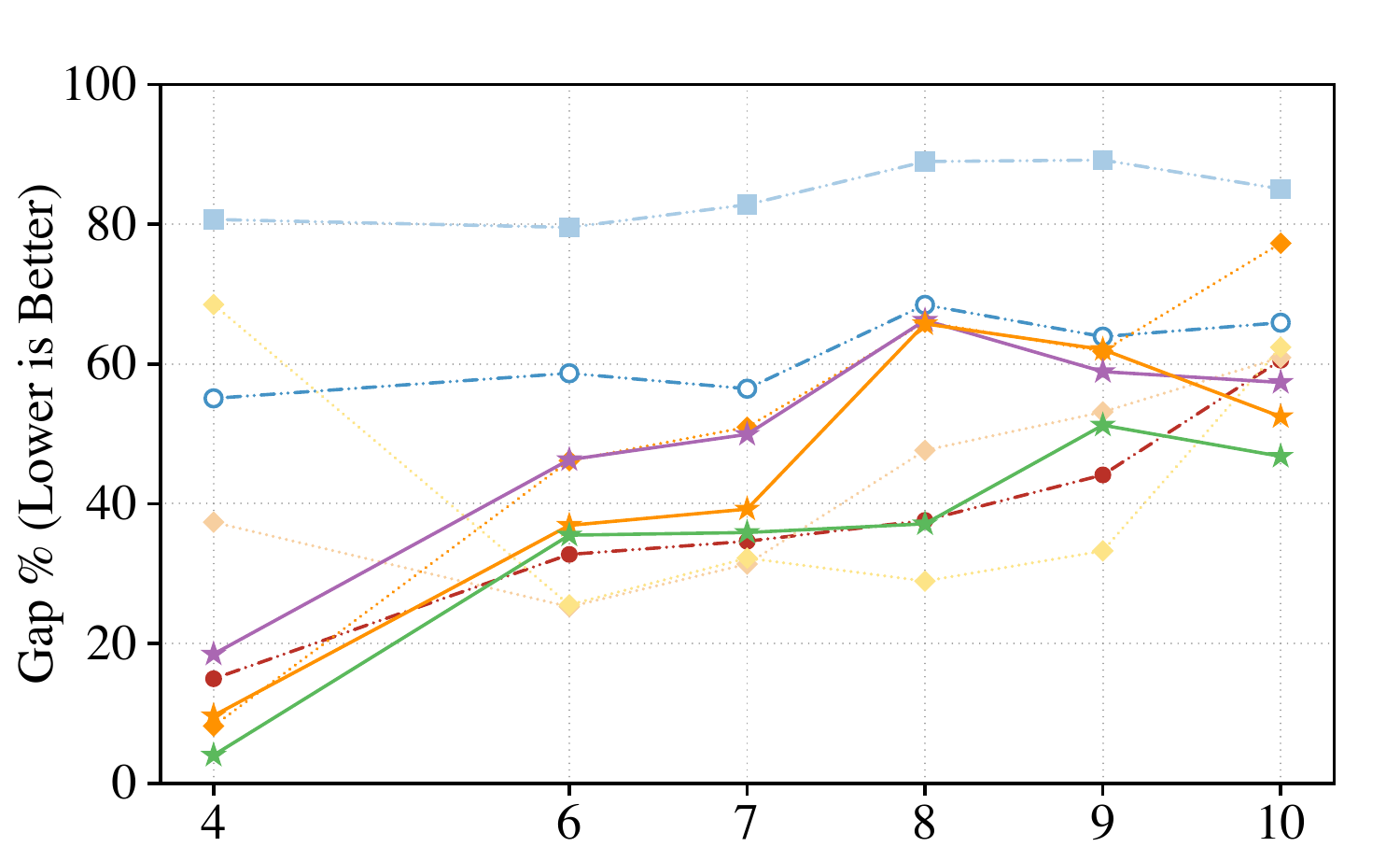}
            \put(0, 140){\textbf{(c)} AIDS}
        \end{overpic}
        \label{fig:aids_results_arg}
    \end{subfigure}
    \hfill
    \begin{subfigure}[b]{0.48\textwidth}
        \centering
        \begin{overpic}[width=\textwidth]{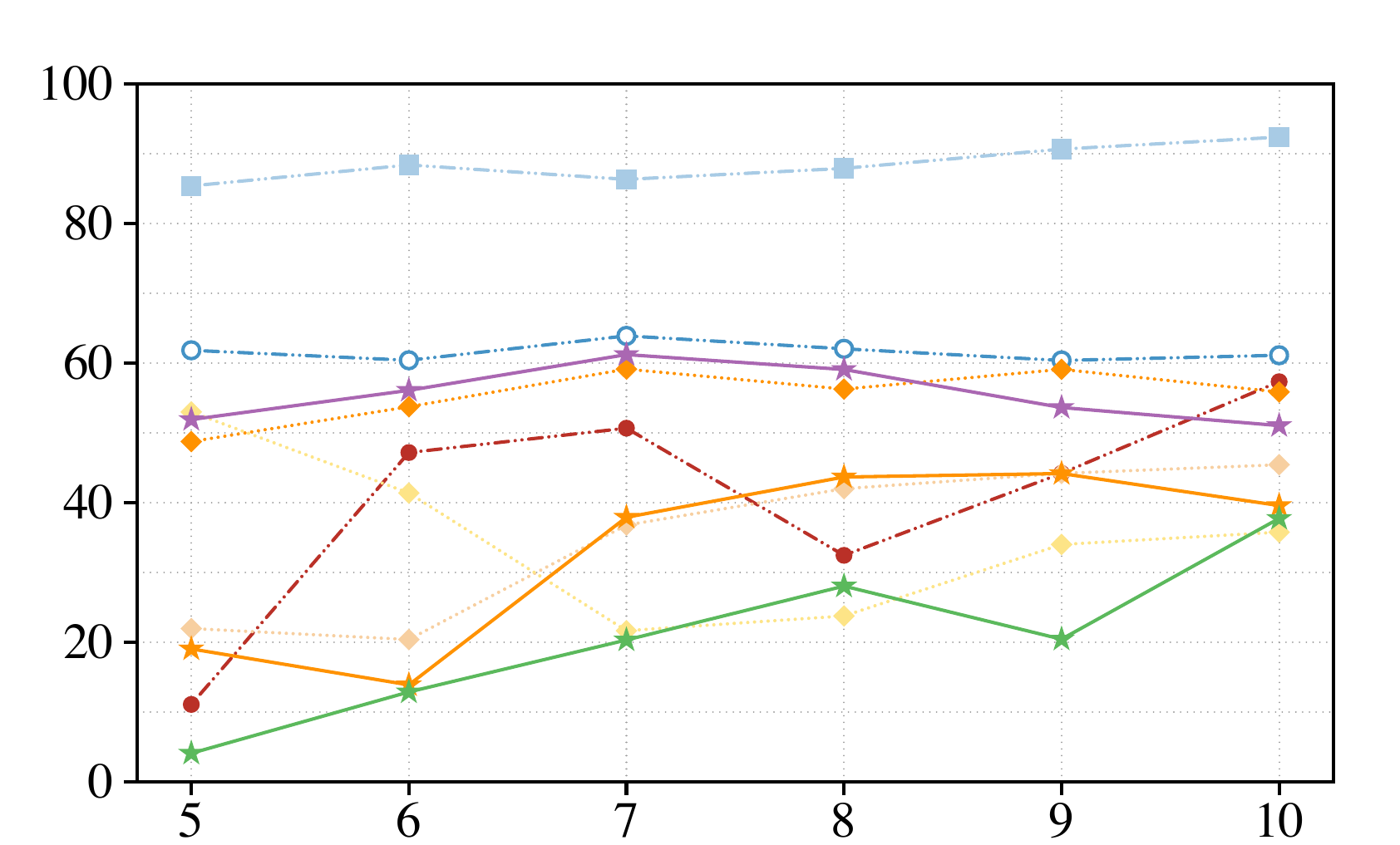}
            \put(0, 140){\textbf{(d)} Linux}
        \end{overpic}
        \label{fig:linux_results_arg}
    \end{subfigure}

    \vspace{1.2em} 

    \begin{subfigure}[b]{0.48\textwidth}
        \centering
        \begin{overpic}[width=\textwidth]{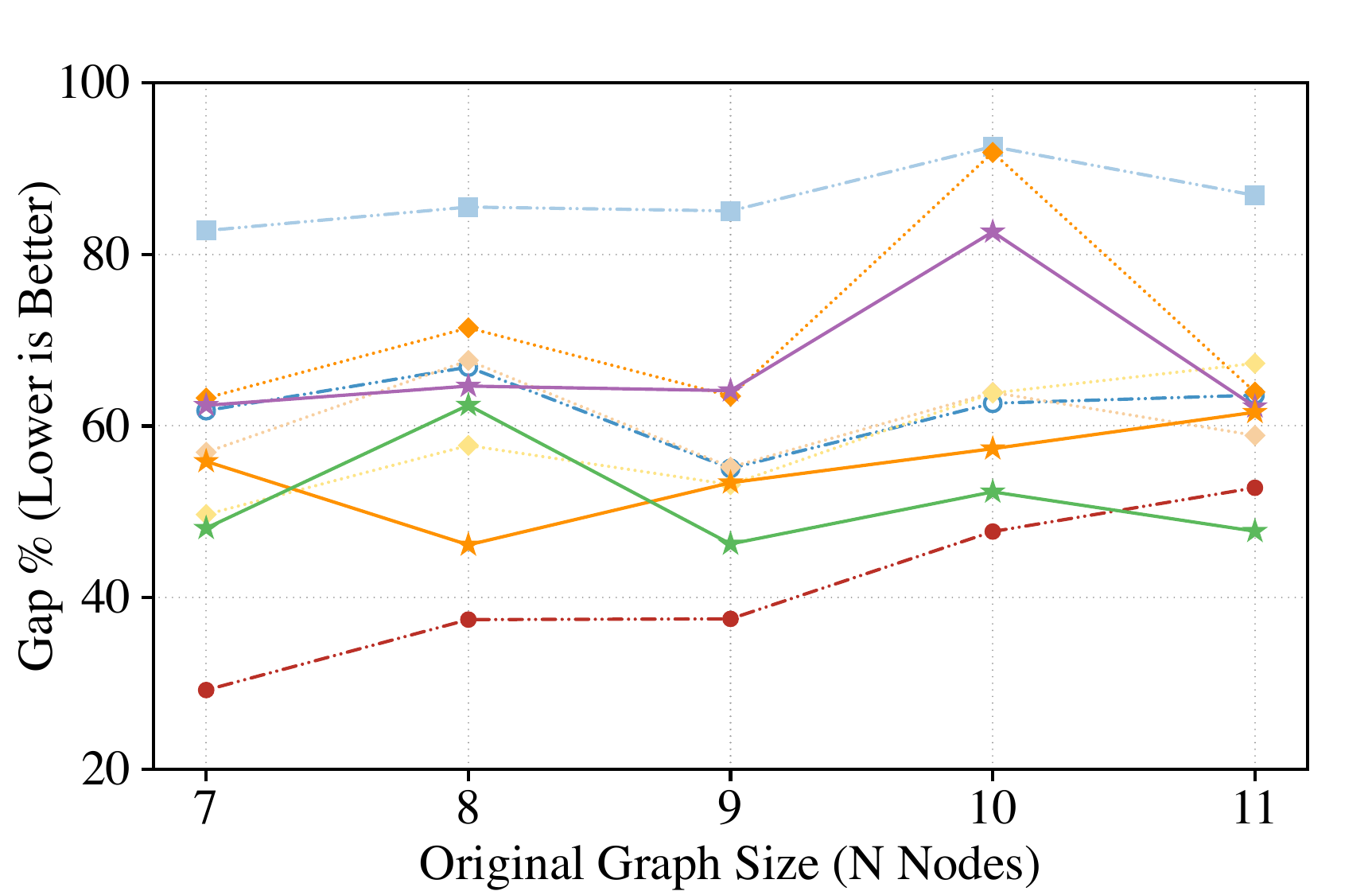}
            \put(0, 155){\textbf{(e)} IMDb}
        \end{overpic}
        \label{fig:imdb_results_arg}
    \end{subfigure}

    \caption{ARG analysis on Regular (3-regular and SK) and Real-world Graph Datasets (AIDS, Linux, and IMDb). Lower ARG values indicate closer proximity to the optimal solution. At $m=3$ under the reported $p=1$ noisy benchmark conditions, DO-QAOA has equal or lower dataset-averaged ARG than FrozenQubits. At $m=2$, the reported ARG is slightly higher for Regular, SK, and AIDS graphs (Table~\ref{tab:additional_comparison}).}
    \label{fig:additional_result}
\end{figure*}

\begin{table}[t]
    \centering
    \caption{Summary of graph benchmarks used in evaluation.}
    \label{tab:benchmarks}
    \begin{tabular}{l l c c}
        \hline
        \hline
        Graph Type & \# Instances & Avg. Nodes \\
        \hline
        Power-law  & $\sim$2{,}000 & 4-20 \\
        Regular graphs        & $\sim$2{,}000 & 4-20 \\
        SK model           & $\sim$2{,}000 & 5-12 \\
        \hline
        
        AIDS       & 700            & 4-10  \\
        Linux          & 1{,}000        & 5-10 \\
        IMDb     & 1{,}500        & 7-11  \\
        \hline
        \hline
    \end{tabular}
\end{table}

To further validate performance on practical workloads, we incorporated real-world graph datasets spanning multiple domains. These include 700 molecular graphs from the AIDS dataset provided by the National Cancer Institute (average graph size of 8 nodes)~\cite{aids_riesen2008iam}, 1{,}000 function call graphs extracted from the Linux kernel (average size of 10 nodes)~\cite{linux_graph}, and 1{,}500 collaboration graphs derived from the IMDb movie database (average size of 6 nodes)~\cite{imdb}.
This diverse dataset selection allows us to evaluate our methods across various domains, offering insights into their scalability and adaptability to different graph structures. Table~\ref{tab:benchmarks} summarizes the characteristics of benchmark graph datasets used in our experiments.

\begin{table*}[!htbp]
\centering
\caption{Summary of the comparison between DO-QAOA and state-of-the-art approaches across metrics and datasets from Table~\ref{tab:benchmarks}, averaged over all node instances, for the main $p=1$ benchmarks with $N\leq20$ under the \texttt{FakeBrisbane} noise model. The Total Shots metric is reported in units of $10^6$ (millions).}
\label{tab:additional_comparison}
\begin{tabular}{llcccccc}
\toprule
\toprule
Metric & Method & Power-Law & Regular & SK & AIDS & Linux & IMDb \\ 
\midrule

\multirow{9}{*}{ARG} 
& Baseline & 61 & 68 & 67 & 61 & 61 & 61 \\
& HEA & 45 & 40 & 37 & 37 & 40 & 40 \\
& Red-QAOA & 90 & 87 & 83 & 84 & 88 & 86 \\
& FrozenQubits ($m$=1) & 52 & 57 & 63 & 51 & 55 & 70 \\
& FrozenQubits ($m$=2) & 44 & 48 & 59 & 42 & 35 & 60 \\
& FrozenQubits ($m$=3) & 43 & 45 & 57 & 41 & 34 & 58 \\
& DO-QAOA ($m$=1) & 52 & 56 & 63 & 49 & 55 & 67 \\
& DO-QAOA ($m$=2) & 37 & 49 & 61 & 44 & 33 & 54 \\
& DO-QAOA ($m$=3) & 26 & 42 & 57 & 35 & 20 & 51 \\ 
\midrule

\multirow{9}{*}{\shortstack{Total Shots\\($\times 10^6$)}} 
& Baseline & 8.19 & 8.19 & 8.19 & 8.19 & 8.19 & 8.19 \\
& HEA & 8.19 & 8.19 & 8.19 & 8.19 & 8.19 & 8.19 \\
& Red-QAOA & 8.19 & 8.19 & 8.19 & 8.19 & 8.19 & 8.19 \\
& FrozenQubits ($m$=1) & 16.38 & 16.38 & 16.38 & 16.38 & 16.38 & 16.38 \\
& FrozenQubits ($m$=2) & 32.77 & 32.77 & 32.77 & 32.77 & 32.77 & 32.77 \\
& FrozenQubits ($m$=3) & 65.54 & 65.54 & 65.54 & 65.54 & 65.54 & 65.54 \\
& DO-QAOA ($m$=1) & 0.09 & 0.09 & 0.09 & 0.09 & 0.09 & 0.09 \\
& DO-QAOA ($m$=2) & 0.13 & 0.11 & 0.19 & 0.16 & 0.18 & 0.19 \\
& DO-QAOA ($m$=3) & 0.17 & 0.23 & 0.23 & 0.21 & 0.20 & 0.23 \\ 
\midrule

\multirow{9}{*}{Time (s)} 
& Baseline & 298 & 149 & 140 & 47 & 46 & 53 \\
& HEA & 6,094 & 2,051 & 552 & 360 & 364 & 434 \\
& Red-QAOA & 50 & 37 & 55 & 45 & 46 & 46 \\
& FrozenQubits ($m$=1) & 442 & 158 & 126 & 90 & 88 & 87 \\
& FrozenQubits ($m$=2) & 721 & 213 & 235 & 196 & 186 & 184 \\
& FrozenQubits ($m$=3) & 1,055 & 557 & 434 & 382 & 367 & 358 \\
& DO-QAOA ($m$=1) & 73 & 25 & 27 & 15 & 15 & 21 \\
& DO-QAOA ($m$=2) & 80 & 20 & 41 & 16 & 26 & 25 \\
& DO-QAOA ($m$=3) & 101 & 82 & 67 & 16 & 35 & 33 \\ 
\midrule

\multirow{9}{*}{CNOTs} 
& Baseline & 44 & 72 & 138 & 28 & 27 & 106 \\
& HEA & 12 & 12 & 8 & 7 & 6 & 9 \\
& Red-QAOA & 8 & 19 & 47 & 7 & 7 & 36 \\
& FrozenQubits ($m$=1) & 11 & 30 & 54 & 7 & 7 & 37 \\
& FrozenQubits ($m$=2) & 6 & 24 & 41 & 4 & 2 & 27 \\
& FrozenQubits ($m$=3) & 3 & 7 & 30 & 1 & 0 & 19 \\
& DO-QAOA ($m$=1) & 11 & 30 & 54 & 7 & 7 & 37 \\
& DO-QAOA ($m$=2) & 6 & 24 & 41 & 4 & 2 & 27 \\
& DO-QAOA ($m$=3) & 3 & 7 & 30 & 1 & 0 & 19 \\ 
\midrule

\multirow{9}{*}{Depths} 
& Baseline & 42 & 53 & 87 & 34 & 30 & 122 \\
& HEA & 16 & 16 & 12 & 11 & 11 & 13 \\
& Red-QAOA & 13 & 19 & 36 & 12 & 11 & 44 \\
& FrozenQubits ($m$=1) & 13 & 25 & 40 & 13 & 12 & 31 \\
& FrozenQubits ($m$=2) & 9 & 20 & 33 & 9 & 7 & 25 \\
& FrozenQubits ($m$=3) & 6 & 11 & 28 & 6 & 4 & 21 \\
& DO-QAOA ($m$=1) & 13 & 25 & 40 & 13 & 12 & 31 \\
& DO-QAOA ($m$=2) & 9 & 20 & 33 & 9 & 7 & 25 \\
& DO-QAOA ($m$=3) & 6 & 11 & 28 & 6 & 4 & 21 \\ 
\bottomrule
\bottomrule
\end{tabular}%
\end{table*}

\textbf{Noise Model.}
All experiments were conducted under realistic noisy conditions using a device-level noise model derived from the \texttt{FakeBrisbane} backend provided by IBM Quantum. This noise model captures key NISQ characteristics, including gate infidelities, readout errors, and decoherence effects, allowing us to simulate hardware-constrained execution faithfully. The noise parameters were obtained through the IBM fake provider and runtime infrastructure, ensuring consistency with contemporary superconducting quantum devices~\cite{design_and_architecture_of_ibm_quantum_engine, qiskit_backend_specifications}.

\textbf{Baselines.}
We compared DO-QAOA against four representative baselines: standard QAOA, the Hardware-Efficient Ansatz (HEA)~\cite{hea_qaoa}, Red-QAOA~\cite{red_qaoa}, and FrozenQubits~\cite{FrozenQubits_Ayanzadeh2023}. To compile each circuit, we used a popular quantum software stack to ensure a fair comparison. In particular, each quantum circuit was compiled using IBM’s Qiskit toolchain with transpilation optimization level~2, which prioritizes reduced compilation time while preserving circuit structure. All methods were evaluated under identical noise models and execution settings. We also used the optimization level 2 to balance compilation overhead and circuit fidelity, avoiding aggressive gate rewriting that may alter the effective circuit depth~\cite{qiskit_compiler, comprehensive_review_of_quantum_circuit_optimization}.

\textbf{Hardware Environment.}
All experiments were conducted on a high-performance workstation equipped with an AMD Ryzen Threadripper PRO 5955WX 16-Core CPU, 512~GB of RAM, and three NVIDIA RTX~4090 GPUs. This configuration enabled large-scale parallelization of the classical simulation and optimization components required by the experiments.

\subsection{Per-Dataset Benchmark Results}\label{app:per_dataset_analysis}

\textbf{DO-QAOA on Regular Graphs.}
We study DO-QAOA on 3-regular graphs, and fully connected graphs (or SK model) provide a complementary benchmark characterized by uniform degree distributions and reduced structural asymmetry~\cite{regular_sk_graphs, harrigan2021, solve_spin_glass}. As shown in Table~\ref{tab:additional_comparison}, Fig.~\ref{fig:additional_result}(a) and Fig.~\ref{fig:additional_result}(b), this homogeneity moderates both the benefits and limitations of divide-and-conquer strategies.

Baseline QAOA again exhibits significant degradation under noise (ARG $\approx 68\%$ on 3-regular, and $\approx 67\%$ on SK graphs), As shown in Table~\ref{tab:additional_comparison}, baseline QAOA has the highest CNOT counts and circuit depths, while HEA yields moderate improvement (ARG $\approx 40\%$ on 3-regular, and $\approx 37\%$ on SK graphs), but their execution increase by one to two orders of magnitude compared to DO-QAOA. Red-QAOA remains suboptimal (ARG $\approx 87\%$ and $\approx 83\%$ ), consistent with the loss of essential connectivity information under aggressive reduction.

FrozenQubits improves approximation quality with increasing $m$, achieving an ARG of $45\%$ on 3-regular, and $57\%$ on SK graphs at $m=3$. DO-QAOA matches and slightly improves upon this trend, further reducing the ARG to $42\%$ on 3-regular, and $57\%$ on SK graphs at $m=3$. While the absolute accuracy gain over FrozenQubits is smaller than in power-law graphs, the computational savings remain substantial. At $p=1$, $m=3$, under the \texttt{FakeBrisbane} noise model, total shots decrease from $65.54\times10^6$ to $0.23\times10^6$ for both datasets. The reported runtime decreases from $557$~s to $82$~s for Regular graphs and from $434$~s to $67$~s for SK (Table~\ref{tab:additional_comparison}).

These observations indicate that even in the absence of pronounced hub structures, the Landscape Similarity Hypothesis remains valid. 
These results are interpreted in terms of parameter reuse; the number of distinct landscape classes is not determined in this study.

\textbf{DO-QAOA on Real-world Graph Datasets.}
We further evaluate DO-QAOA on real-world graphs drawn from diverse application domains, including molecular interaction networks (AIDS), software call graphs (Linux), and collaboration networks (IMDb)~\cite{aids_riesen2008iam, linux_graph, imdb}, as shown in Table~\ref{tab:additional_comparison}, Fig.~\ref{fig:additional_result}(c), Fig.~\ref{fig:additional_result}(d) and Fig.~\ref{fig:additional_result}(e). These datasets exhibit heterogeneous sizes, densities, and noise sensitivities, providing a stringent test of robustness.

Across all real-world benchmarks, baseline QAOA consistently yields large approximation gaps (ARG $\approx 61\%$), while HEA improves stability but plateaus between $37\%$ and $40\%$. Red-QAOA again underperforms, confirming that aggressive reduction strategies are ill-suited for structured, domain-specific graphs.

FrozenQubits shows gradual improvement with increasing $m$, particularly on Linux graphs where freezing high-degree nodes is effective. However, the exponential growth in total shots (up to $65.5\times10^{6}$) and execution time (approximately $380$~s) severely limits scalability.

For Linux at $p=1$, $m=3$, under the \texttt{FakeBrisbane} noise model, DO-QAOA has an ARG of $20\%$, compared with $34\%$ for FrozenQubits; total shots decrease from $65.54\times10^6$ to $0.20\times10^6$, and runtime decreases from $367$~s to $35$~s. For AIDS and IMDb at the same $p$ and $m$, the reported ARG changes from $41\%$ to $35\%$ and from $58\%$ to $51\%$, respectively, with total shots decreasing from $65.54\times10^6$ to $0.21\times10^6$ and $0.23\times10^6$ (Table~\ref{tab:additional_comparison}).

At the same frozen count $m$, DO-QAOA and FrozenQubits
have identical reported CNOT counts and circuit depths
(Table~\ref{tab:additional_comparison}).
Their resource differences therefore reflect the evaluated
optimization and measurement workflows rather than
reductions in these circuit metrics. The aggregate results
do not isolate the contribution of parameter transfer alone.

%

\end{document}